\documentclass{article}
\usepackage[utf8]{inputenc}
\usepackage{amsmath}
\usepackage{amsthm}
\usepackage{microtype}
\usepackage[shortcuts,nospacearound]{extdash}    % Better endling of em dash (---) to cut sentences\===like 
\usepackage{thmtools}
\declaretheorem[name=Theorem,numberwithin=section]{theorem}
\declaretheorem[name=Definition,sibling=theorem]{definition}
\declaretheorem[name=Lemma,sibling=theorem]{lemma}

\declaretheorem[name=Remark,sibling=theorem]{remark}
\usepackage{amssymb}
\usepackage{graphicx}
\usepackage{physics}
\usepackage{todonotes}
\usepackage{pgf,tikz,pgfplots}
\usepackage[english]{babel}
\usepackage[margin=1.1in]{geometry}
\usepackage{authblk}
\pgfplotsset{compat=1.16}
\usepackage{tabularx}
\usepackage{siunitx}
\usepackage{caption}
\usepackage{subcaption}
\usetikzlibrary{fit}
\NewDocumentCommand{\matrixLl}{}{{\kern.2mm\scalebox{.05}{\begin{tikzpicture}[baseline={(-1,-1)}]
  \draw (-1,-1) rectangle (1,1);
  \fill (-1,-1) -- (-1,1) -- (1,-1) -- cycle; 
\end{tikzpicture}}}}
\NewDocumentCommand{\matrixUr}{}{{\kern.2mm\scalebox{.05}{\begin{tikzpicture}[baseline={(-1,-1)}]
  \draw (-1,-1) rectangle (1,1);
  \fill (1,1) -- (-1,1) -- (1,-1) -- cycle; 
\end{tikzpicture}}}}

\usepackage{needspace}

\usepackage[createShortEnv,commandRef=cref,conf={text link section}]{proof-at-the-end}

\usepackage{stmaryrd} % Double brackets \llbracket\rrbracket [[ ]] https://tex.stackexchange.com/questions/107252

\usepackage{multicol}
\newcommand*{\vect}[1]{\ensuremath{\mathbf{#1}}}

\usepackage{float}
\floatstyle{ruled}
\newfloat{Protocol}{htbp}{lop}[section]
\NewDocumentEnvironment{protocol}{m}{\Protocol}{\caption{#1}\endProtocol}
\newfloat{Resource}{htbp}{lop}[section]
\NewDocumentEnvironment{resource}{m}{\Resource\caption{#1}}{\endResource}

\usepackage{zx-calculus}
\usepackage{robust-externalize}

\robExtConfigure{
  new preset={templateZX}{
    latex,
    add to latex options={11pt},
    add to preamble={
      \usepackage{tikz}
      \usepackage{tikz-cd}
      \usepackage{amsmath}
      \newcommand*{\xor}{\oplus}
      \usepackage{mathtools}
      \newcommand*{\eqdef}{\coloneqq}
      \newcommand{\Sim}{\ensuremath{\mathbf{\sigma}}}
      \makeatletter
      \input{__ROBEXT_WAY_BACK__tikzlibraryzx-calculus.code.tex} %% .. is needed because of the subfolder
      \makeatother
      \tikzset{
        myBold/.style={
          /tikz/commutative
          diagrams/arrows={B},
          /zx/user overlay nodes/.style={
            zxSpiders/.append style={BBw},
            zxH/.append style={BBw},
          }
        }
      }
      
    },
    dependencies={tikzlibraryzx-calculus.code.tex},
  },
  new preset={circuit}{
    add to preamble={
      \usepackage{calrsfs}
      \usepackage{braket}
      \DeclareMathAlphabet{\pazocal}{OMS}{zplm}{m}{n}
    },
  },
  if unix={
    compile in parallel after=5,
  },
}

\robExtConfigure{
  new compiled preset={compiled ZX}{templateZX, compile latex template}{},
  rename backup files for arxiv,
}

\cacheEnvironment{ZX}{compiled ZX}
\cacheCommand{zx}{compiled ZX}

\tikzset{
  myBold/.style={
    /tikz/commutative
    diagrams/arrows={B},
    /zx/user overlay nodes/.style={
      zxSpiders/.append style={BBw},
      zxH/.append style={BBw},
    }
  }
}

\makeatletter
\def\fps@figure{htbp}
\def\fps@table{htbp}
\makeatother

\makeatother
\usepackage{xparse}

\let\originalmiddle=\middle
\def\middle#1{\mathrel{}\originalmiddle#1\mathrel{}}
\NewDocumentCommand{\pr}{som}{\Pr\IfBooleanT#1{\IfValueT{#2}{_{\substack{#2}}}}\left[\,#3\IfBooleanF#1{\IfValueT{#2}{\bigm\vert #2}}\,\right]}
\DeclareMathOperator*{\espsymbol}{\mathbb{E}}
\NewDocumentCommand{\esp}{som}{\espsymbol\IfBooleanT#1{\IfValueT{#2}{_{\substack{#2}}}}\left[\,#3\IfBooleanF#1{\IfValueT{#2}{\bigm\vert #2}}\,\right]}
\NewDocumentEnvironment{pral}{soO{}}
 {%
  \Pr\IfValueT{#2}{_{\IfBooleanTF{#1}{\substack{#2}}{#2}}}
\begin{alignedat}[t]{2}
  [\, } {\,]
 #3
 \end{alignedat}
 }
\NewDocumentEnvironment{espal}{soO{}}
 {%
   \IfValueTF{#2}{\underset{\substack{#2}}{\mathbb{E}}}{\mathbb{E}}
   \begin{alignedat}[t]{2}[\,
 }
 {\,]#3\end{alignedat}
}
\makeatother
 
\usepackage{complexity} % get \QMA{}...
\usepackage{calrsfs} % Replace \mathcal with different more "cursive" font
\DeclareMathAlphabet{\pazocal}{OMS}{zplm}{m}{n} % \pazocal is now the old \mathcal
\newcommand*{\cA}{\pazocal{A}}

\newcommand*{\cC}{\pazocal{C}}
\newcommand*{\cD}{\pazocal{D}}

\newcommand*{\cM}{\pazocal{M}}

\renewcommand*{\cP}{\pazocal{P}} % defined in complexity

\newcommand*{\cR}{\pazocal{R}}
\renewcommand*{\cS}{\pazocal{S}} % defined in complexity
\newcommand*{\cT}{\pazocal{T}}

\newcommand*{\cV}{\pazocal{V}}

\newcommand{\Sim}{\ensuremath{\mathbf{\sigma}}}
\newcommand*\GHZ{\ensuremath{}{\sf GHZ}}
\usepackage{mathtools}
\newcommand*{\eqdef}{\coloneqq}
\newcommand*{\xor}{\oplus}

\newcommand*{\source}{\ensuremath{\mathtt{S}}}
\DeclareMathOperator{\TD}{TD}
\newcommand*{\Merge}{\mathsf{Merge}}
\newcommand*{\CheckState}{\mathsf{Check}}

\newcommand*{\eps}[0]{\varepsilon}

\renewcommand*{\R}[0]{\mathbb{R}}

\newcommand*{\Z}[0]{\mathbb{Z}}

\newcommand*{\gateI}{\vect{I}}
\newcommand*{\gateX}{\vect{X}}
\newcommand*{\gateZ}{\vect{Z}}
\newcommand*{\gateH}{\vect{H}}
\newcommand*{\gateRZ}{\vect{R}_z}
\newcommand*{\gateRX}{\vect{R}_x}

\newcommand*{\gateCNOT}{\vect{CNOT}}
\newcommand*{\gateCZ}{\vect{\wedge{}Z}}

\newcommand*{\eqAboveWidth}[2][=]{\mathrel{\stackrel{\substack{#2}}{#1}}}

\newcommand*{\eqAbove}[2][=]{\eqAboveWidth[#1]{#2}}
\newcommand*{\eqeq}[2][=]{\eqAbove[#1]{(\ref{#2})}}
\newcommand*{\eqSubeq}[3][=]{\eqAboveWidth[#1]{(\hyperref[#2]{\ref{#2}.#3})}}
\newcommand*{\eqZXandSubeq}[4][=]{\eqAboveWidth[#1]{\myZXRule{#2}, (\hyperref[#2]{\ref{#3}.#4})}}
\newcommand*{\eqeqDef}[2][=]{\eqAbove[#1]{\ref{#2}}}
\newcommand*{\eqeqLem}[2][=]{\eqAbove[#1]{\ref{#2}}}

\newcommand*{\myZXRule}[1]{\hyperlink{eq:zxRule#1}{#1}}
\newcommand*{\eqZX}[1]{\eqAbove{\myZXRule{#1}}}
\newcommand*{\eqZXX}[2]{\eqAbove{\myZXRule{#1}, \myZXRule{#2}}}
\newcommand*{\eqZXXX}[3]{\eqAbove{\myZXRule{#1}, \myZXRule{#2}, \myZXRule{#3}}}

\newcommand{\addEditor}[2]{%
  \expandafter\newcommand\csname #1AddText\endcsname[2][]{%
    \ifdefined\removeCommentsAcceptPropositions %
      ##2%
    \else%
      \textcolor{#2}{\sout{##1}}\textcolor{#2!60!black}{##2}%
    \fi%
  }%
  \expandafter\newcommand\csname #1Comment\endcsname[1]{%
    \ifdefined\removeCommentsAcceptPropositions %
    \else%
      \textcolor{#2!60!white}{\textbf{##1}}%
    \fi%
  }%
  \expandafter\newcommand\csname #1Remove\endcsname[1]{%
    \ifdefined\removeCommentsAcceptPropositions %
    \else%
      \textcolor{#2}{\sout{##1}}%
    \fi%
  }%
}
\definecolor{darkgreen}{rgb}{0., 0.7, 0}
\addEditor{leo}{darkgreen}
\addEditor{raja}{red}
\addEditor{damian}{blue}

\usepackage[style=trad-alpha,
  sortcites=true,
  doi=false,
  url=false,
  giveninits=true, % Bob Foo --> B. Foo
  isbn=false,
  url=false,
  eprint=false,
  sortcites=false, % \cite{B,A,C}: [A,B,C] --> [B,A,C]
]{biblatex}
\defbibfilter{appendixOnlyFilter}{
  segment=1 % Segment 1 will be chosen to be the one in appendix
  and not segment=0 % Default segment is 0
}
\makeatletter
\protected\def\abx@missing#1{\textbf{??}}
\makeatother
\newcommand{\filter}{\ensuremath{\bot}}

\NewDocumentCommand{\Verif}{O{\ket{G}}}{\ensuremath{\cV_{#1}}}
\newcommand{\CoinFlip}{\ensuremath{\cR_{\mathsf{CoinFlip}}}}
\newcommand{\Concrete}{$\pi_{[n]}\cR\pi_S$}

\definecolor{secondaryColor}{RGB}{206,149,0} %% darker orange, looks like gold. <3
\usepackage[colorlinks, allcolors=secondaryColor]{hyperref}
\usepackage[capitalise,noabbrev]{cleveref}
\crefname{equation}{Eq.}{Eqs.}
\crefname{figure}{Fig.}{Figs.}
\Crefname{equation}{Equation}{Equations}
\Crefname{figure}{Figure}{Figures}
\Crefname{Protocol}{Protocol}{Protocols}
\Crefname{Resource}{Resource}{Resource}

\usepackage [
n,
lambda,
advantage,
operators,
sets,
adversary,
landau,
probability,
notions,
logic,
ff,
mm,
primitives,
events,
complexity,
asymptotics,
keys
] {cryptocode}

\makeatletter
\renewcommand{\sample}{\mathrel{\mathpalette\sample@\relax}}
\newcommand{\sample@}[2]{%
  \ooalign{%
    \hspace{\stretch{2}}\raisebox{0.7\height}{$\m@th\demotestyle{#1}\mathdollar$}\hspace{\stretch{1}}\cr
    $\m@th#1\leftarrow$\cr
  }%
}
\newcommand{\demotestyle}[1]{%
  \ifx#1\displaystyle\scriptstyle\else
  \ifx#1\textstyle\scriptstyle\else
  \scriptscriptstyle\fi\fi
}
\makeatother
\title{All graph state verification protocols are composably secure}
\usepackage{authblk}
\usepackage{orcidlink}
\usepackage{fontawesome5}

\author[1]{Léo Colisson\,\orcidlink{0000-0001-8963-4656}}
\author[2]{Damian Markham}
\author[3]{Raja Yehia\, \orcidlink{0000-0002-7843-7398}}
\affil[1]{QuSoft and Centrum Wiskunde \& Informatica, Science Park 123, 1098 XG Amsterdam, Netherlands}
\affil[2]{Sorbonne Université, CNRS, LIP6, 4 place Jussieu, F-75005 Paris, France}
\affil[3]{ICFO-Institut de Ciencies Fotoniques, The Barcelona Institute of Science and Technology, Castelldefels, Spain}

\begin{document}

\maketitle

\begin{abstract}
  Graph state verification protocols allow multiple parties to share a graph state while checking that the state is honestly prepared, even in the presence of malicious parties. Since graph states are the starting point of numerous quantum protocols, it is crucial to ensure that graph state verification protocols can safely be composed with other protocols, this property being known as \emph{composable security}. Previous works~\cite{YDK21_ComposableSecurityMultipartite} conjectured that such a property could not be proven within the abstract cryptography framework: we disprove this conjecture by showing that \emph{all} graph state verification protocols can be turned into a composably secure protocol with respect to the natural functionality for graph state preparation. Moreover, we show that any \emph{unchanged} graph state verification protocols can also be considered as composably secure for a slightly different, yet useful, functionality. Finally, we show that these two results are optimal, in the sense that any such generic result, considering arbitrary black-box protocols, must either modify the protocol or consider a different functionality.

  Along the way, we show a protocol to generalize entanglement swapping to arbitrary graph states that might be of independent interest.
\end{abstract}

\newpage
\tableofcontents
\newpage

\section{Introduction}

Quantum networks enhance today's networks capabilities by providing a higher level of security, based on the inviolable laws of physics, but also by enabling the emergence of new protocols impossible to obtain classically. The spectrum of quantum protocols is wide, starting from quantum teleportation~\cite{BBC+93_TeleportingUnknownQuantum} to delegated computation~\cite{BFK09_UniversalBlindQuantum}, verifiable computation~\cite{FK17_UnconditionallyVerifiableBlind,MF18_PostHocVerification,GKK19_VerificationQuantumComputation}, multi-party computation~\cite{DNS12_ActivelySecureTwoParty,MultipartyQC,DGJ+20_SecureMultipartyQuantum,GLSV21_ObliviousTransferMiniQCrypt,CMS23_ObliviousTransferZeroKnowledge}, quantum money~\cite{Wie83_ConjugateCoding,QuantumMoney}, anonymous transmission~\cite{CW05_QuantumAnonymousTransmissions,AnuAnonymity}, copy-protection~\cite{Aar09_QuantumCopyProtectionQuantum}, leader elections and coin flipping~\cite{Leaderelection,weakcoinflippin}, e-voting~\cite{HZBB06_QuantumbasedPrivacyVoting,fedeVoting}, and more.

A large fraction of these protocols, like anonymous transmission protocols~\cite{CW05_QuantumAnonymousTransmissions,AnuAnonymity}, expect parties to share before the beginning of the protocol a number of fundamental quantum states like Bell pairs, \GHZ{} states, or, more generally, arbitrary \emph{graph states}. This task is typically achieved using a \emph{graph state verification protocol}, whose role is to securely distribute a graph state among all parties.

These graph state verification protocols should typically be resilient to deviations from possibly malicious parties, whether they are controlling the source or not. Such security properties are usually proven in a weak, so called game-based model. In this model, we can only prove guarantees on the final quantum state, but we cannot really obtain any guarantee on the behavior of the protocol when it is composed into other protocols (which is the whole point of graph state verification protocols!), or when the adversary is allowed to run attacks in parallel.

As a consequence, it is often unclear if the security of the original protocol is preserved when the graph state is obtained via a \emph{graph state verification protocol} instead of being honestly generated by a trusted third party, leading to the natural question:
\begin{center}%
  \emph{%
    Is it safe to compose any arbitrary protocol with any arbitrary graph state verification protocol?\\
    Is it still secure if the adversary can run multiple attacks in parallel?%
  }%
\end{center}

The study of the composition of protocols is typically done in a security framework where the notion of \emph{functionality} or \emph{resource} is introduced in order to abstract the properties of a given protocol~\cite{Can01_UniversallyComposableSecurity,Unr10_UniversallyComposableQuantum,MR11_AbstractCryptography,Mau12_ConstructiveCryptographyNew}. A functionality can be seen as a trusted party: this way a protocol is said to realize a given functionality if it is impossible to say if we are running the actual protocol or the functionality. With this concept in mind, creating new protocols from sub-protocols is a breeze: we just need to prove that the protocol is secure when the sub-protocol is implemented by a functionality, and we are automatically guaranteed that the protocol will still be secure if the functionality is replaced with any sub-protocol realizing this functionality, even if the adversary is allowed to run attacks in parallel. Composing functionalities is therefore fundamental when designing protocols, since many more advanced protocols are often obtained by composing simpler sub-protocols. This allows to build on previous works and to use functionalities as black-boxes with definite input and outputs.

When using the terminology of these frameworks, the above questions can be reformulated as follows:

\begin{center}
  \emph{Do composable graph state verification protocols exist?}
\end{center}

\paragraph{Our results.} In this work we answer positively to this interrogation, proving that \emph{any} secure graph state verification protocol is composable. This answers an open question raised in~\cite{YDK21_ComposableSecurityMultipartite} that was suggesting that there might not even exist a single composable state verification protocol.

\needspace{3cm}More specifically:
\begin{itemize}
\item We present a method to turn any arbitrary graph state verification protocol, secure in the game-based model, into a composably secure protocol realizing the natural functionality $\Verif$ for graph state verification. This ``compilation'' only adds one round of classical communication at the end of the protocol, and mostly preserves the guarantee of the original protocol. More precisely, if the final state obtained in the real protocol is supposed to be $\eps$-close to the target graph state for some notion of closeness, then the protocol $\eps$-realizes $\Verif$. Our results are expressed in the abstract cryptography framework~\cite{AbstractCryptography}.
\item We also show that any \emph{unchanged} graph state verification protocols can also be considered as composably secure for a slightly different, yet useful, functionality.
\item We show that it is \emph{impossible} to prove that any arbitrary unchanged protocols realizes $\Verif$ having only black-box access to the protocol, without either changing the protocol, or the functionality, showing that the above results are optimal. Note that our impossibility result assumes that the simulators has a certain natural structure, which seems hard to avoid when we consider the protocols as black-boxes.
\item Along the way, we show a protocol to generalize entanglement swapping to arbitrary graph states, which might be of independent interest. Since graph-state manipulation can be challenging using the usual density matrix formalism, we use scalable ZX-calculus~\cite{CK17_PicturingQuantumProcesses,CHP19_SZXcalculusScalableGraphical} to prove our result, asserting the relevance of this language for complex graph state operations.
\end{itemize}

This paper is organized as follows: in \cref{sec:Preliminaries} we start by defining formally graph states verification protocols as well as introducing scalable ZX-calculus and composable security in the Abstract cryptography framework. In \cref{sec:CompSecurityofGraphState}, we then define 
our ideal functionality outputting graph states and we show the equivalence between generic graph state verification protocols and this ideal functionality. Our ideal functionality applies corrections on the outputted graph state, which might seem unreasonable in a concrete implementation. Hence, in \cref{sec:withoutCorr}, we show how to modify graph state verification protocols to realize the ideal functionality without these corrections. Finally, in \cref{sec:UseCases}, we show how our result applies to two existing graph-state verification protocols.

\section{Preliminaries}
\label{sec:Preliminaries}
\subsection{Notations}

We assume basic familiarity with quantum computing~\cite{NC10_QuantumComputationQuantum}. For any subset of index $M \subseteq [n]$, and matrix $G \in \Z_2^{n \times m}$ (resp.\ $x \in \Z_2^n$), we denote $x_M$ as the vector obtained from $x$ after removing lines not in $M$. For any quantum gate $\gateX$, $\gateX^{x_M}$ will denote the application of $\gateX$ on all qubits $i \in M$ such that $x_i = 1$. The fidelity $F(\rho, \sigma)$ of two quantum states $\rho$ and $\sigma$ expressed in term of density matrices is defined as $F(\rho, \sigma) \eqdef \Tr \sqrt{\rho^{1/2}\sigma\rho^{1/2}}$.

%%%%%%%%%%%%%%%%%%%%%%%%%%%%%%%%%%%%%
\subsection{Graph states}

A quantum state $\ket{G}$ is called a graph state if it can be represented by a graph $G=(V,E)$ where the vertices $V$ are $\ket{+}$ states and the edges $E=\{(v_i,v_j)\}_{i=1}^{|E|}$ correspond to controlled $\gateZ{}$ gates ($\gateCZ$) between the vertices $v_i$ and $v_j$. Thus a graph state can be expressed as
\begin{equation}
\label{eq:graphstatedef1}
    \ket{G}=\prod_{(v_i,v_j) \in E} \gateCZ^{\{v_i,v_j\}}\ket{+}^{\otimes V}
\end{equation}

As a consequence to this construction, graph states can also be uniquely defined through a set of so-called \textit{stabilizers}. They are operators who leave the graph state unchanged, or in other words, a graph state is the eigenstate of eigenvalue 1 of its stabilizers. For a given graph state $\ket{G}$ with vertices $V =\{v_i\}_{i=1}^{|V|}$, the stabilizers are

\begin{equation}
\label{eq:graphstatedef2}
S_{v_i}= \sigma^{v_i}_X \prod_{v_j \in N(v_i)}\sigma^{v_j}_Z
\end{equation}

where $\sigma_X$ and $\sigma_Z$ are the Pauli matrices and $N(v_i)$ is the neighborhood of $v_i$. For each vertex $v_i$, we have that $S_{v_i}\ket{G} = \ket{G}$. The set $\{S_{v_i}\}_{v_i \in V}$ characterizes the graph state $\ket{G}$ and is an equivalent definition to the one of \cref{eq:graphstatedef1}.\\

Graph states are multipartite entangled states. This means that measuring one of the qubit of a graph state will have an effect on the adjacent qubits of the state. This property is used for example in measurement-based quantum computing~\cite{MBQC} where a computation is done by sequentially measuring the qubits of a universal graph state (e.g. the brickwork state) which propagates through the state. In a communication network setting where a $n$ qubit graph state is shared among $n$ parties, it can be used to create a shared list of correlated bits by asking each party to measure their qubit. This list can then be used in many contexts, for example to create a common secret key among the parties in so-called conference key agreement protocols, the multipartite counterpart of quantum key distribution.\\

\subsection{Scalable ZX-calculus}

The ZX calculus~\cite{CD08_InteractingQuantumObservables} is a language allowing us to prove equality between quantum operations diagrammatically, using a simple set of rewriting rules between graphs. We present in \cref{appendix:introZX} the basic properties of the ZX calculus, as well as one extension called scalable ZX~\cite{CHP19_SZXcalculusScalableGraphical}, that we use in our proofs. Note however that all equations proven in this article can be checked manually without ZX calculus, and that readers not interested in checking the proof should be able to read this article without any ZX knowledge. We refer the interested reader to~\cite{van20_ZXcalculusWorkingQuantum} and~\cite{CK17_PicturingQuantumProcesses} for more details.

%%%%%%%%%%%%%%%%%%%%%%%%%%%%%%%%%
\subsection{Composable security and Abstract Cryptography}
\label{sec:CompSecDef}
As mentioned in the introduction, graph state verification protocols are usually used as sub-protocols in more complex protocols. They are meant to be reused many times, in a setting where a graph state is shared between $n$ parties and some of the parties and/or the source may be dishonest. Hence, they do not only need to be secure for one run, but the security should also hold for an arbitrary number of repetition of the protocol. This level of security is called composable security, as opposed to game-based security where we study specific attack models against a protocol.\\

To prove composable security, one have to use a so called simulation-based framework where the security proofs are composable. In this paper, we will use the Abstract (or Constructive) Cryptography (AC) model, a top-down approach developed by U.\ Maurer and R.\ Renner~\cite{AbstractCryptography,FromToConsCrypto,ConsCryptography}. This framework uses the concept of abstract systems to express cryptography as a resource theory. A cryptography protocol is viewed as the construction of some \textit{ideal} resource $\cS$ out of other \textit{real} resources $\cR$. Resources are box-like abstract systems with interfaces that are accessed by the parties. They represent any non-local operation such as communication channels, but also more involved functionalities to model for instance a coin tossing protocol. They can be composed in sequence or in parallel to create bigger resources. 

Local operations\===like the protocol run by a given party\===are called \emph{converters}. They are plugged into a single interface of a resource, changing the interaction of the resource with the outside world. A converter $\pi$ attached to a resource $\cR$ creates a new resource that we write $\pi\cR$. We usually call the converters $\pi$ acting in the real world \emph{protocols}, as they represent the local operations done by honest parties. In the ideal world, the \emph{honest use of an ideal resource} is done through converters called \emph{filters}, typically denoted $\bot$, that send the input expected from honest parties to the resource.

Finally, the distance between two resources is formalised through the notion of a distinguisher. It is an object that access all the interfaces of two given resources, such as $\pi\cR$ and $\bot\cS$, and tries to distinguish them by sending inputs and comparing the outputs. When considering statistical security, the distinguisher can be given unlimited computing power. On the other hand, computational security only allows distinguishers to have limited computing power. When any distinguisher accessing the interfaces of both resources cannot decide which of the two system is the ideal or the real one, we say that the resources are equivalent and write $\pi\cR\approx\bot\cS$.

Resources, converters and distinguishers are the building blocks of the AC theory.  We refer the reader to~\cite{AbstractCryptography} for more details about the mathematical construction of the framework.\\

The secure construction of an ideal resource, that represents the ideal functionality that we want to achieve, from a concrete resource, that represents the actual realization of the protocol, is proven by showing a series of equivalences. In AC, a dishonest behaviour from a party is represented by unplugging the associated converter on the concrete resource, which means that this party is not following the protocol. This leaves some new interfaces accessible for a distinguisher on the concrete resource. To prove that the security of the protocol still holds, one has to find a converter called a \textit{simulator} to plug into the ideal resource to make it indistinguishable from the dishonest concrete resource. The full security proofs thus consist in finding simulators for each possible subset of dishonest party. We give the security definition in AC below in \cref{def:securityDef}.

\begin{definition}[Security in AC]\label{def:securityDef}
    
 let $\Pi=\{\pi_i\}_{i=1}^n$ be a protocol run by $n$ parties using the concrete resource $\cR$  and let $\cS$ be an ideal resource with all the desired properties expected from the protocol. $\cR$ and $\cS$ have interfaces $I$. We say that \textbf{$\Pi$ securely realizes $\cS$ out of $\cR$  within $\eps$} or that \textbf{$\Pi$ $\eps$-realizes $\cS$} and write $\cR\xrightarrow{(\Pi,\eps)} \cS $ if there exist simulators $\sigma = \{\sigma_{i}\}$  such that\footnote{Technically speaking~\cite[Thm.~2]{MR11_AbstractCryptography}, we should also define filters for $\cR$ possibly restricting the access to $\cR$ for honest participants. For simplicity, and without loss of generality, we will often assume that the filters for $\cR$ are trivially forwarding their inputs, and we integrate the original filters in the parties. This way, proving the correctness of a protocol using $\cR$ can be simplified as $\pi_H \cR \approx \filter_H \cS \sigma_M$ instead of $\pi_H \filter^{\cR}_H \cR \approx \filter_H \cS \sigma_M$.}:
\begin{equation}
\label{eq:resscons}
    \forall H \subseteq I, \pi_{H}  \cR \approx_{\eps} \sigma_{I \setminus H}\cS\bot_H,
\end{equation} 
with $\forall H \subseteq I, \pi_{H}=\{\pi_i\}_{i\in H}$ and $\bot_{H}=\{\bot_i\}_{i\in H}$.\\
\end{definition}

\begin{remark}\label{rq:GlobalSimu}
Note that in this paper, we will equip the ideal resources with a communication channel that can forward any\===possibly quantum\===message between any party. This is only needed to allow simulators to perform non-local operations, and is filtered for honest parties. For a given subset of honest parties $H$, we will thus have to find only one global simulator $\sigma_{I \setminus H}$ such that the above equivalence relation holds.\end{remark}

Since the AC framework is composable, any security proof proven within the framework is composable. This means all $\cR,\cS$ and $\cT$ resources and $\pi,\nu$ converters (protocols) such that $\cR\xrightarrow{\pi} \cS$ and $\cS\xrightarrow{\nu} \cT$ we have that

\begin{equation}
    \cR \xrightarrow{\pi} \cS \wedge \cS \xrightarrow{\nu} \cT \implies \cR \xrightarrow{ \nu \circ \pi} \cT.
\end{equation}

\section{Composable security of graph state verification}
\label{sec:CompSecurityofGraphState}
\subsection{Definition of the ideal verification resources}
In this section, we present the ideal functionalities that we will use to prove the security of generic graph state verification protocols.

\paragraph{Ideal resource.}
Note than in this article we will consider two functionalities:
\begin{itemize}
\item $\Verif^f$ (\cref{protoc:idealResourceVerif}) is the functionality that we consider in this section: on the one side, we can show that any secure graph state verification algorithm realizes this functionality, without any modification to the protocol. However, this functionality is less natural than what one would expect, as it also allows malicious adversaries to apply some malicious corrections to the graph state given to honest adversaries. The function $f$ acts as a safeguard to only allow some corrections to be performed by the adversary. More precisely, $f$ must be a function taking as input a subset of corrupted party $M$ and $(x,z)$, a list of $X$ and $Z$ corrections to apply on the qubits sent to honest parties, and must output either $\top$ if this correction is allowed or $\bot$ otherwise. We discuss later in \cref{rq:justificationF} the motivations behind this verification.
\item $\Verif$ (\cref{protoc:idealResourceVerifNoCorrections}), on the other hand, is much simpler and closer to what one would expect as it simply sends $\ket{G}$ to all parties (unless the malicious parties abort), but we cannot directly show that all protocols realize this functionality: we need to apply an additional step where all parties apply a random stabilizer on their part of the state.
\end{itemize}
Since both resources have pro and cons, we study them separately in the following two sections, starting from $\Verif^f$. $\Verif^f$, that we describe formally in \cref{protoc:idealResourceVerif} and informally in \cref{fig:IdealVerif}, is an abstract system outputting all qubits of a given graph state $\ket{G}$ or an abort signal, with interfaces allowing to model possible dishonest behaviors from the parties and the source. Note that $\Verif^f$ is not a functionality allowing to create any graph state, but for any graph state $\ket{G}$, we can construct an ideal resource $\Verif^f$ outputting $\ket{G}$. 

The $\Verif^f$ resource has $n+1$ interfaces, one for each party and one for the source. For each party $i$, the value $c_i \in \{0,1,\bot\}$ indicates whether the party is honest or dishonest. To an honest party ($c_i=0$), it will output either an Abort signal or a qubit from the graph state $\ket{G}$. To a dishonest party ($c_i =1$ or $ \bot$), it will output the corresponding qubits of $\ket{G}$ and then wait for corrections $(a_i,b_i)$. For readability, in \cref{fig:IdealVerif}, we show on the left the input and output corresponding to honest uses of the resource from the parties, while the bottom input and output correspond to dishonest behavior. It should however be kept in mind that only one input/output interface is accessible to each party. The resource $\Verif^f$ is equipped with a function $f$ that outputs a Boolean stating whether the corrections proposed by the dishonest parties have the correct form. The details of how the resource works is given below in Protocol 1, and the details of $f$ will be explicit in the security proof. 

\begin{figure}[!ht]
    \centering
    \includegraphics[width=15cm]{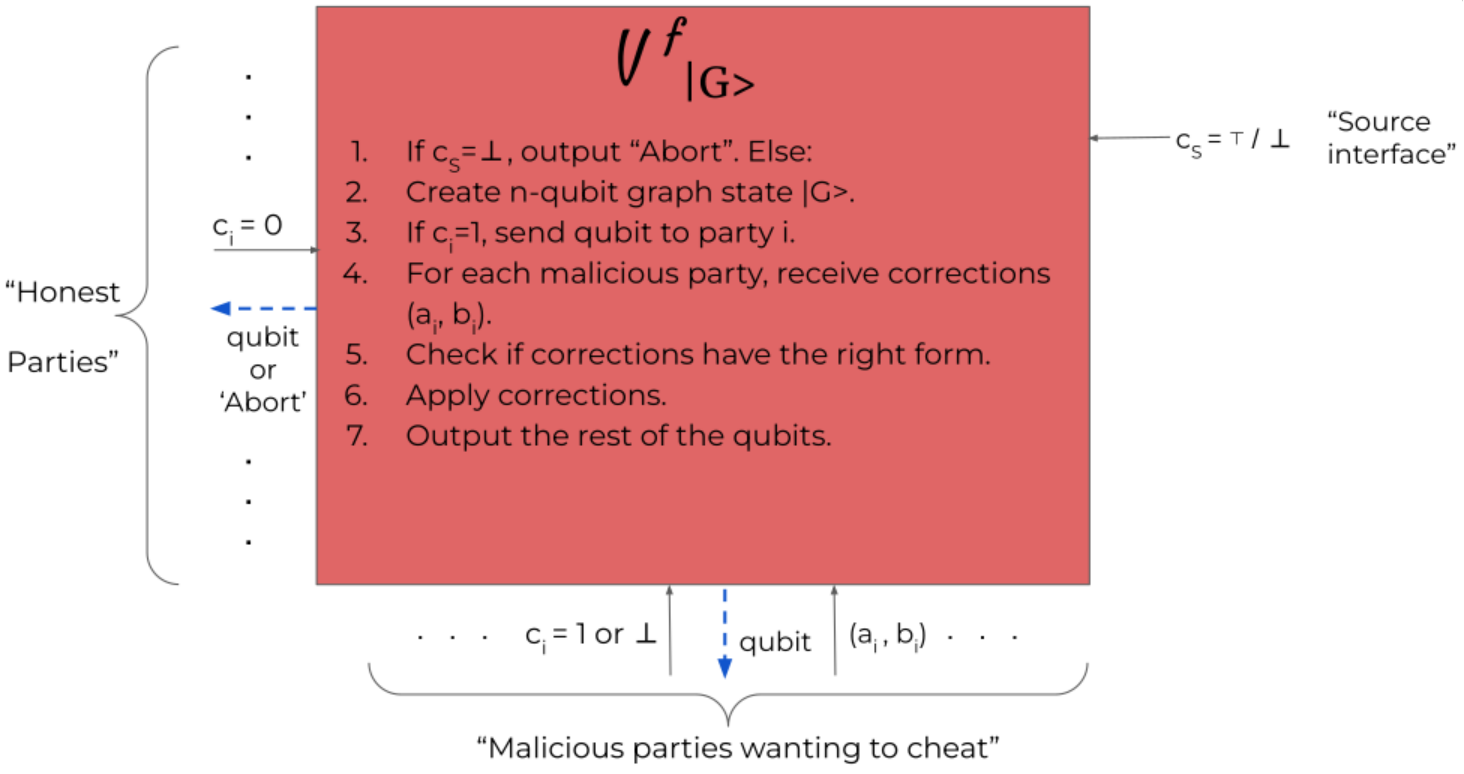}
    \caption{Informal presentation of the ideal resource $\Verif^f$  for verified graph state sharing.}
    \label{fig:IdealVerif}
\end{figure}

\begin{resource}{Ideal resource $\Verif^f$}\label{protoc:idealResourceVerif}
  \begin{enumerate}
  \item Receive from the source's interface $c_\source \in \{\top,\bot\}$
  \item If $c_\source = \bot$, abort and send $\bot$ to all parties. Otherwise it continues.
  \item Create $\ket{G}$.
  \item Receive $\{c_i\} \in \{0,1,\bot\}^n$ from each party. If any $c_i = \bot$ send $\bot$ to all parties and abort. Otherwise, let $M = \{c_i \mid c_i = 1\}$ be the set of malicious parties, $H = [n] \setminus M$ the honest parties, $m = |M|$ the number of malicious parties and $h=|H|$ the number of honest parties.
  \item For each $i \in M$, send the $i$-th qubit of $\ket{G}$ to interface $i$.
  \item Receive from each malicious party $i \in M$ corrections $(a_i,b_i) \in (\{0,1\}^{h})^2$, and define $x = \xor_{i \in M} a_i=\{x_i\}_{i \in H}$ and $z = \xor_{i \in M} a_i=\{z_i\}_{i \in H}$. 
  \item If there is a least one malicious party, check if $f_G(M,x,z) = \top$ (to check that corrections are well formed), if not sent $\bot$ to all parties and abort.
  \item If ok, Apply $\gateZ^{z_1}\otimes...\otimes \gateZ^{z_{h}}$ and $\gateX^{x_1}\otimes...\otimes \gateX^{x_{h}}$ to the remaining qubits of $\ket{G}$.
  \item Send these qubit to the parties.
  \end{enumerate}
  Additionally, we include in $\Verif^f$ an additional communication channel $\cC$ that can forward any\===possibly quantum\===message between any party (see \cref{rq:GlobalSimu}).
  We also define naturally the filter $\filter_\source$ as the converter that sends $\top$, and for any $i \in [n]$, we define $\filter_i$ as the converter that sends $c_i = 0$ and forwards any message sent or received by the ideal function functionality (excluding messages sent on $\cC$ that are just blocked).
\end{resource}

In Abstract Cryptography, we typically model the \emph{honest use of an ideal resource} by adding special converters called \emph{filters} that send the input expected from honest parties to the resource. In our case, an honest use of the $\Verif^f$ resource corresponds to the source sending $c_S = \top$ and each party $i$ sending $c_i = 0$ to the resource. We thus define the following: the filter $\bot_S$ that corresponds to an honest source sending $c_S = \top$ and, for $i\in [n]$, the filter $\bot_i$ that corresponds to an honest party sending $c_i = 0$ to $\Verif^f$. For the sake of simplicity, we will write $\bot_{[n]}$ the filter corresponding to the parallel composition of $\bot_i$ for all $i \in [n]$. An honest use of the ideal resource is thus represented by the so-called filtered resource, that we show in \cref{fig:FilteredResource}.

\begin{figure}[!ht]
    \centering
    \includegraphics[width=7cm]{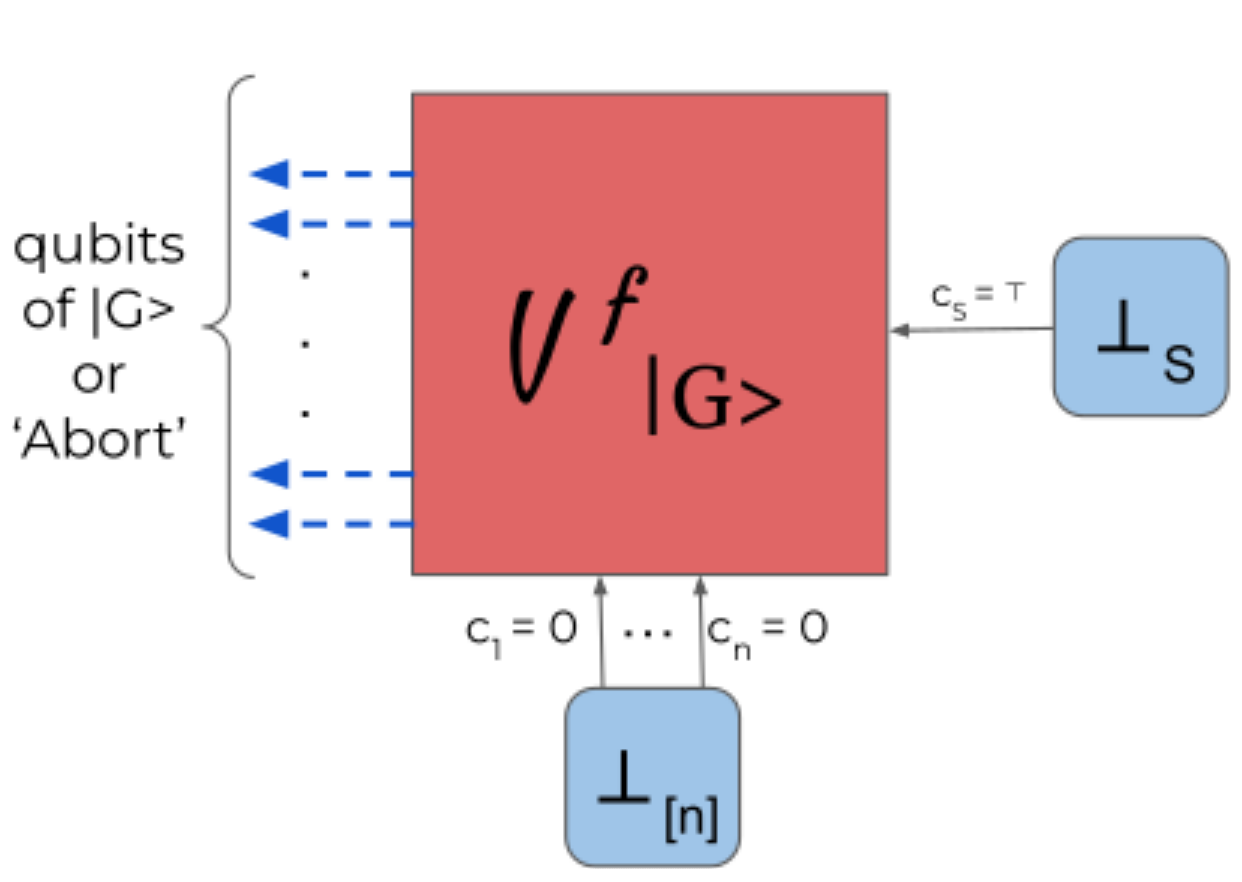}
    \caption{Ideal filtered resource. The filter $ \bot_S$ corresponds to an honest source inputting $c_S = \top$ and, for $i\in [n]$ the filter $\bot_i$ corresponds to an honest party sending $c_i = 0$ to $\Verif^f$}
    \label{fig:FilteredResource}
\end{figure}

\subsection{Definition of the concrete verification protocols}
Generally, protocols to verify the preparation of graph states consist of the following steps: first, quantum states are shared between $n$ parties, one qubit of each state per party, then, they test a random selection among these states and if the checks are positive, they keep the others for later use. The tests usually consist in local operations and measurements, classical communication and multiparty computation that outputs a bit indicating if the state is far away the desired state reduced to the honest parties, for a certain distance (usually the trace distance). By randomising which states they test and which one they use, the parties prevent a malicious source from sending the desired states. Verification protocols can also be done sequentially, by asking the source to send quantum states one after the other and randomly choosing the ones that they test and the ones that they use. \\

Verification protocols provide a bound on the probability that, given the protocol has not aborted after a certain number of tests, the final reduced state that is used for computation or communication by the honest parties is close to the desired graph state reduced to the honest parties. This bound depends on the number of parties, the number of tests, the number of dishonest parties and the maximum distance that is accepted with the desired graph state.\\

\paragraph{Definition of graph state verification protocols.}

The literature uses different (but mostly equivalent) security definitions when considering graph state verification protocols. Note that the security of these protocols is never expressed in the composable AC framework. The main contribution of this article is actually to prove that any protocol fulfilling the security definition of \cref{def:graphStateVerifProtoc} below is also composably secure.

\begin{definition}[Graph state verification protocol]\label{def:graphStateVerifProtoc}
  Let $\Pi = \{\pi_i\}_{i \in [n]} \cup \{\pi_\source\}$ be a protocol between $n$ parties and a source, interacting through a resource $\cR$ in charge of modeling, for instance, the communication channels between all parties. We will say that $\Pi$ is a \emph{$\eps$-graph state verification protocol} if the following properties are respected:
  \begin{itemize}
  \item \textbf{Correctness}: if all parties are honest, they output a state negligibly close (in trace distance) to $\ket{G}$, i.e.\ $\pi_{[n]}\cR \pi_\source$ outputs $\rho$ such that $\TD(\rho, \ket{G}) \geq 1-\negl[\lambda]$.
  \item \textbf{Security}: for any set of honest parties $H \subseteq [n] \cup \{\source\}$, the honest parties output their state at the same time\footnote{This can for instance be done using a broadcast channel, which is anyway implicitly needed in most existing works.}, each party outputting either a special symbol $\ket{\bot}$ if they aborted, or a quantum state otherwise\footnote{Note that this implies that the Hilbert space is spanned by $\{\ket{0},\ket{1}, \ket{\bot}\}$, where $\ket{\bot}$ is orthogonal to the other two states. This is quite practical as this way we can only maintain $n$ registers instead of $2n$, and we do not need to worry about the content of the other register when the abort register contains an abort. Note that otherwise, this is mostly equivalent as we can test if a party aborted by simply measuring $(\ketbra{\bot}{\bot}, I - \ketbra{\bot}{\bot})$, without disturbing the state if it is only spanning $(\ket{0}, \ket{1})$.}. Moreover, when considering any adversary\footnote{$\cA$ might be bounded or unbounded depending on the assumptions on the protocol.} $\cA$ corrupting parties in $[n]\setminus H$, there exists $p \in [0,1]$ such that:
    \begin{align}
      F(\rho, \sigma) \geq 1-\eps(\lambda)\label{eq:graphStateVerifProtocSecurity}
    \end{align}
    where $F$ is the fidelity (we use the definition of~\cite{NC10_QuantumComputationQuantum}, sometimes called the square root fidelity), $\rho \eqdef \esp[\rho_i \gets \pi_{H}\cR \cA]{\Tr_{[n]\setminus H} \rho_i}$ is the averaged state obtained by the honest parties at the end of the protocol, where we average over all randomness involved in $\cA$ and in the whole protocol\footnote{We can also equivalently purify $\cA$ and the protocols using the Stinespring dilation, and simply trace-out all the registers except for the output register owned by the honest party. The Stinespring dilation allows us to postpone any measurement by essentially replacing them with a $\gateCNOT$ on an auxiliary qubit, and to sample randomness by creating a $\ket{+}$ state measured using the postponed measurement that we just described.} and $\sigma \eqdef p\Tr_{[n]\setminus H}(\ketbra{G}{G}) + (1-p) \ketbra*{\bot^{|H|}}{\bot^{|H|}}$ denotes the mixture where all honest parties either abort at the same time with probability $1-p$ or output a qubit that is part of $\ket{G}$.
  \end{itemize}
\end{definition}

\paragraph{Concrete resource}

In this work, we model all graph state verification protocols in the same way: we abstract all the resources required to perform these protocols into one resource $\cR$ and the local operations done by each parties into converters $\{\pi_i\}_{i=1}^n$. For instance, $\cR$ can contain the quantum channel from the source to each party, some authenticated classical channels used to communicate between the parties, a coin flipping resource or a multiparty computation resource that may involve shared randomness resources etc. The converters $\{\pi_i\}_{i=1}^n$, that we will write $\pi_{[n]}$, correspond to the protocol followed by the honest parties and often consist in applying some quantum operations on their qubit, measuring them and using the output in the multiparty computation that tests the state. Finally the source's protocol $\pi_\source$ generally consists either in sending a certain number of copies of a desired state or in sending them one by one. We illustrate in \cref{fig:ConcreteVerif} the concrete resource $\pi_{[n]}\cR \pi_\source$ corresponding to an honest run of the protocol, which simply outputs a state that is close to a fixed graph state $\ket{G}$.

\begin{figure}[!ht]
    \centering
    \includegraphics[width=15cm]{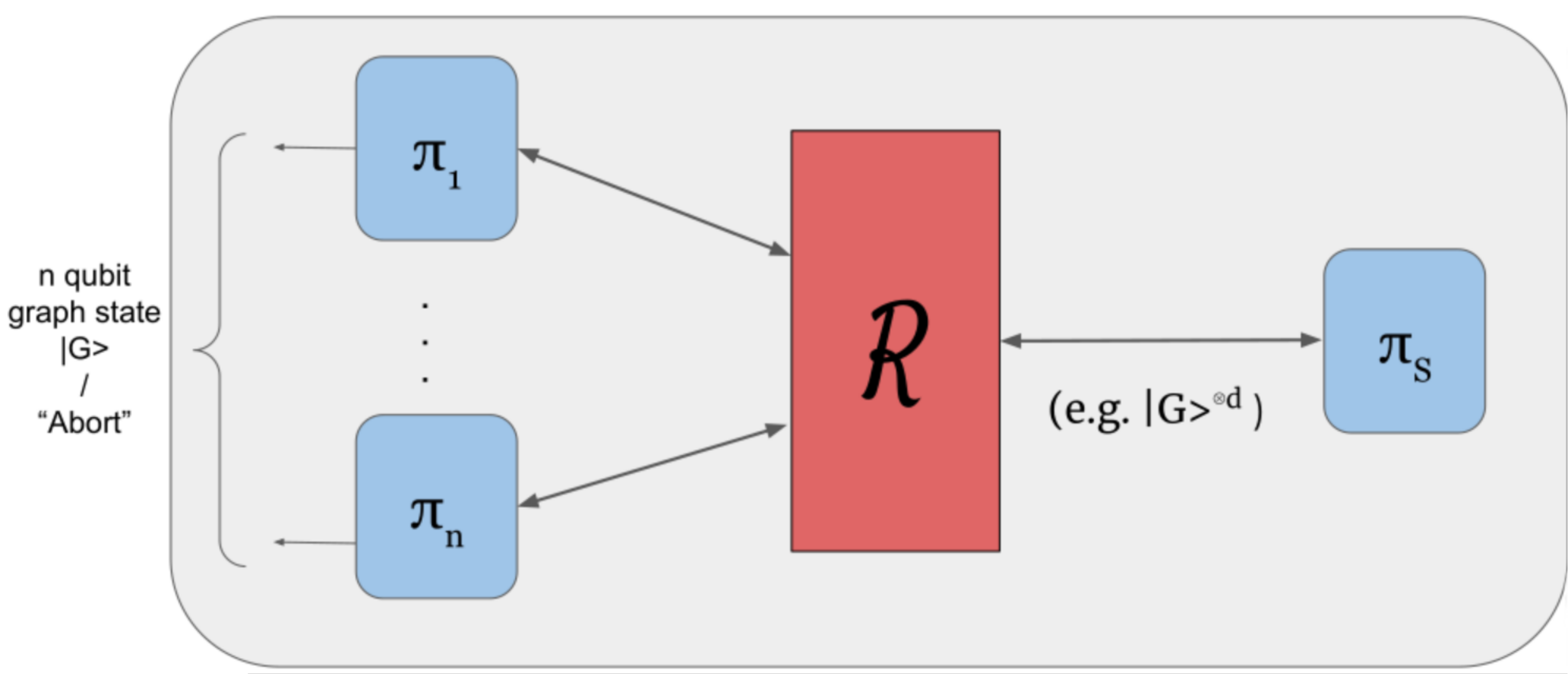}
    \caption{Concrete resource for generic graph state verification protocols}
    \label{fig:ConcreteVerif}
\end{figure}

\subsection{Security proof}
In this section we prove our main result, namely the composable security of any graph state verification protocol. Following the Security \cref{def:securityDef}, we will prove the following theorem in the rest of this section:

\begin{theorem}\label{thm:mainThm} 
  Let $\ket{G}$ be an arbitrary graph state, and let $\Pi$ be an $\eps$-graph state verification protocol according to \cref{def:graphStateVerifProtoc}. Then, $\Pi$ $(2\sqrt{2\eps - \eps^2})$-realizes the functionality $\Verif[\ket{G}]^f$, where the allowed corrections $f$ are defined as follows: let $U,r,R$ be defined like in \cref{lem:GaussPivot}, then we define $f(M, x, y)$ as the function that outputs $\top$ iff the last $|n|-r$ vectors of $Ux$ are $0$’s, and if $(U^T)^{-1}(z \xor G_H x) =
  \begin{bmatrix}
    b\\
    R^Tb
  \end{bmatrix}
  $ for some arbitrary vector $b \in \{0,1\}^r$.
\end{theorem}

\subsubsection{Correctness}

\begin{proof}
  
  The first step is to prove the correctness of the protocol when all parties are honest, i.e.\ we need to show that $\pi_{[n]}\cR \pi_\source \approx \filter_{[n]} \Verif^f \filter_\source$. This is also pictured in \cref{fig:HonestEquiv}, where on the left-hand side we represented the filtered version of the $\Verif^f$ resource, simply outputting the final\===expected\===state, and corresponding to a honest use of the resource.

  \begin{figure}[!ht]
    \centering
    \includegraphics[width=15cm]{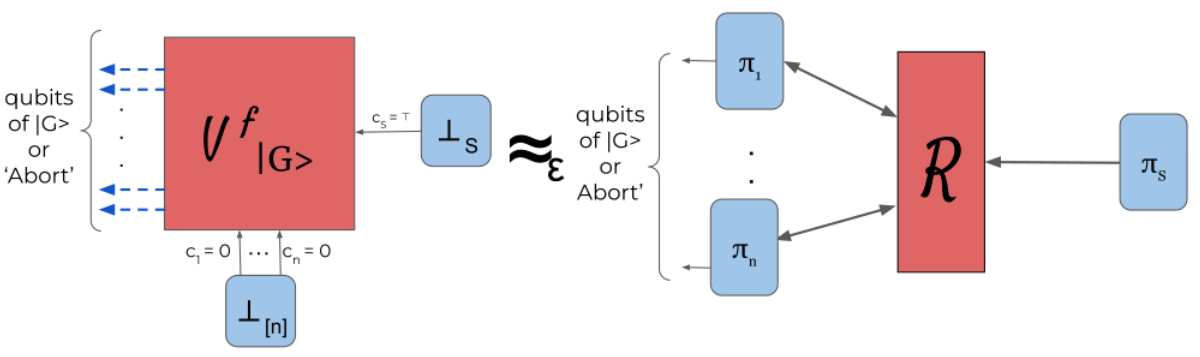}
    \caption{Correctness of the protocol: $\filter_{[n]} \Verif^f \filter_\source$ $\approx$ \Concrete}
    \label{fig:HonestEquiv}
  \end{figure}

  This is actually a direct consequence of the correctness of the protocol. First, we can easily see that since the filters disable any malicious behavior, $\filter_{[n]} \Verif^f \filter_\source$ is the system that outputs $\ket{G}$ with probability $1$. Moreover, due to the correctness property of $\Pi$, $\pi_{[n]}\cR \pi_\source$ outputs a state indistinguishable from $\ket{G}$. Hence, any distinguisher trying to distinguish between the two resource only gets indistinguishable copies of $\ket{G}$, i.e. $\pi_{[n]}\cR \pi_\source \approx \filter_{[n]} \Verif^f \filter_\source$. This ends the correctness proof.
\end{proof}

\subsubsection{Security}
The next step in the security proof, (see \cref{def:securityDef}), is to study the case of dishonest parties tampering with the verification protocol. This means that some subset $M \in [n] \cup \{\source\}$ of the parties, possibly including the source, are no longer following their local protocols $\pi_i$. The concrete resource representing the real protocol becomes $\pi_{H}\cR$ where $H \subseteq [n] \cup \{\source\}$ are the honest parties.

On the ideal resource, the malicious behavior of some subset $M$ of the parties is represented by removing the filters $\{\bot_i\}_{i \in M}$, letting the adversaries access interfaces of the ideal resource $\filter_{H} \Verif^f$. We need to find a simulator $\sigma_M$ such that there exists an $\epsilon$ such that $\pi_{H}\cR  \approx_{\epsilon} \filter_{H} \Verif^f \sigma_M$. We show the representation in AC in \cref{fig:DishonestTOUT}.

\begin{figure}[!ht]
    \centering
    \includegraphics[width=15cm]{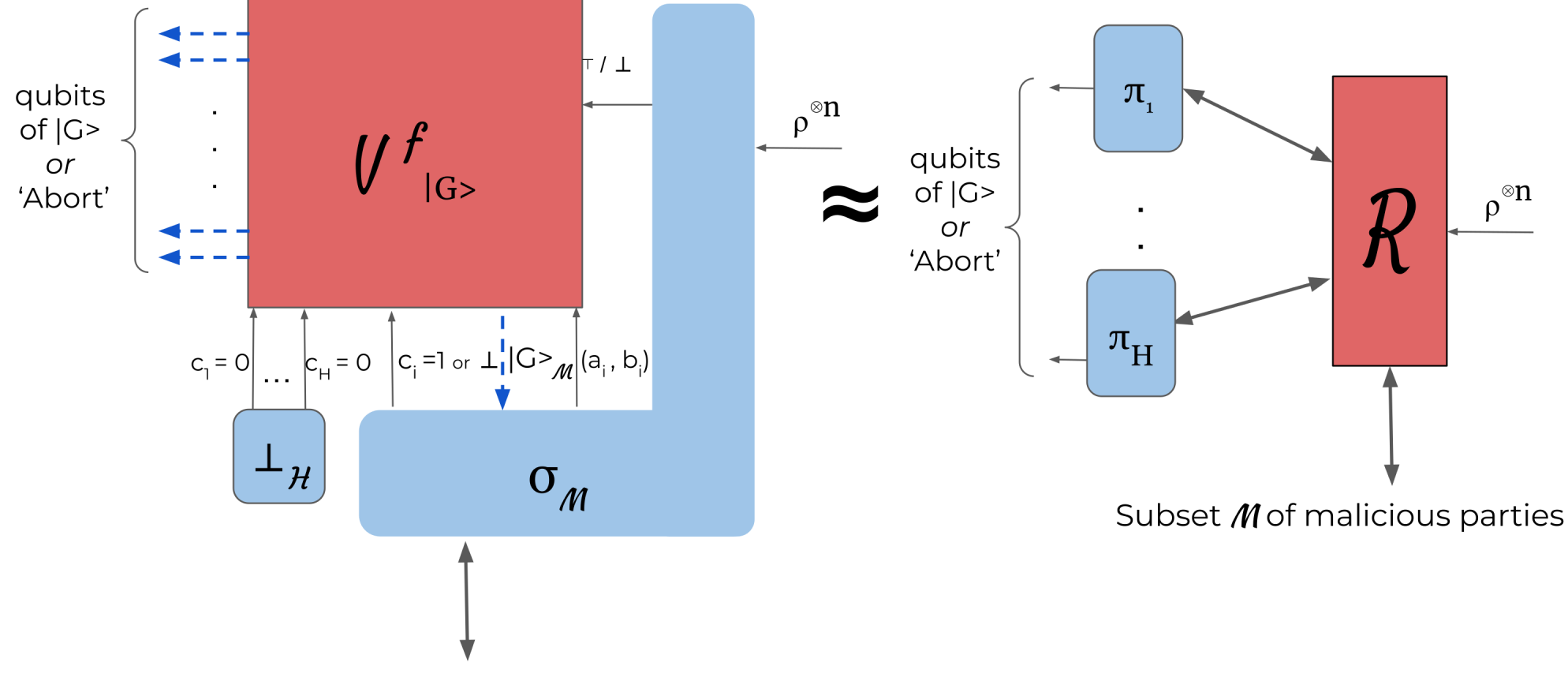}
    \caption{General case of the security proof: we need to find $\sigma_M$ such that $\pi_{H}\cR  \approx_{\epsilon} \filter_{H} \Verif^f \sigma_M$}
    \label{fig:DishonestTOUT}
\end{figure}

In order to construct such a simulator, and thus to prove that all graph state verification protocols are composable, we first study a property of graph states that we call ``mergeable''.

\paragraph{Mergeable states.}

 In order to motivate this property, let us try to prove that a given graph state verification protocol $\Pi = (\pi_1, \dots, \pi_n, \pi_\source)$ is composably secure, where $\pi_\source$ is the protocol followed by the source while $\pi_i$'s are the protocols followed by other parties. As explained above, this is done by finding a simulator to emulate on the ideal resource the interfaces of the concrete resource left open when removing the converters $\pi_i$ corresponding to the dishonest parties. Let us try naively to find such a simulator. If $H$ is the set of honest parties and $M \eqdef \bar{H}$ is the set of malicious parties, possibly including the source, we want to find a simulator $\sigma_M$ such that $\pi_H \cR \approx_\eps \bot_H \Verif \sigma_M$:
\begin{align}
  \begin{ZX}<circuit,forward=\cR>[circuit, thick lines,zx column sep=6mm]
    \rar[latex-latex] & \zxGate{\pi_H} \rar[latex-latex] & \zxGate{\cR} \rar[latex-latex] &
  \end{ZX} \approx_\eps
  \begin{ZX}<circuit,forward=\cV,forward=\Verif>[circuit, thick lines,zx column sep=6mm]
    \rar[latex-latex] & \zxGate{\bot_H} \rar[latex-latex] & \zxGate{\Verif} \rar[latex-latex] & \zxGate{\sigma_M} \rar[latex-latex] &
  \end{ZX}\label{eq:goalSim}
\end{align}
where we represented the interfaces of honest parties on the left and the interface of malicious parties on the right, and where $\bot_H$ is the filter blocking the honest interfaces of the ideal resource.\\

For the sake of explanation, let us first try to find a simulator that works when the distinguisher is first running the honest protocol $\pi_M$ (let us write this part of the distinguisher $\cD_0$): in that case, we can rewrite \cref{eq:goalSim}, which gives us:
\begin{align}
  \begin{ZX}<circuit,forward=\cR,fw=\cD>[circuit, thick lines,zx column sep=6mm]
    \rar[latex-latex] & \zxGate{\pi_H} \rar[latex-latex] & \zxGate{\cR} \rar[latex-latex] & \zxGate{\cD_0} \rar[-latex] &
  \end{ZX} \approx_\eps
  \begin{ZX}<circuit,forward=\cV,forward=\Verif,forward=\cD>[circuit, thick lines,zx column sep=6mm]
    \rar[latex-latex] & \zxGate{\bot_H} \rar[latex-latex] & \zxGate{\Verif} \rar[latex-latex] & \zxGate{\sigma_M} \rar[latex-latex] & \zxGate{\cD_0} \rar[-latex] &
  \end{ZX}\label{eq:goalSimD0}
\end{align}
i.e., after replacing the definition of this first naive distinguisher:
\begin{align}
  \begin{ZX}<circuit,forward=\cR>[circuit, thick lines,zx column sep=6mm]
    \rar[latex-latex] & \zxGate{\pi_H} \rar[latex-latex] & \zxGate{\cR} \rar[latex-latex] & \zxGate{\pi_M} \rar[-latex] &
  \end{ZX} \approx_\eps
  \begin{ZX}<circuit,forward=\cV,forward=\Verif>[circuit, thick lines,zx column sep=6mm]
    \rar[latex-latex] & \zxGate{\bot_H} \rar[latex-latex] & \zxGate{\Verif} \rar[latex-latex] & \zxGate{\sigma_M} \rar[latex-latex] & \zxGate{\pi_M} \rar[-latex] &
  \end{ZX}\label{eq:goalSimD0Replaced}
\end{align}
Due to the correctness of $\Pi$, in the real world both parties will share $\ket{G}$, i.e.\ we can simplify the LHS with:
\begin{align}
  \begin{ZX}<circuit,forward=\cR>[circuit, thick lines,zx column sep=6mm]
     & \zxGate{\ket{G}} \ar[rd,|--,-latex,start anchor={[xshift=1mm]south}, Bn.Args={M}{pos=.75},thick] \ar[ld,|--,-latex,start anchor={[xshift=-1mm]south},Bn.Args={H}{pos=.75},thick] & \\[\zxZeroRow+3mm]
     &                                                        & 
  \end{ZX} \approx_\eps
  \begin{ZX}<circuit,forward=\cV,forward=\Verif>[circuit, thick lines,zx column sep=6mm]
    \rar[latex-latex] & \zxGate{\bot_H} \rar[latex-latex] & \zxGate{\Verif} \rar[latex-latex] & \zxGate{\sigma_M} \rar[latex-latex] & \zxGate{\pi_M} \rar[-latex] &
  \end{ZX}\label{eq:simplifiedGoal}
\end{align}

Since we want our simulator to work for any protocol $\Pi$, the most natural thing to do is to let the simulator start with simulating $\pi_H \cR$ directly in a black-box manner otherwise we would not even know what to send to the distinguisher. The simulator hence forwards all messages received on the outer interface to the appropriate honest party via $\cR$, and all messages from honest parties to the outer interface. At the end of this interaction, the simulator will obtain a state outputted by the honest parties: if one party aborted, the simulator can tell the ideal functionality to abort but this should never occur in this simplified analysis where the distinguisher is running the honest protocol. If no party aborted, and if $\Pi$ is secure, then:
\begin{itemize}
\item we should get from $\pi_H$ a state close to $\ket{G}^H$,
\item and on the other hand, the simulator will receive from the ideal functionality $\Verif$ a state $\ket{G}^M$.
\end{itemize}
By the correctness of $\pi$ and the definition of $\Verif$, we can therefore rewrite the RHS of \cref{eq:simplifiedGoal} as follows, where $\zx{\zxGate{?}}$ represents a yet unknown operation run by the simulator that we need to determine:
\begin{align}
  &\begin{ZX}<circuit,forward=\cV,forward=\Verif>[circuit, thick lines,zx column sep=6mm]
    \rar[latex-latex] & \zxGate{\bot_H} \rar[latex-latex] & \zxGate{\Verif} \rar[latex-latex] & \zxGate{\sigma_M} \rar[latex-latex] & \zxGate{\pi_M} \rar[-latex] &
  \end{ZX}\label{eq:startDiagProofImpossibliliy}\\ &=
  \begin{ZX}<circuit,forward=\cV,forward=\Verif,fw=\cR>[circuit, thick lines,zx column sep=6mm,mbr=1]
    \rar[latex-latex] & \zxGate{\bot_H} \rar[latex-latex] & \zxGate{\Verif} \rar[latex-latex] & \zxGate{?} \rar[latex-] \zxCont{1}{3}{$\sigma_M$} & \zxGate{\pi_H} \rar[latex-latex] & \zxGate{\cR} \rar[latex-latex] & \zxGate{\pi_M} \rar[-latex] &
  \end{ZX}\\
  &=
  \begin{ZX}<circuit,forward=\cV,forward=\Verif>[circuit, thick lines,zx column sep=6mm]
   & \zxGate{\ket{G}} \ar[rd,|--,-latex,start anchor={[xshift=1mm]south}, Bn.Args={M}{pos=.75},thick] \ar[ld,|--,-latex,start anchor={[xshift=-1mm]south},Bn.Args={H}{pos=.75},thick] &            & \zxGate{\ket{G}} \ar[rd,|--,-latex,start anchor={[xshift=1mm]south}, Bn.Args={M}{pos=.75},thick] \ar[ld,|--,-latex,start anchor={[xshift=-1mm]south},Bn.Args={H}{pos=.75},thick] & \\[\zxZeroRow-1mm]
   &                                                                                                                                                                                        & \zxGate{?} &                                                                                                                                                                                        & 
  \end{ZX}
\end{align}
We can inject this back into \cref{eq:simplifiedGoal} to get:
\begin{align}
  \begin{ZX}<circuit,forward=\cR>[circuit, thick lines,zx column sep=6mm]
     & \zxGate{\ket{G}} \ar[rd,|--,-latex,start anchor={[xshift=1mm]south}, Bn.Args={M}{pos=.75},thick] \ar[ld,|--,-latex,start anchor={[xshift=-1mm]south},Bn.Args={H}{pos=.75},thick] & \\[\zxZeroRow+3mm]
     &                                                        & 
  \end{ZX} \approx_\eps
  \begin{ZX}<circuit,forward=\cV,forward=\Verif>[circuit, thick lines,zx column sep=6mm]
   & \zxGate{\ket{G}} \ar[rd,|--,-latex,start anchor={[xshift=1mm]south}, Bn.Args={M}{pos=.75},thick] \ar[ld,|--,-latex,start anchor={[xshift=-1mm]south},Bn.Args={H}{pos=.75},thick] &            & \zxGate{\ket{G}} \ar[rd,|--,-latex,start anchor={[xshift=1mm]south}, Bn.Args={M}{pos=.75},thick] \ar[ld,|--,-latex,start anchor={[xshift=-1mm]south},Bn.Args={H}{pos=.75},thick] & \\[\zxZeroRow-1mm]
   &                                                                                                                                                                                        & \zxGate{?} &                                                                                                                                                                                        & 
  \end{ZX}\label{eq:finalSimplifiedGeqGG}
\end{align}
We should therefore search for a ``merging'' operation that can combine two shared state $\ket{G}$ into a single copy of $\ket{G}$ while working \emph{only} on (different) parts of two copies of $\ket{G}$. Unfortunately, it is easy to see that it is impossible to obtain such a merging map. For instance, if we take $\ket{G}$ to be a Bell pair, one of the simplest graph state, then after tracing out the registers of $\zx{\zxGate{?}}$ (that cannot signal any measurement outcome to the outside world) we obtain two un-entangled states, which cannot possibly be equal to a single entangled Bell pair. This can be seen for instance diagrammatically using the formalism described in \cite{CK17_PicturingQuantumProcesses}, or using non-signaling as formalised in~\cref{thm:impossibilityBlackBox}:
\begin{align}
  \begin{ZX}[circuit]
                          & \ar[ld,C]   \\
    \zxGateMulti{3}{1}{?} &             \\[\zxZeroRow-2mm]
    \rar                  & \zxGround{} \\[\zxZeroRow-2mm]
    \ar[dr,C]             &             \\
                          & 
  \end{ZX} =
  \begin{ZX}[circuit]
                         & \ar[ld,C]   \\
    \zxGround{}          &             \\
    \zxGround{} \ar[dr,C] &             \\
                         & 
  \end{ZX} =
  \begin{ZX}[circuit]
    \zxGround-{} \rar & \\
    \zxGround-{} \rar &
  \end{ZX} \neq
  \begin{ZX}[circuit]
    \zxN{} \dar[C] \rar & \\
    \zxN{} \rar         & 
  \end{ZX}  
\end{align}
This is not surprising since, for instance, in quantum teleportation after the Bell measurement Bob needs to apply additional corrections to recover the original state. We formalise now this statement:

\textEnd{\subsection{Proofs of \cref{sec:CompSecurityofGraphState}}}

\begin{theoremE}[Impossibility of black-box realization of $\Verif$][end, restate]\label{thm:impossibilityBlackBox}
  There exists some\footnote{Actually most graph states have this property as soon as they are not separable.} graph states $\ket{G}$ such that for any graph state verification protocol $\Pi \eqdef \{\pi_i\}_{i \in [n] \cup \source}$ using a resource $\cR$ and producing $\ket{G}$ (i.e.\ $\{\pi_i\}_{i \in [n] \cup \source} \cR$ outputs $\ket{G}$ shared among the $n$ parties), it is \emph{impossible} to prove that $\Pi$ is $\eps$-realizing $\Verif$ (\cref{protoc:idealResourceVerif}) for any $\eps < 1/2$ if the simulator is black-box\footnote{We call it black-box since the definition of the simulator is mostly independent of the protocol, as it can only execute $\Pi$ without having access to its code. This definition of black-box simulator is relatively generic: if the simulator does not know the definition of the protocol $\Pi$, it can basically only forward the messages of the protocol to the distinguisher. For instance, if all exchanged messages are signed with a key unknown to the simulator (but which is part of the public description of the protocol), the only way to communicate with the distinguisher is to forward the messages sent by running $\{\pi_i\}_{i \in H} \cR$. We could formalize this by giving an even more generic definition of black-box simulator that is not explicitly asked to forward the messages of $\{\pi_i\}_{i \in H} \cR$ to the distinguisher, but this would obfuscate our proof without clear benefits.}, in the sense that the simulator interacts with the interface controlled by the environment by running $\{\pi_i\}_{i \in H} \cR$, forwarding all messages between the malicious interfaces of $\cR$ and the environment.
\end{theoremE}
\begin{proofE}
  This proof is diagrammatically illustrated with equations starting from \cref{eq:startDiagProofImpossibliliy}. In the following, we will consider the graph state $\ket{G} \eqdef \ket{00}+\ket{11}$, consisting of a Bell pair. By contradiction, we assume the existence of a protocol $\Pi \eqdef \{\pi_0, \pi_1, \pi_\source\}$ $\eps$-realizing $\Verif$, where the simulator $\sigma$ is black-box as defined in \cref{thm:impossibilityBlackBox}. In particular, we can consider the case where parties $0$ and $\source$ are honest and party $1$ is corrupted. We also consider the following distinguisher $\cD$:
  \begin{itemize}
  \item $\cD$ runs the honest protocol $\pi_1$, interacting with the corrupted interfaces. At the end of the protocol, it gets a bipartite quantum state $\rho^{0,1}$ from the output from $\pi_1$ and from the honest interface of the functionality $\Verif$.
  \item Then, $\cD$ will measure both qubits in the computational basis, outputting $0$ if both outcomes are equal and different from $\bot$ and $1$ otherwise.
  \end{itemize}

  We will show that this distinguisher can distinguish the real world from the ideal world with an advantage greater than $\eps$, raising a contradiction since both worlds must be indistinguishable against any distinguisher and any subset of corrupted parties.

  First, if he distinguisher is interacting with the real world $\pi_{\{0,\source\}} \cR$, because $\cD$ is following the honest protocol, we have $\rho = \{\pi_i\}_{i \in [n] \cup \source} \cR$. Since $\Pi$ is correct by assumption (recall that $\{\pi_i\}_{i \in [n] \cup \source} \cR$ outputs $\ket{G} = \ket{00}+\ket{11}$), we always have $\rho = \ket{00}+\ket{11}$, and therefore $\cD$ will always measure $2$ identical bits, hence always outputting $0$.

  On the other hand, let us consider the case where the distinguisher is interacting with the real world. Since the simulator is black-box according to definition given in \cref{thm:impossibilityBlackBox}, and since $\Pi$ expects no input from any party, we can assume, without loss of generality, that the simulator starts by running $\{\pi_i\}_{i \in \{0,\source\}} \cR$, forwarding all messages between the malicious interfaces of $\cR$ and the environment, and obtaining a state $\rho^0$ outputted by $\pi_0$. Moreover, since $\cD$ is honestly running $\pi_1$, we know, by the correctness of the protocol, that $\rho^0$ is half of a Bell-state shared with $\cD$. Similarly, since in the ideal world, the protocol never aborts, without loss of generality we can assume that the simulator also sends $c_1 = \top$ to $\Verif$ when starting (otherwise, we can always convert any simulator that sets $c_1 = \bot$ into a better simulator that sets $c_1 = \top$): the simulator will then receiving a state $\rho^{\sigma}$ from $\Verif$, where $\rho^\sigma$ is half of a Bell pair shared with the distinguisher. Then, we can call $U^\sigma(\rho^0,\ket{0^l},\rho^\sigma)$ the rest of the quantum map performed by the simulator after receiving $\rho^0$ and $\rho^\sigma$, where $\ket{0^l}$ is an arbitrary auxiliary register. Therefore, the global state obtained when the simulator has finished to run is:
  \begin{align}
    (I_2 \otimes U^\sigma \otimes I_2)(\ket{00}+\ket{11})\otimes \ket{0^l} \otimes (\ket{00}+\ket{11})
  \end{align}
  where $I_2$ is the identity acting on a single qubit. If we consider now the view of the distinguisher right before performing its measurement, we can trace out the map performed by the simulator:
  \begin{align}
    \Tr_\sigma(
    (I_2 \otimes U^\sigma \otimes I_2)(\ket{00}+\ket{11})\otimes \ket{0^l} \otimes (\ket{00}+\ket{11})
    (\bra{00}+\bra{11})\otimes \bra{0^l} \otimes (\bra{00}+\bra{11})(I_2 \otimes U^{\sigma\dagger} \otimes I_2)
    )
  \end{align}
  By the non-signaling principle, $U^\sigma$ cannot modify this state since it is traced out. As a consequence, this state is equal to:
  \begin{align}
    \Tr_\sigma(
    (I_2 \otimes I_{2^{l+2}} \otimes I_2)(\ket{00}+\ket{11})\otimes \ket{0^l} \otimes (\ket{00}+\ket{11})
    (\bra{00}+\bra{11})\otimes \bra{0^l} \otimes (\bra{00}+\bra{11})(I_2 \otimes I_{2^{l+2}} \otimes I_2)
    )
  \end{align}
  But it is easy to see that this state is equal to the identity density matrix of $2$ qubits, as we discard, twice, one share of a Bell pair. Hence, after measuring this state in the computational basis, the distinguisher will obtain $0$ with probability $1/2$ instead of $1$, so the advantage in distinguishing is $1/2 > \eps$. Therefore, $\Pi$ cannot $\eps$ realize $\Verif$, raising a contradiction.
\end{proofE}

We are therefore left with two options:
\begin{itemize}
\item either we change the functionality,
\item or we change the protocol.
\end{itemize}
We present both approaches in this article. We will focus in this section on the first approach, while the second approach will be seen in \cref{sec:withoutCorr}, building on the results introduced here. \\

The previous impossibility result suggests a first modification: Let the simulator be allowed to communicate the output $x$ of their measurements to the ideal functionality. This in turns let the functionality apply these corrections $\xi_H(x)$ to the honest part of the graph state. This brings us to the following picture, building on \cref{eq:finalSimplifiedGeqGG}:
\begin{align}
  \begin{ZX}<circuit,forward=\cR>[circuit, thick lines,zx column sep=6mm,mbr=2]
     & \zxGate{\ket{G}} \ar[rd,|--,-latex,start anchor={[xshift=1mm]south}, Bn.Args={M}{pos=.75},thick] \ar[ld,|--,-latex,start anchor={[xshift=-1mm]south},Bn.Args={H}{pos=.75},thick] & \\[\zxZeroRow+3mm]
     &                                                        & 
  \end{ZX} \approx_\eps
  \begin{ZX}<circuit,forward=\cV,forward=\Verif>[circuit, thick lines,zx column sep=6mm,mbr=2]
                                    & \zxContName{(xi)(G)}{below:$\Verif^f$}                               & \zxGate[a=G]{\ket{G}} \ar[rd,|--,-latex,start anchor={[xshift=1mm]south}, Bn'Args={M}{pos=.80},thick] \ar[ld,|--,-latex,start anchor={[xshift=-1mm]south},Bn'Args={H}{pos=.80},thick] &            & \zxGate{\ket{G}} \ar[rd,|--,-latex,start anchor={[xshift=1mm]south}, Bn.Args={M}{pos=.75},thick] \ar[ld,|--,-latex,start anchor={[xshift=-1mm]south},Bn.Args={H}{pos=.75},thick] & \\[\zxZeroRow-1mm]
  \rar[latex-,Bn'Args={H}{pos=.50},thick] & \zxGate[a=xi]{\xi_H} \ar[rr,latex-,cl,shift left=-2mm,"{x}" {below}] &                                                                                                                                                                                       & \zxGate{?} &                                                                                                                                                                                  & 
  \end{ZX}\label{eq:mergeableStateIntro}
\end{align}

\begin{remark}\label{rq:justificationF}  It is important to allow only certain, harmless, corrections, by letting the functionality check that the corrections $x$ are valid, after verifying that $f(x) = \top$. 
  Indeed, allowing arbitrary corrections would allow the adversary to perform attacks that might be impossible to perform with only access to $\ket{G}^M$. For instance, if we allow arbitrary corrections on a GHZ state $\frac{1}{\sqrt{2}}(\ket{0\dots0} + \ket{1 \dots 1})$, then nothing prevents an adversary from flipping the bit of the honest parties in a different way, for instance to produce $\frac{1}{\sqrt{2}}(\ket{1 0\dots0} + \ket{01 \dots 1})$, where the honest parties get the first two qubits. This would be disastrous, for instance if we build a coin tossing protocol from this GHZ state by measuring the first two qubits in the computational basis, since the first two parties would get different outcomes, which is impossible to obtain from a normal GHZ. Note that it might be hard to justify why a correction $x$ is harmless as it might depend on the protocol that we target. Protocols that expect a graph state $\ket{G}$ should in theory reprove that they are secure when using the functionality $\Verif^f$, or use, instead, the result that we present in \cref{sec:withoutCorr}. However, we believe that the set of corrections we obtain is harmless, since it is basically a subset of the stabilizers of the graph state. Therefore, intuitively, any correction $x$ such that $f(x) = \top$ can be done by the adversary by applying a stabilizer on his side, which will automatically propagate some corrections to the honest players. This property is actually formalized later in order to obtain our second result in \cref{sec:withoutCorr}.
\end{remark}

 Informally speaking, a state is \emph{mergeable} if the following holds: if this state is shared between two parties Alice and Bob, and if another copy of this state is shared between Bob and Charlie, then is is be possible to obtain a single copy of this state between Alice and Charlie, under the constraint that Charlie should not do any operation, that Bob should perform an arbitrary measurement to obtain an outcome $m$, and that Alice only does an operation that only depends on $m$. In a sense, this can be seen as a generalization of entanglement swapping to arbitrary states, with additional constraints on the set of allowed operations. Said differently, we will say that a state $\ket{G}$ is mergeable with respect to $f$ and $\xi_H$ if we can find $\zx[circuit]{\rar & \zxGate{?} \rar &}$ (called $\xi_\Sim$ from now) such that \cref{eq:mergeableStateIntro} is true, where $f(x) = \top$. Since this must be true for any set of corrupted party, the function $f$ and $\xi_H$ will be different for any subset $H$ of honest parties. More formally:

\begin{definition}[Mergeable states]\label{def:mergeable}
  A state $\ket{\psi}^{R_1,\dots,R_n}$ on $n$ registers $\{R_i\}_{i \in [n]}$ is said to be \emph{mergeable} with respect to a function $f \colon \cP([n]) \times \cM$ (taking as input a set of honest party and a measurement outcome), and a collection of quantum maps $\{\xi_{H}\}_{H \subseteq [n]}$ (taking as input the registers $R_{i}$ for $i \in H$, together with an additional classical input, and outputting the same $R_{i}$'s registers) if for any subset $H \subseteq [n]$ of registers, there exists a quantum map $\xi_\Sim$ (taking inputs from the registers $R_{\bar{H}}$ of a first state, and $R_{H}$ of a second state, and outputting a single classical value in $\cM$ such that $f(H, x) = \top$) and $\xi_H$ (taking the register $R_{H}$ of the first state and the classical output of $\xi_\Sim$, and outputting a quantum register), such that:
  \begin{align}
    (\xi_H \otimes I_{n-|H|})(I_{|H|} \otimes \xi_\Sim \otimes I_{n-|H|})(\ket{\psi}^{R_H,R_{\bar{H}}} \otimes \ket{\psi}^{R_H,R_{\bar{H}}}) = \ket{\psi}^{R_H,R_{\bar{H}}}\label{eq:mergeable}
  \end{align}
  Note that this is always trivially possible if $H$ is empty or equal to $[n]$. In picture: 
  \begin{align}
    \begin{ZX}<add to preamble={\usepackage{braket}}>[circuit,thick lines]
      \zxInputMulti{2}{\ket{\psi}^{R_H,R_{\bar{H}}}} \ar[to=xiH] &                                                 & \zxGate[a=xiH]{\xi_H} \rar & \\
      \rar                                                       & \zxGateMulti{2}{1}{\xi_\Sim} \ar[to=xiH,cl,--|] &                         \\
      \zxInputMulti{2}{\ket{\psi}^{R_H,R_{\bar{H}}}} \rar        &                                          &                       & \\
                                                      \ar[rrr]   &                                                 &                       & 
    \end{ZX}
    = \ket{\psi}^{R_H,R_{\bar{H}}}
\end{align}
\end{definition}

We show now that any graph state is mergeable. Note that during a first reading, it might help to start with the simpler construction of \cref{cor:ghzstatesAreMergeable} that focuses only on \GHZ{} states.

\begin{theoremE}[Any graph state is mergeable][end, restate]\label{thm:anygraphismergeable}
  For any graph $G = (V,E)$, $\ket{G}$ is mergeable (\cref{def:mergeable}) with respect to the maps $\{\xi_{H}\}_{H \subseteq [n]}$ that take two lists of X and Z corrections $(x,z) \in (\Z_2^{|H|})^2$ and applies $\gateX^x \gateZ^z$ on the input qubits.

  The merge procedure is described diagrammatically in \cref{fig:allStatesMergeable}, and for completeness we reformulate it here. Let $n = |V|$, $H$ and $M$ be any partition of $V$. For simplicity, we assume that we reorder elements of $V$ to have elements of $H$ ordered before elements of $M$. Let $\Gamma$ be the biadjency graph between $H$ and $M$ (cf.\ \cref{def:notationGraphs}). Then, we define $\xi_\Sim$ as follows (cf.\ illustration \cref{fig:allStatesMergeable}), where the $i$-th qubit of $\ket{G}$ belongs to register $H$ (resp.\ $M$) iff $i \in H$ (resp.\ $M$).
  \begin{itemize}
  \item It applies $\gateCZ$ gates on any pair $(i,j)$ of qubits of register $H$ iff $(i,j) \in G_H$ and similarly it applies $\gateCZ$ gates on any pair $(i,j)$ of qubits of register $M$ iff $(i,j) \in G_M$.
  \item It applies Hadamard gates on all qubits of register $M$.
  \item It computes $U$, $V$, $r$ and $R$ according to \cref{lem:GaussPivot} and applies the unitary $\ket{x} \mapsto \ket{V^{-1}x}$ on register $M$ and $\ket{x} \mapsto \ket{Ux}$ on register $H$. This is always possible since $U$ and $V$ are invertible. We propose moreover in \cref{lem:physicallyImplementable} a way to implement them more efficiently, without auxiliary qubits and using only $\gateCNOT$ and swap operations.
  \item It performs $r$ Bell measurements (projection on one of the four Bell states) between the first $r$ qubits of each register. The Bell measurements are between the $i$-th qubit of register $M$ with the $i$-th qubit of register $H$, where a measurement outcome $(b_i, c_i) \in \{0,1\}^2$ means that the $i$-th pair was projected on the Bell state $\ket{0c_i}+(-1)^{b_i}\ket{1\bar{c_i}}$. The outcomes are gathered into two vector $b = (b_i)_{i \in [r]}$ and $c = (c_i)_{i \in [r]}$.
  \item It performs a measurement in the $\{H\ket{a_i}\}_{a_i \in \{0,1\}}$ basis on the $|M|-r$ remaining qubits of register $M$ (the outcomes are gathered into a vector $a$), and a measurement in the computational basis $\{\ket{d_i}\}_{d_i \in \{0,1\}}$ on the $|H|-r$ remaining qubits of register $H$ (the outcomes are gathered into a vector $d$).
  \item It computes $x \eqdef U^{-1}
    \begin{bmatrix}
      c \xor Rd\\
      \textbf{0}
    \end{bmatrix}$, and $z \eqdef \left(U^T
      \begin{bmatrix}
        b\\
        R^Tb
      \end{bmatrix}\right) \xor Gx$ where $Gx$ is the set of neighbours of $x$ as defined in \cref{lem:stabilizerDiag}, and returns the corrections $(x,z)$.
  \end{itemize}
\end{theoremE}

\begin{figure}
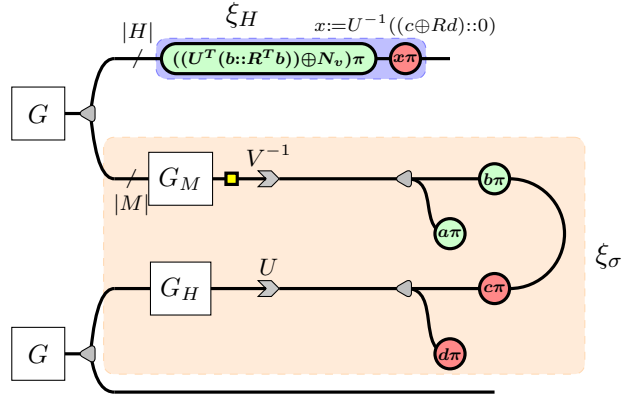

  \centering
  \begin{ZX}[myBold, execute at end picture={
      \zxNamedBox[opacity=.3,fit margins={right=5mm,top=1mm,left=3mm}][orange]{(GM)(measD)(bellA1)}{right:$\xi_\Sim$}
      \zxNamedBox[]{(endA)(xpi)}{above left:$\xi_H$}
    }]
                                                                                                                                                                                                                                                                                                                                                                    \\[3mm]
                      &                                         & [\zxwCol] \zxN{} \ar[rrr,Bn'=|H|] & [\zxWCol+1mm]           &             & [-12mm] \zxZ[a=endA]{((U^T(b{::}R^Tb)) \xor N_v)\pi} \rar & \zxX[a=xpi,label={[overlay,yshift=1mm]\scriptstyle x \eqdef U^{-1}((c \xor Rd){::}0)}]*{x\pi} \rar &                                                                \\
    \zxBox[fill=white]{G} \rar[B] & \zxDivider[BBw]{} \ar[ru,<'] \ar[rd,<.]                                                                                                                                                                                                                                                                                                     \\
                      &                                         & \zxN{} \rar[,Bn.=|M|]             & \zxBox[fill=white,a=GM]{G_M} \rar    & \zxH{} \rar & \zxMatrix/{V^{-1}} \rar                                   & \zxDivider[a=divA]{} \ar[to=measA,s] \ar[to=bellA1]                                          &                      & \zxZ[a=bellA1]*{b\pi} \ar[to=bellA2,C-] \\
                      &                                         &                                   &                     &             &                                                           &                                                                                              & \zxZ[a=measA]*{a\pi} &                                         \\
                      &                                         & \zxN{} \rar                       & \zxBox[fill=white]{G_H} \ar[rr] &             & \zxMatrix/{U} \rar                                        & \zxDivider[a=divB]{} \ar[to=measD,s] \ar[to=bellA2]                                          &                      & \zxX[a=bellA2]*{c\pi}                   \\
    \zxBox{G} \rar[B] & \zxDivider[BBw]{} \ar[ru,<'] \ar[rd,<.] &                                   &                     &             &                                                           &                                                                                              & \zxX[a=measD]*{d\pi}                                           \\
                      &                                         & \zxN{} \ar[to=endB]               &                     &             &                                                           &                                                                                              &                      & \zxN[a=endB]{}
  \end{ZX}  
  \caption{Representation in ZX-calculus of the procedure to merge two copies of a graph state into a single copy. The orange area denotes the merging map $\xi_\sigma$ while the blue one represents the corrections $\xi_H$ to apply. Note that $H$ and $M$ form a partition of the vertices of $G$, and that $G_H$, $G_M$, $\Gamma$, $U$, $V$ and $R$ are defined like in \cref{def:notationGraphs} and \cref{lem:GaussPivot}.}
  \label{fig:allStatesMergeable}
\end{figure}

\begin{proofE}[text link={Note that this lemma is at the heart of our construction. Since the proof is quite technical and heavily relies on scalable ZX-calculs, we defer the \hyperref[proof:prAtEnd\pratendcountercurrent]{full proof} to \pratendSectionlikeCref.}]
  First, we can decompose the graphs given as input to the merge procedure given in \cref{fig:allStatesMergeable} as described in \cref{lem:graphicalCutG} (note that the outcome state might depend on the measurement outcomes $a,b,c,d$, hence the notation):
  \begin{align}
    \ket{\psi_{a,b,c,d}} \eqdef
    \begin{ZX}[myBold]
                                                                                                                                                                                                                                                                                                                         \\[3mm]
      \zxBox{G_H} \rar & \zxZ{} \dar \ar[to=endA]   &                     &             & \zxZ[a=endA,pilb={overlay,yshift=-1mm,font=\fontsize{6}{10}\selectfont\boldmath}]{((U^T(b{::}R^Tb)) \xor Gx)\pi} \rar & \zxX[label={[overlay,yshift=1mm]\scriptstyle x \eqdef U^{-1}((c \xor Rd){::}0)}]*{x\pi} \rar &                                                                \\
                       & \zxMatrix'{\Gamma} \dar[H]                                                                                                                                                                                                                                                                      \\
      \zxBox{G_M} \rar & \zxZ{} \rar                & \zxBox{G_M} \rar    & \zxH{} \rar & \zxMatrix/{V^{-1}} \rar                                        & \zxDivider[a=divA]{} \ar[to=measA,s] \ar[to=bellA1]                                          &                      & \zxZ[a=bellA1]*{b\pi} \ar[to=bellA2,C-] \\
                       &                            &                     &             &                                                                &                                                                                              & \zxZ[a=measA]*{a\pi} &                                         \\
      \zxBox{G_H} \rar & \zxZ{} \dar \rar           & \zxBox{G_H} \ar[rr] &             & \zxMatrix/{U} \rar                                             & \zxDivider[a=divB]{} \ar[to=measD,s] \ar[to=bellA2]                                          &                      & \zxX[a=bellA2]*{c\pi}                   \\
                       & \zxMatrix'{\Gamma} \dar[H] &                     &             &                                                                &                                                                                              & \zxX[a=measD]*{d\pi}                                           \\
      \zxBox{G_M} \rar & \zxZ{} \ar[to=endB]        &                     &             &                                                                &                                                                                              &                      & \zxN[a=endB]{}
    \end{ZX}    
  \end{align}
  Then, we can see that:
  \begin{align}
    \begin{ZX}[myBold]
                       &                       &                &[\zxwCol] \\
      \zxBox{G_M} \rar & \zxZ[B]{} \rar \ar[u] & \zxBox{G_M} \rar & 
    \end{ZX} \eqeqLem{lem:GCommutes}
    \begin{ZX}[myBold]
                       &                  &                       & [\zxwCol] \\
      \zxBox{G_M} \rar & \zxBox{G_M} \rar & \zxZ[B]{} \rar \ar[u] & 
    \end{ZX} \eqeqDef{def:graphDiag}
    \begin{ZX}[myBold]
                     &                  &                  &                       & [\zxwCol] \\
      \zxZ[B]{} \rar & \zxBox{G_M} \rar & \zxBox{G_M} \rar & \zxZ[B]{} \rar \ar[u] & 
    \end{ZX} \eqeqLem{lem:GDestroys}
    \begin{ZX}[myBold]
                     &                       & [\zxwCol] \\
      \zxZ[B]{} \rar & \zxZ[B]{} \rar \ar[u] & 
    \end{ZX} \eqZXX{S}{I}
    \begin{ZX}[myBold]
      \zxN{} \ar[rd,<.={1-=0,1L=.7,2-=.7,2L=0}] & [\zxZeroCol+4mm] \\[\zxZeroRow+4mm]
                        & \zxN{}
    \end{ZX}\label{eq:fusionG}
  \end{align}
  The same is of course true for $G_H$ instead of $G_M$, we can therefore simplify the above graph as follows:
  \begin{align}
    \ket{\psi_{a,b,c,d}}
    &\eqeq{eq:fusionG}
      \begin{ZX}[myBold,
        execute at end picture={
          \zxNamedBox[fit margins={all=.5mm,top=2mm}][blue]{(gamma2)(measD)}{}
        },
        ]
                                                                                                                                                                                                                                                                                                                                                                                                                                                     \\[3mm]
      \zxBox{G_H} \rar & \zxZ{} \ar[to=Gamma1,<.] \ar[to=endA]   &                                               & [\zxHCol+3mm]                            & \zxZ[a=endA,pilb={overlay,yshift=-1mm,font=\fontsize{6}{10}\selectfont\boldmath}]{((U^T(b{::}R^Tb)) \xor Gx)\pi} \rar & \zxX[label={[overlay,yshift=1mm]\scriptstyle x \eqdef U^{-1}((c \xor Rd){::}0)}]*{x\pi} \rar &                                                                \\[2mm]
                       &                                         & \zxMatrix[a=Gamma1]{\Gamma} \ar[H,to=topH,<.] & \zxH[a=topH]{} \rar                      & \zxMatrix/{V^{-1}} \rar                                                                                                & \zxDivider[a=divA]{} \ar[to=measA,s] \ar[to=bellA1]                                          &                      & \zxZ[a=bellA1]*{b\pi} \ar[to=bellA2,C-] \\
                       &                                         &                                               &                                          &                                                                                                                        &                                                                                              & \zxZ[a=measA]*{a\pi} &                                         \\
                       &                                         &                                               & \zxMatrix[a=gamma2]*{\Gamma} \ar[to=myU] & \zxMatrix[a=myU]/{U} \rar                                                                                              & \zxDivider[a=divB]{} \ar[to=measD,s] \ar[to=bellA2]                                          &                      & \zxX[a=bellA2]*{c\pi}                   \\
                       &                                         &                                               &                                          &                                                                                                                        &                                                                                              & \zxX[a=measD]*{d\pi}                                           \\
      \zxBox{G_M} \rar & \zxZ{} \ar[to=gamma2,<',H] \ar[to=endB] &                                               &                                          &                                                                                                                        &                                                                                              &                      & \zxN[a=endB]{}
    \end{ZX}\label{eg:zxBeforeUnfolding}
  \end{align}
  We will first simplify the \tikz[baseline=(a.base)]\node[rounded corners,inner sep=2pt,fill=blue!50!white,opacity=.5,text opacity=1](a){colored part};, but we prove before an equivalent form of $\Gamma$ that will also be useful later: Let $U$, $V$, $r$ and $R$ be like in \cref{lem:GaussPivot}. We claim that:
  \begin{align}
    \begin{ZX}[mbr,myBold]
      \zxNL \zxMatrix{\Gamma} \zxNR
    \end{ZX} =
    \begin{ZX}[mbr=2,myBold]
           &                             &                                                 & \zxN[a=midNode]{} \ar[to=Zmerge,'>]    &                                  &                            &                    & [\zxWCol] \\
      \rar & [\zxwCol] \zxMatrix{U} \rar & \zxDivider{} \ar[to=MatA,<.] \ar[to=midNode,<'] & \zxMatrix[a=MatA]{R} \ar[to=Zmerge] & \zxX[a=Zmerge]{} \ar[to=divM,'>] & \zxDivider[a=divM]-{} \rar & \zxMatrix/{V} \rar &           \\
           &                             &                                                 & \zxX{} \ar[to=divM,.>]
    \end{ZX}\label{eq:formGamma}
  \end{align}
  This is easy to see by mechanically using the diagrammatic representation of block, identity, and zero matrices:
  \begin{align}
    \begin{ZX}[mbr=1,myBold]
      \zxNL \zxMatrix{\Gamma} \zxNR
    \end{ZX} &\eqeqLem{lem:GaussPivot}
    \begin{ZX}[mbr=1,myBold]
      \zxNL \zxMatrix{V
        \begin{bmatrix}
          I_r & R\\
          \mathbf{0} & \mathbf{0}
        \end{bmatrix}
        U} \zxNR
    \end{ZX} \eqSubeq{eq:matrixManipulationB}{a}
    \begin{ZX}[mbr=1,myBold]
      \zxNL \zxMatrix{U} \rar & \zxMatrix{\begin{bmatrix}
          I_r & R\\
          \mathbf{0} & \mathbf{0}
        \end{bmatrix}} \rar & \zxMatrix{V} \zxNR
    \end{ZX} \eqSubeq{eq:matrixManipulationA}{a}
    \begin{ZX}[mbr=2,myBold]
           & [\zxwCol]         &                                        & \zxMatrix[a=topM]{\makebox[4mm][c]{$\begin{bmatrix}I_r & R\end{bmatrix}$}} \ar[to=divM,'>]  &                    & [\zxWCol] \\
      \rar & \zxMatrix/{U} \rar & \zxZ{} \ar[to=topM,<'] \ar[to=botM,<.] &                                      & \zxDivider[a=divM]-{} \rar & \zxMatrix/{V} \rar &           \\
           &                   &                                        & \zxMatrix[a=botM]/{\mathbf{0}} \ar[to=divM,.>]
    \end{ZX} \eqSubeq{eq:defMatrix}{b}
    \begin{ZX}[mbr=2,myBold]
           & [\zxwCol]          &                                        &                & \zxMatrix[a=topM]{\makebox[0mm][c]{$\begin{bmatrix}I_r & R\end{bmatrix}$}} \ar[to=divM,'>] &                            &                    & [\zxWCol] \\
      \rar & \zxMatrix/{U} \rar & \zxZ{} \ar[to=topM,<'] \ar[to=botM,<.] &                &                                                        &                                   & \zxDivider[a=divM]-{} \rar & \zxMatrix/{V} \rar &           \\
           &                    &                                        & \zxZ[a=botM]{} &                                                        & \zxX{} \ar[to=divM,.>]
    \end{ZX}\\
    & \eqZXX{S}{I}
    \begin{ZX}[mbr=2,myBold]
      \rar & [\zxwCol] \zxMatrix/{U} \rar & \zxMatrix[a=topM]{\begin{bmatrix}I_r & R\end{bmatrix}} \ar[to=divM,'>] &                    &  & [\zxWCol] \\
           &                              &                                      & \zxDivider[a=divM]-{} \rar      & \zxMatrix/{V} \rar &              \\
           &                              & \zxX{} \ar[to=divM,.>]
    \end{ZX} \eqSubeq{eq:matrixManipulationA}{b}
    \begin{ZX}[mbr=2,myBold]
      \rar & [\zxwCol] \zxMatrix{U} \rar & \zxDivider{} \ar[to=MatA,<.] \rar & [\zxwCol] \zxMatrix{I_r} \ar[to=Zmerge,'>] &                                  &                            &                    & [\zxWCol] \\
           &                              &                                   & \zxMatrix[a=MatA]{R} \ar[to=Zmerge]        & \zxX[a=Zmerge]{} \ar[to=divM,'>] & \zxDivider[a=divM]-{} \rar & \zxMatrix/{V} \rar &           \\
           &                              &                                   & \zxX{} \ar[to=divM,.>]
    \end{ZX} \eqSubeq{eq:defMatrix}{a}
    \begin{ZX}[mbr=2,myBold]
           &                             &                                                 & \zxN[a=midNode]{} \ar[to=Zmerge,'>]    &                                  &                            &                    & [\zxWCol] \\
      \rar & [\zxwCol] \zxMatrix{U} \rar & \zxDivider{} \ar[to=MatA,<.] \ar[to=midNode,<'] & \zxMatrix[a=MatA]{R} \ar[to=Zmerge] & \zxX[a=Zmerge]{} \ar[to=divM,'>] & \zxDivider[a=divM]-{} \rar & \zxMatrix/{V} \rar &           \\
           &                             &                                                 & \zxX{} \ar[to=divM,.>]
    \end{ZX}\label{eq:defGamma}
  \end{align}
  Therefore, we have:
    \begin{align}
      \begin{ZX}[mbr,myBold,
        execute at end picture={
          \zxNamedBox[fit margins={all=.5mm,top=2mm}][blue]{(gamma2)(measD)}{}
        },
        ]
        \zxNL \zxMatrix[a=gamma2]*{\Gamma} \rar & \zxMatrix{U} \rar & \zxDivider[a=divB]{} \ar[to=measD,s] \ar[r] & \zxN{} \\
                                                &                   &                                             &  & \zxX[a=measD]*{d\pi}
      \end{ZX}  &\eqeq{eq:formGamma}
    \begin{ZX}[mbr=2,myBold]
           &                              &                                         &                                                  & \zxN[a=midNode]{} \ar[to=Zmerge,'>]  &                              &                    &                   & [\zxWCol]                                            \\
      \rar & [\zxwCol] \zxMatrix*{V} \rar & \zxDivider{} \ar[<.,to=myX] \ar[to=myZ] & \zxX[a=myZ]{} \ar[to=MatA,<.] \ar[to=midNode,<'] & \zxMatrix[a=MatA]*{R} \ar[to=Zmerge] & \zxDivider[a=Zmerge]-{} \rar & \zxMatrix*{U} \rar & \zxMatrix{U} \rar & \zxDivider[a=divB]{} \ar[to=measD,s] \ar[r] & \zxN{} \\
           &                              &                                         &                                                  & \zxX[a=myX]{}                        &                              &                    &                   &                                             & \zxX[a=measD]*{d\pi}
    \end{ZX} \eqeqLem{lem:injectiveSurjective}
    \begin{ZX}[mbr=2,myBold]
           &                              &                                         &                                                  & \zxN[a=midNode]{} \ar[to=Zmerge,'>]  &                              &                                                      \\
      \rar & [\zxwCol] \zxMatrix*{V} \rar & \zxDivider{} \ar[<.,to=myX] \ar[to=myZ] & \zxX[a=myZ]{} \ar[to=MatA,<.] \ar[to=midNode,<'] & \zxMatrix[a=MatA]*{R} \ar[to=Zmerge] & \zxDivider[a=Zmerge]-{} \rar & \zxDivider[a=divB]{} \ar[to=measD,s] \ar[r] & \zxN{} \\
           &                              &                                         &                                                  & \zxX[a=myX]{}                        &                              &                                             & \zxX[a=measD]*{d\pi}
    \end{ZX} \eqZX{R}
    \begin{ZX}[mbr=2,myBold]
           &                              &                                         &                                                  & \zxN[a=midNode]{}                   &        \\
      \rar & [\zxwCol] \zxMatrix*{V} \rar & \zxDivider{} \ar[<.,to=myX] \ar[to=myZ] & \zxX[a=myZ]{} \ar[to=MatA,<.] \ar[to=midNode,<'] & \zxMatrix[a=MatA]*{R} \ar[to=measD] & \zxX[a=measD]*{d\pi} \\
           &                              &                                         &                                                  & \zxX[a=myX]{}                       & 
    \end{ZX}\\
    &\eqSubeq{eq:matrixManipulationC}{b}
    \begin{ZX}[mbr=1,myBold]
      \rar & [\zxwCol] \zxMatrix*{V} \rar & \zxDivider{} \ar[<.,to=myX] \ar[to=myZ] & \zxX[a=myZ]{Rd\pi} \ar[r] & \zxN{} \\
           &                              &                                         &                           & \zxX[a=myX]{} & 
    \end{ZX}
  \end{align}

  By injecting this into the \tikz[baseline=(a.base)]\node[rounded corners,inner sep=2pt,fill=blue!50!white,opacity=.5,text opacity=1](a){colored part}; of \cref{eg:zxBeforeUnfolding}, we get:
  \begin{align}
    \ket{\psi_{a,b,c,d}}
    &\eqeq{eq:fusionG}
      \begin{ZX}[myBold,
        execute at end picture={
          \zxNamedBox[fit margins={all=.5mm,top=1.2mm,left=1mm}][yellow]{(Gamma1)(measA)(bellA2)}{}
        },
        ]
                                                                                                                                                                                                                                                                                                                                                                                                                                                                  \\[3mm]
      \zxBox{G_H} \rar & \zxZ{} \ar[to=Gamma1,<.] \ar[to=endA]   &                                               & [\zxHCol+3mm]                           & \zxZ[a=endA,pilb={overlay,yshift=-1mm,font=\fontsize{6}{10}\selectfont\boldmath}]{((U^T(b{::}R^Tb)) \xor Gx)\pi} \rar & \zxX[label={[overlay,yshift=1mm]\scriptstyle x \eqdef U^{-1}((c \xor Rd){::}0)}]*{x\pi} \rar &                                                                              \\[2mm]
                       &                                         & \zxMatrix[a=Gamma1]{\Gamma} \ar[H,to=topH,<.] & \zxH[a=topH]{} \rar                     & \zxMatrix/{V^{-1}} \rar                                                                                                & \zxDivider[a=divA]{} \ar[to=measA,s] \ar[to=bellA1]                                          & \zxZ[a=bellA1]*{b\pi} \ar[to=bellA2] & \zxX[a=bellA2]*{(c \xor Rd)\pi} \\
                       &                                         &                                               &                                         &                                                                                                                        &                                                                                              & \zxZ[a=measA]*{a\pi}                    &                                    \\
                       &                                         &                                               & \zxMatrix[a=gamma2]/*{V} \ar[to=divB,wc] & \zxDivider[a=divB]{} \ar[to=measD,s] \ar[to=bellA2,C-]                                                                 &                                                                                              &                                                                              \\
                       &                                         &                                               &                                         &                                                                                                                        &                                                                                              & \zxX[a=measD]*{}                                                             \\
      \zxBox{G_M} \rar & \zxZ{} \ar[to=gamma2,<',H] \ar[to=endB] &                                               &                                         &                                                                                                                        &                                                                                              &                                         & \zxN[a=endB]{}
    \end{ZX}\label{eq:simplifiedLowerPart}
  \end{align}
  We focus now on the \tikz[baseline=(a.base)]\node[rounded corners,inner sep=2pt,fill=yellow!50!white,opacity=.5,text opacity=1](a){yellow-colored}; part of this diagram. First, we can use the well-known equality $\zx{\rar[B] &[\zxwCol] \zxH[B]{} \rar[B] & \zxH[B]{} \rar[B] &[\zxwCol]} = \zx{\rar[B] &[\zxWCol]}$ that is easy to prove as follows:
  \begin{align}
    \begin{ZX}
      \rar[B] &[\zxwCol] \zxH[B]{} \rar[B] & \zxH[B]{} \rar[B] &[\zxwCol]
    \end{ZX} \eqZX{I}
    \begin{ZX}
      \rar[B] &[\zxwCol] \zxH[B]{} \rar[B] & \zxZ[B]{} \rar[B] & \zxH[B]{} \rar[B] &[\zxwCol]
    \end{ZX} \eqZX{H}
    \begin{ZX}
      \rar[B] & [\zxwCol] \zxX[B]{} \rar[B] & [\zxwCol]
    \end{ZX} \eqZX{I}
    \begin{ZX}
      \rar[B] &[\zxWCol]
    \end{ZX}\label{eq:HH}
  \end{align}
  which gives:
  \begin{align}
    &\begin{ZX}[mbr=1,myBold,
      execute at end picture={
        \zxNamedBox[fit margins={all=.5mm,top=1.2mm,left=1mm}][yellow]{(Gamma1)(measA)(bellA2)}{}
      },
      ]
      \zxN{} \rar & [\zxwCol] \zxMatrix[a=Gamma1]{\Gamma} \ar[H,to=topH,<.] & [\zxHCol+2mm] \zxH[a=topH]{} \rar & \zxMatrix/{V^{-1}} \rar & \zxDivider[a=divA]{} \ar[to=measA,s] \ar[to=bellA1] & \zxZ[a=bellA1]*{b\pi} \ar[to=bellA2] & \zxX[a=bellA2]*{(c \xor Rd)\pi} \zxNR\\ 
                  &                                                         &                                   &                         &                                                     & \zxZ[a=measA]*{a\pi}                 & 
    \end{ZX} \eqeq{eq:HH}
    % We remove the --H--H--
    \begin{ZX}[mbr=1,myBold,
      execute at end picture={
        %\zxNamedBox[fit margins={all=.5mm,top=1.2mm,left=1mm}][yellow]{(Gamma1)(measA)(bellA2)}{}
      },
      ]
      \zxN{} \rar & [\zxwCol] \zxMatrix[a=Gamma1]{\Gamma}  \rar & \zxMatrix/{V^{-1}} \rar & \zxDivider[a=divA]{} \ar[to=measA,s] \ar[to=bellA1] &[\zxwCol] \zxZ[a=bellA1]*{b\pi} \ar[to=bellA2] & \zxX[a=bellA2]*{(c \xor Rd)\pi} \zxNR\\ 
                  &                                             &                         &                                                     & \zxZ[a=measA]*{a\pi}                 & 
    \end{ZX} \eqeq{eq:defGamma}
    \begin{ZX}[mbr=2,myBold,
      execute at end picture={
        %\zxNamedBox[fit margins={all=.5mm,top=1.2mm,left=1mm}][yellow]{(Gamma1)(measA)(bellA2)}{}
      },
      ]
           &                             &                                                 & \zxN[a=midNode]{} \ar[to=Zmerge,'>] &                                  &                            &                    & [\zxwCol]               &                                                     &[\zxwCol]                                                                                 \\
      \rar & [\zxwCol] \zxMatrix{U} \rar & \zxDivider{} \ar[to=MatA,<.] \ar[to=midNode,<'] & \zxMatrix[a=MatA]{R} \ar[to=Zmerge] & \zxX[a=Zmerge]{} \ar[to=divM,'>] & \zxDivider[a=divM]-{} \rar & \zxMatrix/{V} \rar & \zxMatrix/{V^{-1}} \rar & \zxDivider[a=divA]{} \ar[to=measA,s] \ar[to=bellA1] & \zxZ[a=bellA1]*{b\pi} \ar[to=bellA2] & \zxX[a=bellA2]*{(c \xor Rd)\pi} \zxNR \\         
           &                             &                                                 & \zxX{} \ar[to=divM,.>]              &                                  &                            &                    &                         &                                                     & \zxZ[a=measA]*{a\pi}                 &                                          \\
    \end{ZX}\nonumber                                                                                                                                                                                                                                                                                                                                                                   \\
    & \eqSubeq{eq:matrixManipulationB}{a}
    \begin{ZX}[mbr=2,myBold,
      execute at end picture={
        %\zxNamedBox[fit margins={all=.5mm,top=1.2mm,left=1mm}][yellow]{(Gamma1)(measA)(bellA2)}{}
      },
      ]
           &                             &                                                 & \zxN[a=midNode]{} \ar[to=Zmerge,'>] &                                  &                                                                                                                                                                    \\
      \rar & [\zxwCol] \zxMatrix{U} \rar & \zxDivider{} \ar[to=MatA,<.] \ar[to=midNode,<'] & \zxMatrix[a=MatA]{R} \ar[to=Zmerge] & \zxX[a=Zmerge]{} \ar[to=divM,'>] & \zxDivider[a=divM]-{} \rar & \zxDivider[a=divA]{} \ar[to=measA,s] \ar[to=bellA1] &[\zxwCol] \zxZ[a=bellA1]*{b\pi} \ar[to=bellA2] & \zxX[a=bellA2]*{(c \xor Rd)\pi} \zxNR \\         
           &                             &                                                 & \zxX{} \ar[to=divM,.>]              &                                  &                            &                                                     & \zxZ[a=measA]*{a\pi} & 
    \end{ZX} \eqZX{R}
    \begin{ZX}[mbr=2,myBold,
      execute at end picture={
        % \zxNamedBox[fit margins={all=.5mm,top=1.2mm,left=1mm}][yellow]{(Gamma1)(measA)(bellA2)}{}
        \node[right=5mm of measA,align=center,font=\tiny](scal){Scalar we\\can remove};
        \draw[->] (scal.west) to ([xshift=1mm]measA.east);
      },
      ]
           &                             &                                                 & \zxN[a=midNode]{} \ar[to=Zmerge,'>] &                                  &                                                                                 \\
      \rar & [\zxwCol] \zxMatrix{U} \rar & \zxDivider{} \ar[to=MatA,<.] \ar[to=midNode,<'] & \zxMatrix[a=MatA]{R} \ar[to=Zmerge] & \zxX[a=Zmerge]{}  \ar[to=bellA1] & \zxZ[a=bellA1]*{b\pi} \ar[to=bellA2] & \zxX[a=bellA2]*{(c \xor Rd)\pi} \zxNR \\         
           &                             &                                                 &                                     & \zxX{} \ar[to=measA,s]           & \zxZ[a=measA]*{a\pi}                 & 
    \end{ZX} \eqZX{Pi}
    \begin{ZX}[mbr=2,myBold,
      execute at end picture={
        % \zxNamedBox[fit margins={all=.5mm,top=1.2mm,left=1mm}][yellow]{(Gamma1)(measA)(bellA2)}{}
      },
      ]
           &                             &                                                 & \zxZ[a=midNode]*{b\pi} \ar[to=Zmerge,'>] &                                      &                                   &                                          \\
      \rar & [\zxwCol] \zxMatrix{U} \rar & \zxDivider{} \ar[to=MatA,<.] \ar[to=midNode,<'] & \zxMatrix[a=MatA]{R} \ar[to=bellA1]      & \zxZ[a=bellA1]*{b\pi} \ar[to=Zmerge] & \zxX[a=Zmerge]{}   \ar[to=bellA2] & \zxX[a=bellA2]*{(c \xor Rd)\pi} \zxNR \\         
    \end{ZX} \nonumber\\
    &\eqSubeq{eq:matrixManipulationB}{b}
    % Commute green and R
    \begin{ZX}[mbr=2,myBold,
      execute at end picture={
        % \zxNamedBox[fit margins={all=.5mm,top=1.2mm,left=1mm}][yellow]{(Gamma1)(measA)(bellA2)}{}
      },
      ]
           &                             &                                                   & \zxZ[a=midNode]*{b\pi} \ar[to=Zmerge,'>] &                                      &                                   &                                          \\
      \rar & [\zxwCol] \zxMatrix{U} \rar & \zxDivider{} \ar[to=bellA1,<.] \ar[to=midNode,<'] & \zxZ[a=bellA1]*{R^Tb\pi} \ar[to=MatA]       & \zxMatrix[a=MatA]{R}  \ar[to=Zmerge] & \zxX[a=Zmerge]{}   \ar[to=bellA2] & \zxX[a=bellA2]*{(c \xor Rd)\pi} \zxNR \\         
    \end{ZX} \eqZX{Z}
    % Merge green nodes
    \begin{ZX}[mbr=2,myBold,
      execute at end picture={
        % \zxNamedBox[fit margins={all=.5mm,top=1.2mm,left=1mm}][yellow]{(Gamma1)(measA)(bellA2)}{}
      },
      ]
           &                             &                                         &                                                                & \zxN[a=midNode]{} \ar[to=Zmerge,'>]  &                                   &  &                                       \\
      \rar & [\zxwCol] \zxMatrix{U} \rar & \zxZ[a=bellA1]{(b {::} R^Tb)\pi} \rar & \zxDivider{} \ar[to=bellA1,<.] \ar[to=midNode,<'] \ar[to=MatA] & \zxMatrix[a=MatA]{R}  \ar[to=Zmerge] & \zxX[a=Zmerge]{}   \ar[to=bellA2] & \zxX[a=bellA2]*{(c \xor Rd)\pi} \zxNR \\         
    \end{ZX} \nonumber\\&\eqSubeq{eq:matrixManipulationB}{b}
    % green other side of U
    \begin{ZX}[mbr=2,myBold,
      execute at end picture={
        % \zxNamedBox[fit margins={all=.5mm,top=1.2mm,left=1mm}][yellow]{(Gamma1)(measA)(bellA2)}{}
      },
      ]
           &                                        &                             &                                                                & \zxN[a=midNode]{} \ar[to=Zmerge,'>]  &                                   &  &                                       \\
      \rar & \zxZ[a=bellA1]{U^T(b {::} R^Tb)\pi} & [\zxwCol] \zxMatrix{U} \rar & \zxDivider{} \ar[to=bellA1,<.] \ar[to=midNode,<'] \ar[to=MatA] & \zxMatrix[a=MatA]{R}  \ar[to=Zmerge] & \zxX[a=Zmerge]{}   \ar[to=bellA2] & \zxX[a=bellA2]*{(c \xor Rd)\pi} \zxNR \\         
    \end{ZX}\eqZXX{S}{I}
    % Start to move red node
    \begin{ZX}[mbr=2,myBold,
      execute at end picture={
        % \zxNamedBox[fit margins={all=.5mm,top=1.2mm,left=1mm}][yellow]{(Gamma1)(measA)(bellA2)}{}
      },
      ]
           &                                        &                             &                                                                & \zxX[a=midNode]{(c \xor Rd)\pi} \ar[to=Zmerge,'>]  &                       &[\zxwCol] \\
      \rar & \zxZ[a=bellA1]{U^T(b {::} R^Tb)\pi} & [\zxwCol] \zxMatrix{U} \rar & \zxDivider{} \ar[to=bellA1,<.] \ar[to=midNode,<'] \ar[to=MatA,X] & \zxMatrix[a=MatA]{R}  \ar[to=Zmerge] & \zxX[a=Zmerge]{} \rar & \\         
    \end{ZX}\nonumber\\&\eqZX{Z}
    % Start to move red node
    \begin{ZX}[mbr=2,myBold,
      execute at end picture={
        % \zxNamedBox[fit margins={all=.5mm,top=1.2mm,left=1mm}][yellow]{(Gamma1)(measA)(bellA2)}{}
      },
      ]
           &                                        &                             &                                    &                                                                & \zxN[a=midNode]{} \ar[to=Zmerge,'>]  &                       & [\zxwCol] \\
      \rar & \zxZ[a=bellA1]{U^T(b {::} R^Tb)\pi} & [\zxwCol] \zxMatrix{U} \rar & \zxX{((c \xor Rd){::} 0)\pi} \rar & \zxDivider{} \ar[to=bellA1,<.] \ar[to=midNode,<'] \ar[to=MatA] & \zxMatrix[a=MatA]{R}  \ar[to=Zmerge] & \zxX[a=Zmerge]{} \rar &           \\         
    \end{ZX}\eqeq{eq:bijectionA}
    % Swipe U
    \begin{ZX}[mbr=2,myBold,
      execute at end picture={
        % \zxNamedBox[fit margins={all=.5mm,top=1.2mm,left=1mm}][yellow]{(Gamma1)(measA)(bellA2)}{}
      },
      ]
           &                                        &                             &                                    &                                                                & \zxN[a=midNode]{} \ar[to=Zmerge,'>]  &                       & [\zxwCol] \\
      \rar & \zxZ[a=bellA1]{U^T(b {::} R^Tb)\pi} & [\zxwCol] \zxMatrix*{U^{-1}} \rar & \zxX{((c \xor Rd){::} 0)\pi} \rar & \zxDivider{} \ar[to=bellA1,<.] \ar[to=midNode,<'] \ar[to=MatA] & \zxMatrix[a=MatA]{R}  \ar[to=Zmerge] & \zxX[a=Zmerge]{} \rar &           \\         
    \end{ZX}\nonumber\\&\eqSubeq{eq:matrixManipulationC}{a}
    % Commute U-1 and red node
    \begin{ZX}[mbr=2,myBold,
      execute at end picture={
        % \zxNamedBox[fit margins={all=.5mm,top=1.2mm,left=1mm}][yellow]{(Gamma1)(measA)(bellA2)}{}
      },
      ]
           &                                             &                                         &                          &                                                                & \zxN[a=midNode]{} \ar[to=Zmerge,'>]  &                       & [\zxwCol] \\
      \rar & \zxZ[a=bellA1]{U^T(b {::} R^Tb)\pi} \rar & \zxX{U^{-1}((c \xor Rd){::} 0)\pi} \rar & \zxMatrix*{U^{-1}}  \rar & \zxDivider{} \ar[to=bellA1,<.] \ar[to=midNode,<'] \ar[to=MatA] & \zxMatrix[a=MatA]{R}  \ar[to=Zmerge] & \zxX[a=Zmerge]{} \rar &           \\         
    \end{ZX}\eqeq{eq:bijectionA}
    % Commute U-1 and red node
    \begin{ZX}[mbr=2,myBold,
      execute at end picture={
        % \zxNamedBox[fit margins={all=.5mm,top=1.2mm,left=1mm}][yellow]{(Gamma1)(measA)(bellA2)}{}
      },
      ]
           &                                             &                                         &                          &                                                                & \zxN[a=midNode]{} \ar[to=Zmerge,'>]  &                       & [\zxwCol] \\
      \rar & \zxZ[a=bellA1]{U^T(b {::} R^Tb)\pi} \rar & \zxX{U^{-1}((c \xor Rd){::} 0)\pi} \rar & \zxMatrix{U}  \rar & \zxDivider{} \ar[to=bellA1,<.] \ar[to=midNode,<'] \ar[to=MatA] & \zxMatrix[a=MatA]{R}  \ar[to=Zmerge] & \zxX[a=Zmerge]{} \rar &           \\         
    \end{ZX}\label{eq:rewriteTopPart}
  \end{align}
  We inject this back into the \tikz[baseline=(a.base)]\node[rounded corners,inner sep=2pt,fill=yellow!50!white,opacity=.5,text opacity=1](a){yellow-colored}; part of \cref{eq:simplifiedLowerPart}, after defining $x \eqdef U^{-1}((c \xor Rd){::}0)$:
  \begin{align}
    \ket{\psi_{a,b,c,d}}
    &\eqeq{eq:rewriteTopPart}
      \begin{ZX}[myBold,
        execute at end picture={
          %\zxNamedBox[fit margins={all=.5mm,top=1.2mm,left=1mm}][yellow]{(Gamma1)(measA)(bellA2)}{}
        },
        ]
                                                                                                                                                                                                                                                                                                                                                                                                                                                    \\[3mm]
       \zxBox{G_H} \rar  & \zxZ{} \ar[to=bellA1,<.] \ar[to=endA]   &                                          & [\zxHCol+3mm]                            & \zxZ[a=endA,pilb={overlay,yshift=-1mm,font=\fontsize{6}{10}\selectfont\boldmath}]{((U^T(b{::}R^Tb)) \xor Gx)\pi} \rar & \zxX[label={[overlay,yshift=1mm]\scriptstyle x \eqdef U^{-1}((c \xor Rd){::}0)}]*{x\pi} \rar &                                                                  \\[2mm]
                         &                                         &                                          &                                          &                                                                                                                        &                                                                                              & \zxN[a=midNode]{} \ar[to=Zmerge,'>]                              \\
                         &                                         & \zxZ[a=bellA1]{U^T(b {::} R^Tb)\pi} \rar & \zxX{x\pi} \rar  & \zxMatrix{U}  \rar                                                                                                     & \zxDivider{} \ar[to=bellA1,<.] \ar[to=midNode,<'] \ar[to=MatA]                               & \zxMatrix[a=MatA]{R}  \ar[to=Zmerge] & \zxX[a=Zmerge,a=bellA2]{} \\
                         &                                         &                                          & \zxMatrix[a=gamma2]-/*{V} \ar[to=divB,wc] & \zxDivider[a=divB]{} \ar[to=measD,s] \ar[to=bellA2,C-]                                                                 &                                                                                              &                                                                  \\
                         &                                         &                                          &                                          &                                                                                                                        &                                                                                              & \zxX[a=measD]*{}                                                 \\
        \zxBox{G_M} \rar & \zxZ{} \ar[to=gamma2,<',H] \ar[to=endB] &                                          &                                          &                                                                                                                        &                                                                                              &                                      & \zxN[a=endB]{}
      \end{ZX}\eqeq{eq:formGamma}
      \begin{ZX}[mbr=3,myBold,
        execute at end picture={
          %\zxNamedBox[fit margins={all=.5mm,top=1.2mm,left=1mm}][yellow]{(Gamma1)(measA)(bellA2)}{}
        },
        ]
                                                                                                                                                                                                                                                                                                                \\[3mm]
       \zxBox{G_H} \rar  & \zxZ{} \ar[to=bellA1] \ar[to=endA]      &  & [\zxHCol+3mm] & \zxZ[a=endA,pilb={overlay,yshift=-1mm,font=\fontsize{6}{10}\selectfont\boldmath}]{((U^T(b{::}R^Tb)) \xor Gx)\pi} \rar & \zxX[label={[overlay,yshift=1mm]\scriptstyle x \eqdef U^{-1}((c \xor Rd){::}0)}]*{x\pi} \rar &[6mm] \\[2mm]
                         & \zxZ[a=bellA1]{U^T(b {::} R^Tb)\pi} \ar[to=Xcorr]                                                                                                                                                                                                                                    \\
                         & \zxX[a=Xcorr]{x\pi} \ar[to=gamma3]                                                                                                                                                                                                                           \\
                         & \zxMatrix[a=gamma3]'{\Gamma}                                                                                                                                                                                                                                                          \\
        \zxBox{G_M} \rar & \zxZ{} \ar[to=gamma3,H] \ar[to=endB] &  &               &                                                                                                                        &                                                                                              &  & \zxN[a=endB]{}
      \end{ZX} \nonumber\\&\eqZXX{S}{Pi}
      % Fuse Z corr
      \begin{ZX}[mbr=3,myBold,
        execute at end picture={
          %\zxNamedBox[fit margins={all=.5mm,top=1.2mm,left=1mm}][yellow]{(Gamma1)(measA)(bellA2)}{}
        },
        ]
                                                                                                                                                                                                             \\[3mm]
       \zxBox{G_H} \rar          & \zxZ[a=endA]{Gx\pi} \rar & \zxX[label={[overlay,yshift=1mm]\scriptstyle x \eqdef U^{-1}((c \xor Rd){::}0)}]*{x\pi} \rar & \zxZ{} \ar[to=gamma3]  \rar                  & \\
                                 &                           &                                                                                              & \zxMatrix[a=gamma3]'{\Gamma}                   \\
        \zxBox{G_M} \ar[to=botZ] &                           &                                                                                              & \zxZ[a=botZ]{} \ar[to=gamma3,H] \ar[to=endB] & \zxN[a=endB]{}
      \end{ZX}\eqeq{eq:stabilizerDiag}
      % Stabilizer
      \begin{ZX}[mbr=2,myBold,
        execute at end picture={
          %\zxNamedBox[fit margins={all=.5mm,top=1.2mm,left=1mm}][yellow]{(Gamma1)(measA)(bellA2)}{}
        },
        ]
       \zxBox{G_H} \rar          & \zxZ{} \ar[to=gamma3]  \rar                  & \\
                                 & \zxMatrix[a=gamma3]'{\Gamma}                   \\
        \zxBox{G_M} \ar[to=botZ] & \zxZ[a=botZ]{} \ar[to=gamma3,H] \ar[to=endB] & \zxN[a=endB]{}
      \end{ZX} \eqeqLem{lem:graphicalCutG} \zx[myBold]{\zxBox{G}\zxNR} = \ket{G}
  \end{align}
  Which concludes the proof.
\end{proofE}

Since the above theorem is true for any graph state, it is also true for \GHZ{} states. However, the merging operation can be significantly simplified in that setting:
\begin{corollaryE}[\GHZ{} states are mergeable][end,restate]\label{cor:ghzstatesAreMergeable}
  Any \GHZ{} state of size $|n|$ (each qubit being a separate register) is mergeable (\cref{def:mergeable}) with respect to the collection of quantum maps $\{\xi_H\}_{H \subseteq [n]}$, where $\xi_H$ takes two bits $(x,z) \in \{0,1\}^2$ as input, applies $\gateZ^{z}$ of the first qubit, and $\gateX^{x}$ on all its input qubits (if $|H|$ is empty, it does not do anything).
\end{corollaryE}
\begin{proofE}
  While this is a direct consequence of \cref{thm:anygraphismergeable}, we provide here a simpler proof specific to GHZ states. First, we note that if $|H|$ is empty, \cref{eq:mergeable} is trivially achieved by defining $\xi_\Sim$ as the map that measures completely its input state, and outputs, say, $x=0$ and $z=0$. Since neither $\xi_\Sim$ nor $\xi_H$ touch the second state at all, this outputs the second state hence achieving \cref{eq:mergeable}.

  If $H = [n]$, on the other hand, $\xi_\Sim$ has only access to the second state. If we define $\xi_\Sim$ that measures completely its input state, and outputs $x = 0$, $z = 0$ to $\xi_H$, then $\xi_H$ will left the first state intact, therefore achieving \cref{eq:mergeable}.
  
  If $|H|$ is not empty, then we define $\xi_\Sim$ as follows:
  \begin{itemize}
  \item $\xi_\Sim$ performs a Bell measurement (i.e.\ a projection\footnote{This measurement can be done via a CNOT gate, and measuring the first qubit in the Hadamard basis to get $z_0$ and the second qubit in the computational basis to get $x$.} on one of the four Bell states $\ket{0x}+(-1)^{z_0}\ket{1\bar{x}}$, where $(x,z_0) \in \{0,1\}^2$) between the last qubit of the registers in $R_{\bar{H}}$, and the first qubit of $R_H$, getting outcomes $(x,z_0)$.
  \item Then $\xi_\Sim$ measures all remaining qubits in the Hadamard basis, getting outcomes $\{z_i\}_{i \in \{1,\dots,n-2\}}$
  \item Finally, $\xi_\Sim$ outputs $x$ and $z \eqdef \xor_{i \in \{0,n-2\}} z_i$.
  \end{itemize}
  We prove now that this $\xi_\Sim$ quantum maps achieves \cref{eq:mergeable} using ZX calculus:
  \begin{align}
    \begin{ZX}[
      execute at end picture={
        \node[
          fit=(z0)(zn2)(zne)(z1),
          node on layer=background,
          rounded corners,
          draw=orange,
          dashed,
          fill=orange!50!white,
          opacity=.3,
          "{\scalebox{.9}{$\xi_\Sim$, outputs $(x, \xor_i z_i \pi)$}}" {right, inner sep=1mm}
        ]{};
        \node[
          fit=(xpi2)(xorz),
          node on layer=background,
          rounded corners,
          draw=blue,
          dashed,
          fill=blue!50!white,
          opacity=.3,
          "{\scalebox{.9}{$\xi_E$}}" {right, inner sep=1mm}
        ]{};
      },
      ]
      \zxN{}                                                            & [\zxwCol] \zxN{} \ar[d,3 vdotsr=\scalebox{.7}{$|E|$ times}] \ar[rr] &  & \zxZ[a=xorz]{\xor_i z_i \pi} \rar & \zxX{x\pi}\dar[3 vdots] \rar & [\zxwCol] \zxN{} \\[\zxDotsRow]
      \zxN{}                                                            & \zxN{} \ar[rrr]                         &  &                           & \zxX[a=xpi2]{x\pi} \rar              & \zxN{}           \\
      \zxZ{} \ar[ur,<'] \ar[uur,<'] \ar[dr,<.] \ar[ddr,<.] \ar[dddr,<.] &                                                                                                                          \\
      \zxN{}                                                            & \zxN{}\ar[d,3 vdots] \rar               & \zxZ[a=z1]{z_{1}\pi}                                                                 \\
      \zxN{}                                                            & \zxN{} \rar                             & \zxZ[a=zne]{z_{n-|E|-1}\pi}                                                           \\
      \zxN{}                                                            & \zxX{x\pi} \ar[d,C-]                                                                                                     \\
      \zxN{}                                                            & \zxZ[a=z0]{z_0\pi}                                                                                                             \\
      \zxN{}                                                            & \zxN{}\ar[d,3 vdots] \rar               & \zxZ{z_{n-|E|}\pi}                                                             \\
      \zxN{}                                                            & \zxN{} \rar                             & \zxZ[a=zn2]{z_{n-2}\pi}                                                               \\
      \zxZ{} \ar[ur,<'] \ar[uur,<'] \ar[dr,<.] \ar[uuur,<'] \ar[ddr,<.] &                                                                                                                          \\
      \zxN{}                                                            & \zxN{}\dar[3 vdotsr=\scalebox{.7}{$n-|E|$ times}] \ar[rrrr]           &  &                           &                              & \zxN{}           \\[\zxWRow+1mm]
      \zxN{}                                                            & \zxN{}    \ar[rrrr]                     &  &                           &                              & \zxN{}           \\
    \end{ZX}
    &=
    \begin{ZX}
      \zxN{}                                                 & [\zxwCol] \zxN{} \ar[d,3 vdotsr=\scalebox{.7}{$|E|$ times}] \ar[rr] &  & \zxZ[a=xorz]{\xor_i z_i \pi} \rar & \zxX{x\pi}\dar[3 vdots] \rar & [\zxwCol] \zxN{} \\[\zxDotsRow]
      \zxN{}                                                 & \zxN{} \ar[rrr]                                                     &  &                                   & \zxX[a=xpi2]{x\pi} \rar      & \zxN{}           \\
      \zxZ{\xor_i z_i \pi} \ar[ur,<'] \ar[uur,<'] \ar[dr,N] &                                                                                                                                                              \\
      \zxN{}                                                 & \zxX{x\pi}\\
      \zxZ{} \ar[dr,<.] \ar[ur,N] \ar[ddr,<.]               &                                                                                                                                                              \\
      \zxN{}                                                 & \zxN{}\dar[3 vdotsr=\scalebox{.7}{$n-|E|$ times}] \ar[rrrr]         &  &                                   &                              & \zxN{}           \\[\zxWRow+1mm]
      \zxN{}                                                 & \zxN{}    \ar[rrrr]                                                 &  &                                   &                              & \zxN{}           \\
    \end{ZX}  \\
    &=
    \begin{ZX}
      \zxN{}                                                 & [\zxwCol] \zxN{} \ar[d,3 vdotsr=\scalebox{.7}{$|E|$ times}] \ar[rr] &  & \zxX{x\pi}\dar[3 vdots] \rar & [\zxwCol] \zxN{} \\[\zxDotsRow]
      \zxN{}                                                 & \zxN{} \ar[rrr]                                                     &  &                                   & \zxX[a=xpi2]{x\pi} \rar      & \zxN{}           \\
      \zxZ{} \ar[ur,<'] \ar[uur,<'] \ar[dr,N] &                                                                                                                                                              \\
      \zxN{}                                                 & \zxX{x\pi}\\
      \zxZ{} \ar[dr,<.] \ar[ur,N] \ar[ddr,<.]               &                                                                                                                                                              \\
      \zxN{}                                                 & \zxN{}\dar[3 vdotsr=\scalebox{.7}{$n-|E|$ times}] \ar[rrrr]         &  &                                   &                              & \zxN{}           \\[\zxWRow+1mm]
      \zxN{}                                                 & \zxN{}    \ar[rrrr]                                                 &  &                                   &                              & \zxN{}           \\
    \end{ZX}
    =
    \begin{ZX}
      \zxN{}                       & \zxN{} \ar[dd,3 vdotsr=\scalebox{.7}{$n$ times}] \\
      \zxZ{} \ar[ur,<'] \ar[dr,<.] &        \\
      \zxN{}                       & \zxN{} \\
    \end{ZX}  
  \end{align}
\end{proofE}

\paragraph{Security proof of \cref{thm:mainThm}.} Building on this corollary, we prove now the security part of \cref{thm:mainThm}, schematised in \cref{fig:DishonestTOUT}. The simulators that we will define are informally drawn in \cref{fig:Simulators}.

\begin{figure}[!ht]
  \centering
  \begin{subfigure}[b]{0.3\textwidth}
    \centering
    \includegraphics[width=\textwidth]{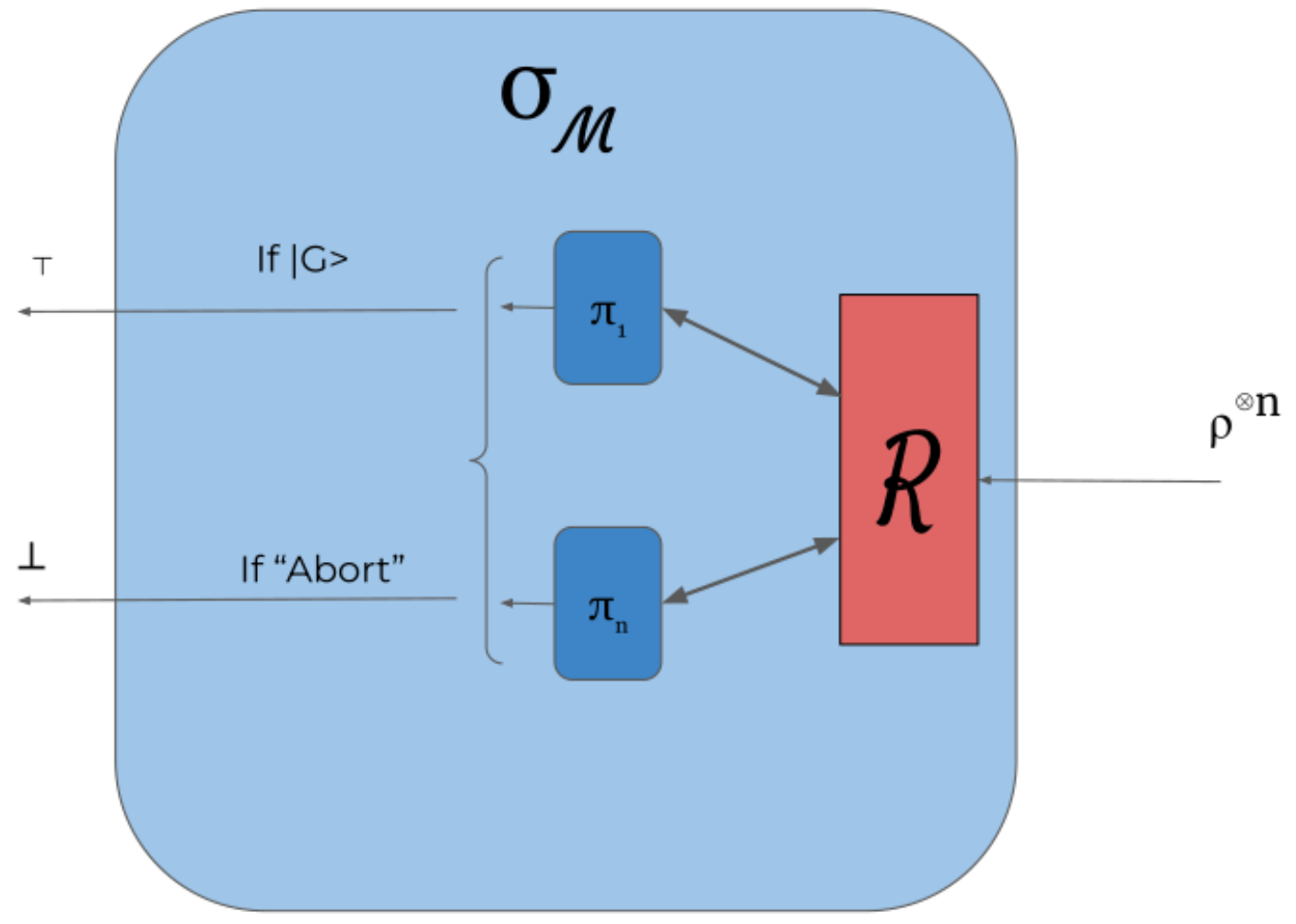}
    \caption{Special case where the subset of corrupted parties $M$ contains only the source (i.e.\ $M = \{\source\}$).}
    \label{fig:SigmaS}
  \end{subfigure}
  \hfill
  \begin{subfigure}[b]{0.3\textwidth}
    \centering
    \includegraphics[width=\textwidth]{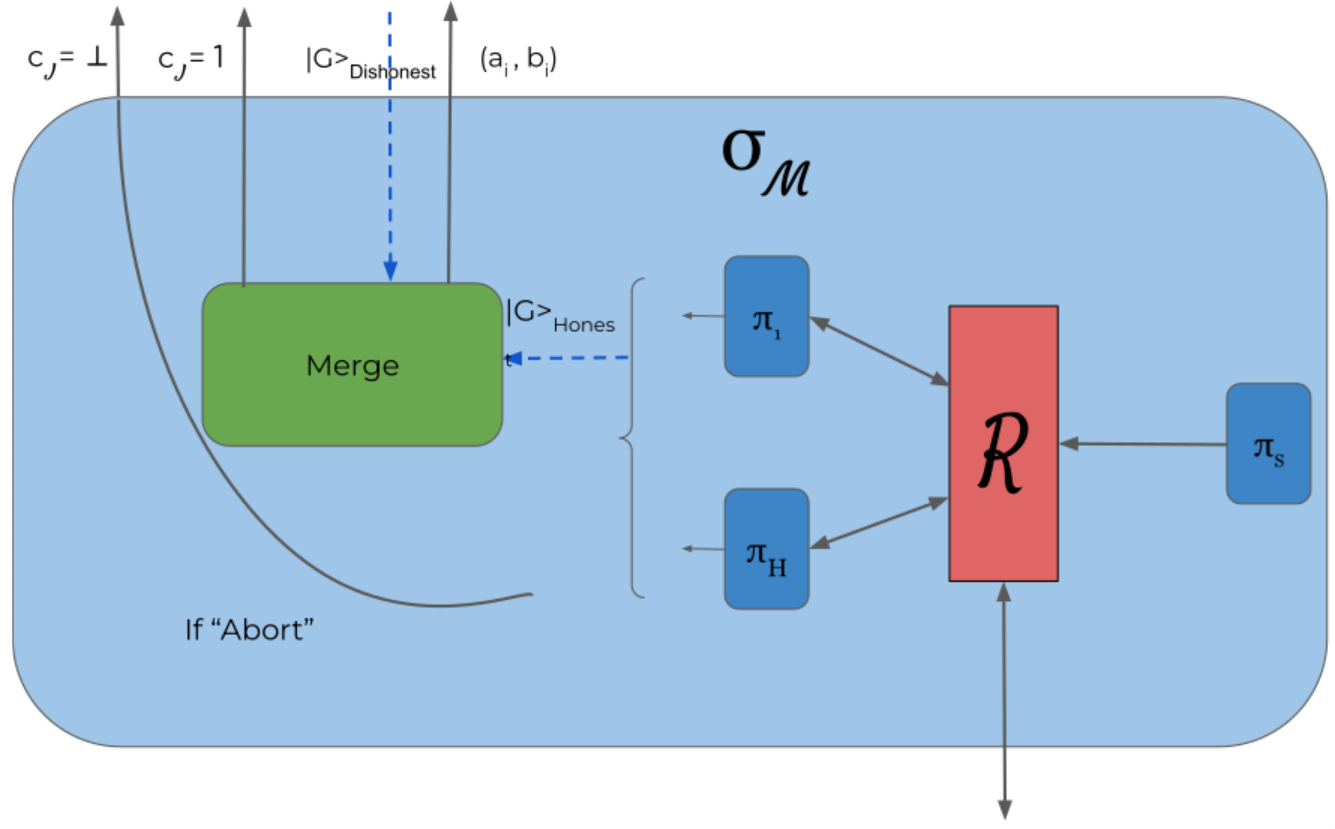}
    \caption{Special case where the subset of corrupted parties $M$ contains only some parties (excluding the source).}
    \label{fig:SigmaJ}
  \end{subfigure}
  \hfill
  \begin{subfigure}[b]{0.3\textwidth}
    \centering
    \includegraphics[width=\textwidth]{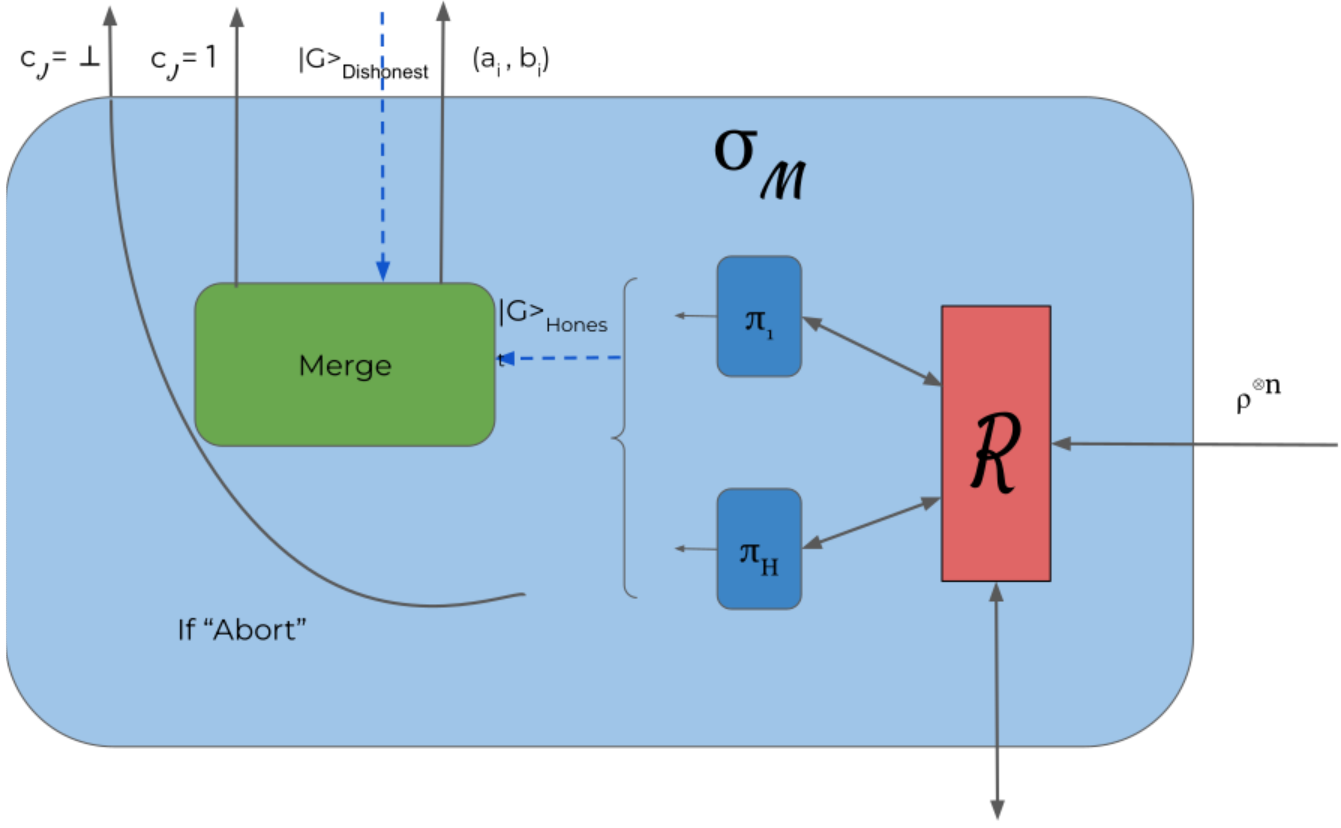}
    \caption{Special case where the subset of corrupted parties $M$ contains the source and a coalition of dishonest parties.}
    \label{fig:SigmaD}
  \end{subfigure}
  \caption{Simulators for (a) a dishonest source, (b) a subset of dishonest parties and (c) a coalition of dishonest parties and source. Note that even if we describe three simulators depending on the set of corrupted resources for clarity, in the proof we directly describe the generic case where the set $M$ of corrupted parties can be arbitrary.}
  \label{fig:Simulators}
\end{figure}

\begin{lemmaE}[{Security proof of \cref{thm:mainThm}}][end, restate, text link section full proof]
  Let $M \subseteq [n] \cup \{\source\}$ be an arbitrary subset of corrupted parties, and $H = [n] \cup \{\source\} \setminus M$ be the honest parties. Then there exists a simulator $\Sim_M$ such that $\pi_H \cR \approx \filter_{H} \Verif^f \Sim_M$ for each possible subset $M$.
\end{lemmaE}
\begin{proof}[Proof sketch]
  This proof is relatively technical and is therefore deferred to appendix. Yet, the informal idea is to use the fact that the state is mergeable. More precisely, the simulator can simulate the protocol run by honest parties, forwarding all messages from/to the distinguisher, to get the state outputted by the honest parties. If no party aborted, then the state should be close enough to $\ket{G}^H$ (part of the technicality is to write this properly since part of this state is owned by the adversary). In that case, the simulator will then get from $\Verif^f$ another copy of $\ket{G}^M$, and use the fact that this state is mergeable in order to merge these two copies of $\ket{G}$ into a single state close to $\ket{G}$ up to some corrections, shared between the distinguisher and the functionality. The merging operation will provide these corrections to the simulator, that can forward them to the ideal resource. The resource will then apply these corrections to ensure that the shared state is now close to $\ket{G}$, before outputting its part of $\ket{G}$ to the distinguisher.
\end{proof}
\begin{proofE}
  Let $M \subseteq [n] \cup \{\source\}$ be the subset of corrupted parties, and $H = [n] \cup \{\source\} \setminus M$ be the honest parties. We need to prove that there exists a simulator $\Sim_M$ such that $\pi_H \cR \approx \filter_{H} \Verif^f \Sim_M$ for each possible subset $M$.
  \begin{remark}
    Note that for simplicity, we consider $\Sim_M$ to be global (i.e.\ it is a single entity able to communicate will all interfaces), while the constructive cryptography framework typically expect simulators to be local\footnote{Note that some framework do not make that choice, like in universal composability.}: this is without loss of generality since we can easily turn it back into a set of local simulators, one simulator performing the operations of $\Sim_M$, where the input (resp.\ outputs) are obtained (resp.\ forwarded) from (resp.\ to) the right interface by mean of the channel $\cC$ included in $\Verif^f$, using the appropriate simulator connected to this interface as a proxy.
  \end{remark}

Since it is the most general case, we will focus on proving security for an arbitrary coalition of dishonest parties and source (\cref{fig:SigmaD}). The other cases follow directly from this one. The simulator $\Sim_M$ can informally be summarized as follows:
\begin{align}
  \begin{ZX}<circuit,forward=\cV,forward=\Verif,fw=\cR,fw=\filter>[circuit, thick lines,zx column sep=6mm,mbr=2,
    /utils/exec={%
      \def\blockFwIfFx{
        \zxGate[a=fna]{\parbox{1.5cm}{\centering\tiny Forward if $c_i$'s $\neq$ abort and $f(x) = \top$, else $\ket{\bot}^{H}$}}
        \rar[latex-,Bn'Args={H}{pos=.60},thick]
      }%
      \def\blockXiSigma{%
        \zxGate{\xi_\sigma}
        \zxContName[orange]{(cR)}{below:$\sigma_M$}
        \rar[latex-,Bn'Args={H}{pos=.6}, thick]%
      }%
      \def\blockInterfaceParties{%
        \zxInputMulti{1}{\parbox{1.8cm}{\centering\tiny Interface of honests parties}}
        \rar[latex-,Bn'Args={H}{pos=.50},thick]%
      }%
      \def\blockVerif{%
        \zxContName{(fna)(G)}{below:$\Verif^f$}%
      }%
      \def\blockG{%
        \zxGate[a=G]{\ket{G}}
        \ar[rd,|--,-latex,start anchor={[xshift=1mm]south}, Bn'Args={M}{pos=.80},thick]
        \ar[ld,|--,-latex,start anchor={[xshift=-1mm]south},Bn'Args={H}{pos=.75},thick]
        \ar[rrd,'>,start anchor=east, end anchor=north,latex-, "c_i = \bot\text{\tiny\ if abort}"]
      }%
      \def\blockFwIfNoAbort{%
        \zxGate{\parbox{1.2cm}{\centering\tiny Forward if \\no abort}}
        \rar[latex-,Bn'Args={H}{pos=.6}, thick]%
      }%
      \def\blockPiH{%
        \zxGate{\pi_H}
        \rar[latex-latex]%
      }%
      \def\blockCr{%
        \zxGate[a=cR]{\cR}
        \rar[latex-latex]%
      }
      \def\blockMultipleParties{%
        \zxOutputMulti{1}{\parbox{1.5cm}{\centering\tiny Interface of corrupted parties}}%
      }%
      \def\blockXiH{%
        \zxGate[a=xi]{\xi_H} \ar[rr,latex-,cl,shift left=-2mm,"{x}" {below}]%
      }%
      \def\blockFilter{%
        \zxGate{\filter_H}%
        \ar[to=G,"c_i = 0",-latex]
      }%
    },
    ]
    \blockFilter           &              & \blockVerif & \blockG &               &                   &           &          & \\[\zxZeroRow-1mm]
    \blockInterfaceParties & \blockFwIfFx & \blockXiH   &         & \blockXiSigma & \blockFwIfNoAbort & \blockPiH & \blockCr & \blockMultipleParties
  \end{ZX}
\end{align}
More formally, $\Sim_M$ is defined as:
\begin{enumerate}
\item $\Sim_M$ will first locally simulate $\pi_H \cR$, where the interfaces of the malicious parties are directly connected to the distinguisher.
\item If an honest party in the local simulation $\pi_H$ aborts within $\sigma_M$, the simulator will send $\bot$ to the ideal functionality and abort. Otherwise, we denote as $\rho^{R_H}$ the state outputted at the end of the protocol by the honest parties in $H \setminus \{\source\}$ (if this set is empty, $\rho^{R_H}$ is just the empty state).
\item Then, the simulator will send $c_i = 1$ to the functionality for all $i \in M \setminus \{\source\}$ to obtain a quantum state $\ket{G}^{R_M}$ (similarly, if $M \setminus \{\source\}$ is empty, the simulator does not send any message and defines $\ket{G}^{R_M}$ as the empty state).
\item The simulator runs the merging map $\xi_\Sim$ described in \cref{thm:anygraphismergeable} applied to $\rho^{R_H}$ and $\ket{G}^{R_M}$ (notice that $\ket{G}^{R_H}$ is replaced with $\rho^{R_H}$), and gets $(x,y) \in (\{0,1\}^{|H|})^2$.
\item The simulator sends $(x,y)$ on an arbitrary malicious interface in $M \setminus \{\source\}$ (if $M \setminus \{\source\}$ is empty, i.e.\ only the source may be malicious, then the simulator does not do anything at that step).
\end{enumerate}

We will prove now that $\pi_H \cR \approx_{\epsilon'} \filter_{H} \Verif^f \Sim_M$ for some $\epsilon'$ defined later, by defining a series of hybrid systems close to each others. We recall that $\pi_H \cR$ and $\filter_{H} \Verif^f \Sim_M$ are the resources corresponding respectively to the concrete protocol and ideal protocol with a subset $H$ of honest parties and $M=\bar{H}$ of dishonest parties.

The first step will be to simplify the system combining the functionality and simulator. We first remark that verifying if $f(x) = \top$ is useless since $\xi_\sigma$ always output $x$ such that $f(x) = \top$ (we defined $f$ exactly to have this property). The second remark we can make is that it is much simpler to analyse this system if we group the operations differently in order to remove any interaction between the functionality and the simulator. So informally, we will group the operations related to the merging operation together instead of having them shared between the simulator and the functionality, while we will check if we need to abort only at the very end, giving this informal picture:

\begin{align}
  \begin{ZX}<circuit,forward=\cV,forward=\Verif,fw=\cR,fw=\Merge,fw=\CheckState>[circuit, thick lines,zx column sep=6mm,mbr=2,
    /utils/exec={%
      \def\blockFwIfFx{
        \zxGate[a=fna]{\parbox{2cm}{\centering\tiny Output $\ket{\bot}^{H}$ if input contains $\bot$, else forward}}
        \zxContName{}{below:\parbox{5cm}{\centering\small$\CheckState$}}%
        \rar[latex-,Bn'Args={H}{pos=.60},thick]
      }%
      \def\blockXiSigma{%
        \zxGate[a=xiSigma]{\xi_\sigma}
        \rar[latex-,Bn'Args={H}{pos=.6}, thick]%
      }%
      \def\blockInterfaceParties{%
        \zxInputMulti{1}{\parbox{1.8cm}{\centering\tiny Interface of honests parties}}
        \rar[latex-,Bn'Args={H}{pos=.50},thick]%
      }%
      \def\blockMerge{%
        \zxContName{(G)(xiSigma)(xi)}{below:\parbox{5cm}{\centering\small$\Merge$ \\\tiny(if input contains $\bot$, do identity)}}%
      }%
      \def\blockG{%
        \zxGate[a=G]{\ket{G}}
        \ar[rd,|--,-latex,start anchor={[xshift=1mm]south}, Bn'Args={M}{pos=.80},thick]
        \ar[ld,|--,-latex,start anchor={[xshift=-1mm]south},Bn'Args={H}{pos=.75},thick]
      }%
      \def\blockPiH{%
        \zxGate{\pi_H}
        \rar[latex-latex]%
      }%
      \def\blockCr{%
        \zxGate[a=cR]{\cR}
        \rar[latex-latex]%
      }
      \def\blockMultipleParties{%
        \zxOutputMulti{1}{\parbox{1.5cm}{\centering\tiny Interface of corrupted parties}}%
      }%
      \def\blockXiH{%
        \zxGate[a=xi]{\xi_H} \ar[rr,latex-,cl,shift left=-2mm,"{x}" {below}]%
      }%
    },
    ]
                           &              & \blockMerge & \blockG &               &           &          & \\[\zxZeroRow-1mm]
    \blockInterfaceParties & \blockFwIfFx & \blockXiH   &         & \blockXiSigma & \blockPiH & \blockCr & \blockMultipleParties
  \end{ZX}
\end{align}

More formally, let $\Merge$ be the map working on the registers $R_H$ outputted by the honest parties that applies the identity if one of its input qubit is $\ket{\bot}$ (this can be done via a partial measurement like in the step $2$ of the definition of the simulator) and that applies otherwise the merging operation $\xi_H (I_{|H|} \otimes \xi_\Sim) (\ket{G}^{R_H,R_{\bar{H}}} \otimes I_{|H|})$ that merges the state in $R_H$ with a newly created state $\ket{G}$ after applying the corrections. With this definition, the system representing the simulator and the functionality and filters is strictly indistinguishable from this system:

\begin{itemize}
\item First we simulate locally $\pi_H \cR$ with the distinguisher's input to obtain $\rho_H$ on register $R_H$.
\item Then, we apply $\Merge$.
\item Then, we do a partial measurement of the state obtained at the previous step, checking if one input is $\ket{\bot}$. If yes, we replace the state with $\ket{\bot}^{\otimes |H|}$, otherwise we apply the identity (we call this operation $\CheckState$).
\item Finally, for each $i \in H \setminus \{\source\}$ we send the $i$-th qubit to the party $i$.
\end{itemize}
This hybrid system is obviously indistinguishable from $\filter_{H} \Verif^f \Sim_M$ as the system is only an identical simplification of the original system (again, $f(x)$ is always equal to $\top$, and otherwise in both cases we abort if any input contains $\ket{\bot}$ and apply the merging operation otherwise).

We remark now that it is nearly equal to $\pi_H \cR$, except for the application of $\Merge$ and $\CheckState$ at the end. To prove that these two operations do not significantly change the state obtained after the first step, we consider the quantum state obtained at the end of the first step, i.e.\ when getting the outcomes of $\pi_H$. Without loss of generality, we can assume that we run a purified version of the distinguisher and protocol, or, equivalently, that we consider the averaged state (averaging over all randomness involved in the protocol and in the distinguisher) where the purification of the density state is kept on a register kept by the environment\footnote{This purification is only used to simplify computations as this way we do not need to consider each run separately. Note that this cannot decrease the distinguishing probability of the distinguisher to purify the protocol, since we can purify any protocol and distinguisher by simply replacing any sampling operation by a measurement of a $\ket{+}$, and any measurement in the computational basis can be replaced by a $\gateCNOT$ on an auxiliary qubit kept by the environment, where the target state will then contain the result of the measurement. One can easily see that this is equivalent since we can consider that the state given to the environment is traced out.}. We can also show that the dimension of the purified system is at least $2(n+1)$ (if needed we can add $\ket{0}$ states on the distinguisher's state, this dimension is just chosen large enough so that any later purification has at most this dimension). We call $\ket{\psi}^{R_\Sim, R_E}$ the joint state between the simulator and the environment and distinguisher after receiving the outputs of $\pi_H$:

\begin{align}
  \begin{ZX}<circuit,forward=\cV,forward=\Verif,fw=\cR,fw=\Merge,fw=\CheckState>[circuit, thick lines,zx column sep=6mm,mbr=2,
    /utils/exec={%
      \def\blockFwIfFx{
        \zxGate[a=fna]{\parbox{2cm}{\centering\tiny Output $\ket{\bot}^{H}$ if input contains $\bot$, else forward}}
        \zxContName{}{below:\parbox{5cm}{\centering\small$\CheckState$}}%
        \rar[latex-,Bn'Args={H}{pos=.60},thick]
      }%
      \def\blockXiSigma{%
        \zxGate[a=xiSigma]{\xi_\sigma}
        \rar[latex-,Bn'Args={H}{pos=.5}, thick]%
      }%
      \def\blockInterfaceParties{%
        \zxInputMulti{1}{\parbox{1.8cm}{\centering\tiny Interface of honests parties}}
        \rar[latex-,Bn'Args={H}{pos=.50},thick]%
      }%
      \def\blockMerge{%
        \zxContName{(G)(xiSigma)(xi)}{below:\parbox{5cm}{\centering\small$\Merge$ \\\tiny(if input contains $\bot$, do identity)}}%
      }%
      \def\blockG{%
        \zxGate[a=G]{\ket{G}}
        \ar[rd,|--,-latex,start anchor={[xshift=1mm]south}, Bn'Args={M}{pos=.80},thick]
        \ar[ld,|--,-latex,start anchor={[xshift=-1mm]south},Bn'Args={H}{pos=.75},thick]
      }%
      \def\blockPiH{%
        \zxGate{\pi_H}
        \rar[latex-latex]%
      }%
      \def\blockCr{%
        \zxGate[a=cR]{\cR}
        \rar[latex-latex]%
      }
      \def\blockMultipleParties{%
        \zxOutputMulti{1}{\parbox{1.5cm}{\centering\tiny Interface of corrupted parties}}%
      }%
      \def\blockXiH{%
        \zxGate[a=xi]{\xi_H} \ar[rr,latex-,cl,shift left=-2mm,"{x}" {below}]%
      }%
      \def\blockSlice{%
        \zxSlice[][][.65]{$\ket{\psi}^{R_H,R_E}$                                                                                           \\[1mm]($R_E$ owned by distinguisher and environment)}%
      }%
      \def\blockRecognizeReal{%
        \zxVInputMulti[enlarge brace=5mm]{2}{\parbox{2.5cm}{\centering\tiny We recognize here the real world}}%
      }%
    },
    ]                      &              &             &         & \blockSlice   &                     &                                  \\[\zxZeroRow+1mm]
                           &              & \blockMerge & \blockG &               &                     &          &                       \\[\zxZeroRow-1mm]
    \blockInterfaceParties & \blockFwIfFx & \blockXiH   &         & \blockXiSigma & \blockPiH           & \blockCr & \blockMultipleParties \\
                           &              &             &         &               & \blockRecognizeReal & 
  \end{ZX}
\end{align}

  Since $\ket{\psi}^{R_H, R_E}$ is the state obtained in the real world, and because all the operations performed by the distinguisher are the same in both the ideal and the concrete worlds, it is enough to prove that $\TD(((\CheckState \circ \Merge)^{R_H} \otimes I^{R_E})\ket{\psi}^{R_H, R_E}),\ket{\psi}^{R_H, R_E})$ is small to conclude that the distinguisher cannot distinguish the two worlds. Indeed, the distinguisher\footnote{To be more precise, we could use the second part of the distinguisher that runs after obtaining $\ket{\psi}^{R_H,R_E}$ to distinguish these two states.} could otherwise distinguish two states close in trace distance which is impossible by the laws of physics. First, let $\rho \eqdef \Tr_E(\ket{\psi}^{R_H,R_E})$. Then, since $\Pi$ is secure (\cref{thm:mainThm}), by assumption there exists $p$ such that:
  \begin{align}
    F\left(\rho, \sigma\right) \geq 1-\eps(\lambda)\label{eq:Fbound}
  \end{align}
  where $\sigma$ is defined like in \cref{def:graphStateVerifProtoc}. However, we can remark that $\sigma$ admits the following purification:
  \begin{align}
    \ket{G_p} \eqdef \sqrt{p} \ket{G}\ket{0}\ket{0^k} + \sqrt{(1-p)} \ket{\bot^n}\ket{1}\ket{0^k}
  \end{align}
  where $k$ is chosen such that the dimension of $\ket{G_p}$ equals the dimension of $\ket{\psi^{R_H, R_E}}$, since $\sigma = \Tr_{[n]\setminus H}\ketbra{G_p}{G_p}$.
  
  But, by Uhlmann's theorem (see, e.g.\ \cite[Thm.~9.4]{NC10_QuantumComputationQuantum}), there exists two purifications $\ket{\phi_\sigma}$ and $\ket{\phi_\rho}$ of, respectively, $\sigma$ and $\rho$, such that
  \begin{align}
    F\left(\rho, \sigma\right) = |\braket{\phi_\sigma}{\phi_\rho}|\label{eq:Fisbracketphi}
  \end{align}
  Without loss of generality, we can append $\ket{0}$'s to $\ket{\phi_\sigma}$ and $\ket{\phi_\rho}$ to ensure their dimension is equal to those of $\ket{\psi}^{R_H, R_E}$ while maintaining the fact that $F\left(\rho, \sigma\right) = |\braket{\phi_\sigma}{\phi_\rho}|$. Moreover, because $\ket{\phi_\rho}$ and $\ket{\psi}^{R_H, R_E}$ (resp.\ $\ket{\phi_\sigma}$ and $\ket{G_p}$) have the same reduced density matrix\footnote{It is a standard exercise to show, using the Schmidt decomposition, that two states with the same dimension and the same purification are equal up to a unitary applied on the purification register.}, there exists $U_\rho$ (resp.\ $U_\sigma$) such that $\ket{\psi}^{R_H,R_E} = (I^{R_H} \otimes U_\rho) \ket{\phi_\rho}$ (resp.\ $\ket{G_p} = (I \otimes U_\sigma) \ket{\phi_\sigma}$). Therefore:
  \begin{align}
    \begin{split}
      &\TD(((\CheckState \circ \Merge)^{R_H} \otimes I^{R_E})\ket{\psi}^{R_H, R_E},\ket{\psi}^{R_H, R_E})
    \end{split}\\
    \begin{split}
      &= \TD(((\CheckState \circ \Merge)^{R_H} \otimes U_\rho^{R_E})\ket{\phi_\rho},(I^{R_H} \otimes U_\rho^{R_E}) \ket{\phi_\rho})
    \end{split}\\
    \intertext{Then, using the triangle inequality twice we get:}
    \begin{split}
      &\leq
        \TD(((\CheckState \circ \Merge)^{R_H} \otimes U_\rho^{R_E})\ket{\phi_\rho},((\CheckState \circ \Merge)^{R_H} \otimes U_\rho^{R_E})\ket{\phi_\sigma})\\
      &\quad + \TD(((\CheckState \circ \Merge)^{R_H} \otimes U_\rho^{R_E})\ket{\phi_\sigma}, (I^{R_H} \otimes U_\rho^{R_E}) \ket{\phi_\sigma})\\
      &\quad + \TD((I^{R_H} \otimes U_\rho^{R_E}) \ket{\phi_\sigma}, (I^{R_H} \otimes U_\rho^{R_E}) \ket{\phi_\rho})
    \end{split}\\
    \intertext{Then, using the fact that TD is symmetric, cannot be increased with post-processing, and is left unchanged when adding/removing/replacing a unitary on both inner terms (needed to replace $U_\rho$ with $U_\sigma$), we can simplify it as:}
    \begin{split}
      &\leq
        2\TD(\ket{\phi_\rho},\ket{\phi_\sigma}) + \TD(((\CheckState \circ \Merge)^{R_H} \otimes U_\sigma^{R_E})\ket{\phi_\sigma}, (I^{R_H} \otimes U_\sigma^{R_E})\ket{\phi_\sigma})
    \end{split}\\
    \begin{split}
      &=
        2\TD(\ket{\phi_\rho},\ket{\phi_\sigma}) + \TD(((\CheckState \circ \Merge)^{R_H} \otimes I^{R_E})\ket{G_p}, \ket{G_p})
    \end{split}\label{eq:simplifiedorigstuff}
  \end{align}
  However, by the definition of the trace distance on pure state and by definition of $\ket{\phi_\rho}$ and $\ket{\phi_\sigma}$, we have:
  \begin{align}
    \TD(\ket{\phi_\rho},\ket{\phi_\sigma})
    &= \sqrt{1-|\braket{\phi_\rho}{\phi_\sigma}|^2}
      \eqeq{eq:Fisbracketphi} \sqrt{1-F\left(\rho, \sigma\right)^2}
      \eqeq[\leq]{eq:Fbound} \sqrt{1-|1-\eps|^2}
      = \sqrt{1-(1 + \eps^2 - 2\eps)}\\
    &= \sqrt{2\eps - \eps^2}\label{eq:valueoftd}
  \end{align}
  Moreover, we claim that $((\CheckState \circ \Merge)^{R_H} \otimes I^{R_E})\ket{G_p} = \ket{G_p}$ and that therefore:
  \begin{align}
    \TD(((\CheckState \circ \Merge)^{R_H} \otimes I^{R_E})\ket{G_p}, \ket{G_p}) = 0\label{eq:tdisnull}
  \end{align}
  This can be seen by starting from the definition of $\ket{G_p}$ (the proof of this claim ends at \cref{eq:endOfSubclaim}):
  \begin{align}
    ((\CheckState \circ \Merge)^{R_H} \otimes I^{R_E})\ket{G_p}
    &= ((\CheckState \circ \Merge)^{R_H} \otimes I^{R_E}) (\sqrt{p} \ket{G}\ket{0}\ket{0^k} + \sqrt{1-p} \ket{\bot^n}\ket{1}\ket{0^k})
  \end{align}
  Then, since $\Merge$ behave as identity when the state contains $\bot$, and similarly $\CheckState$ is identity if the state does not contain $\bot$, this can be rewritten as:
  \begin{align}
    ((\CheckState \circ \Merge)^{R_H} \otimes I^{R_E})\ket{G_p} &= \sqrt{p} (\Merge^{R_H} \otimes I^{R_E})\ket{G}\ket{0}\ket{0^k} + \sqrt{1-p} (\CheckState^{R_H} \otimes I^{R_E})\ket{\bot^n}\ket{1}\ket{0^k}
  \end{align}
  Finally $\CheckState \ket{\bot^{|H|}} = \ket{\bot^{|H|}}$ and since any graph state is mergeable (\cref{def:mergeable} and \cref{thm:anygraphismergeable}), we have:
  \begin{align}
    ((\CheckState \circ \Merge)^{R_H} \otimes I^{R_E})\ket{G_p} &= \sqrt{p} \ket{G}\ket{0}\ket{0^k} + \sqrt{1-p} \ket{\bot^n}\ket{1}\ket{0^k} = \ket{G_p}\label{eq:endOfSubclaim}
  \end{align}
  finishing the proof of our claim $((\CheckState \circ \Merge)^{R_H} \otimes I^{R_E})\ket{G_p} = \ket{G_p}$.
  
  Now, we combine \cref{eq:tdisnull} and \cref{eq:valueoftd} into \cref{eq:simplifiedorigstuff} to get:
  \begin{align}
    \TD(((\CheckState \circ \Merge)^{R_H} \otimes I^{R_E})\ket{\psi}^{R_H, R_E},\ket{\psi}^{R_H, R_E}) \leq 2\sqrt{2\eps - \eps^2}
  \end{align}
  As discussed previously, the advantage in distinguishing these two states directly gives an upper bound on the advantage of distinguishing the ideal world from the real world (since these two states are actually the states held by the distinguisher at the end of its interaction with the ideal resource or the real protocol). But it is a well known fact that the best probability of distinguishing $\ket{\phi}$ from $\ket{\psi}$ is $\frac{1}{2} (1+\TD(\ket{\phi},\ket{\psi}))$, i.e.\ the advantage in distinguishing $\ket{\psi}$ from $\ket{\phi}$ is actually $\TD(\ket{\phi},\ket{\psi})$, so if we define:
  \begin{align}
    \eps' \eqdef 2\sqrt{2\eps - \eps^2}
  \end{align}
  then we have $\pi_H \cR \approx_{\eps'}\filter_{H} \Verif^f \Sim_M$. For any subset of corrupted party $M$, we are able to construct such a simulator $\Sim_M$. Following \cref{def:securityDef}, $\Pi$ $\eps'$-realizes the functionality $\Verif[\ket{G}]^f$, concluding our proof.
\end{proofE}

Note that while we focus in this proof only on graph states, the result in this section also easily extend to any mergeable state.

%%%%%%%%%%%%
%%%%%%%%%%%%

\section{Realizing the ideal resource without corrections}\label{sec:withoutCorr}
\textEnd{\subsection{Proofs of \cref{sec:withoutCorr}}}

We showed in \cref{thm:mainThm} that all verification protocols, provided some basic security properties, realize without any change the resource $\Verif^f$. However, this resource does allow the adversary to apply some corrections on the state of the honest parties. While we believe that such corrections are harmless for the security as they can already, to a certain degree, be created by applying some operations on the state owned by the malicious party, one might prefer a ``cleaner'' interface for the resource, at the cost of having a slightly more involved protocol. In this section, we show that any protocol realizing $\Verif^f$ can be turned into a new protocol that realizes a new resource $\Verif$ (note that the dependency on $f$ is gone). 

The modification consists in applying a random stabilizer of $G$ on their share of the graph state at the end of the protocol. We abstract the sampling and distribution of this stabilizer in a separate, additional, resource. In practice, this operation can for instance be done via a coin flipping protocol.

We formally describe the resource $\Verif$ in \cref{protoc:idealResourceVerifNoCorrections}. Apart from allowing malicious adversaries to abort and to get their share of the quantum state before honest parties, it simply sends $\ket{G}$ to all parties. We also define the filters $\filter_i$ for each party $i$ that sends $c_i = \top$ and we show an informal representation of the filtered resource $\bot_{[n]\cup \source}\Verif$ in \cref{fig:CuteFiltered}. \\

\begin{resource}{Ideal resource $\Verif$}\label{protoc:idealResourceVerifNoCorrections}
  \begin{enumerate}
  \item Create a quantum state $\ket{G}$.
  \item Receive for each party $i$ a bit $c_i \in \{\top,\bot\}$. If any $c_i = \bot$, send the $i$-th qubit of $\ket{G}$ to party $i$, and wait for another abort bit $c_i' \in \{\top,\bot\}$: if any $c_i' = \bot$, abort by sending $\bot$ on all interfaces; otherwise, for each $i$ such that $c_i = \top$, send the $i$-th qubit of $\ket{G}$ to party $i$.
  \end{enumerate}
 Additionally, we include in $\Verif$ an additional communication channel $\cC$ that can forward any\===possibly quantum\===message between any party (see \cref{rq:GlobalSimu}).
\end{resource}

 \begin{figure}[!ht]
     \centering
     \includegraphics[width=10cm]{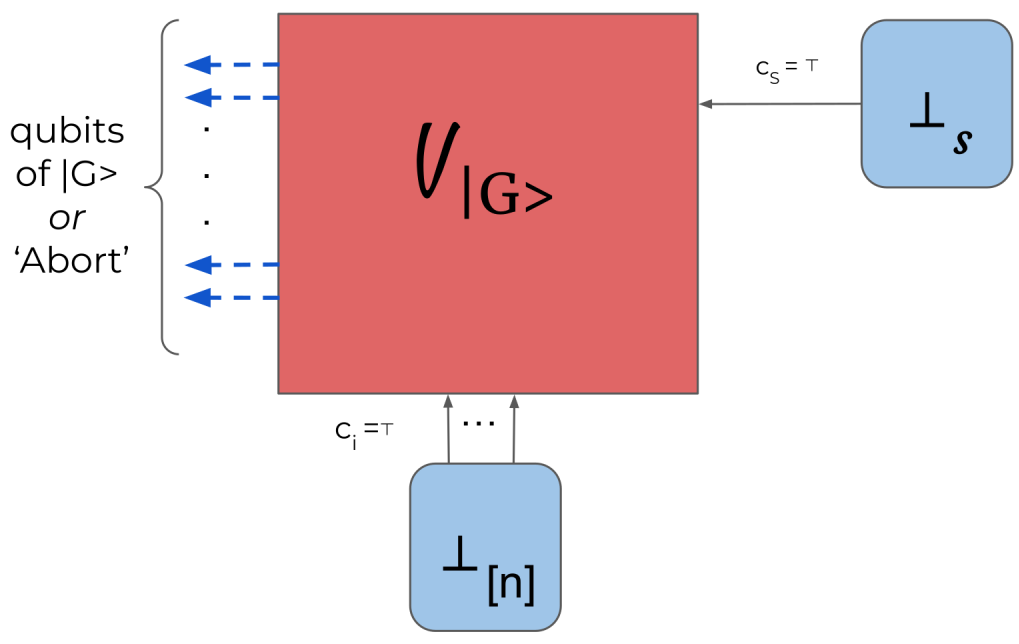}
     \caption{Ideal filtered resource $\pi_{[n]\cup \source}\Verif$ representing a more "clean" ideal graph state verification resource than in the previous section.}
     \label{fig:CuteFiltered}
 \end{figure}

We define in \cref{protoc:coinFlip} the ideal functionality that samples a random $x$, abstracting a common randomness generation resource. Note that one might want a simpler version of \cref{protoc:coinFlip} where the random $x$ is simply sent to all parties without considering any abort. In practice such protocols are impossible to obtain without a trusted third party, as a malicious adversary can usually first check the value of the random bit, and abort before letting the other party aware of this value. For this reason, we need to add an exchange of messages to allow the adversary to abort without letting the other parties know the value of the random bit.

\begin{resource}{Ideal resource $\CoinFlip$\label{protoc:coinFlip}}
  \begin{enumerate}
  \item Sample a random bit string $x \sample \{0,1\}^n$
  \item For each party $i \in [n]$, receive $c_i \in \{\top,\bot\}$. If any $c_i = \bot$, send $x$ to party $i$, and wait for another abort bit $c_i' \in \{\top,\bot\}$: if $c_i' = \bot$, abort by sending $\bot$ on all interfaces, and otherwise send $x$ to all parties.
  \end{enumerate}
  We define the associated filter that, connected to interface $i$, sets $c_i = \top$ and forwards the $x$ to the outer interface.
\end{resource}

We define in \cref{protoc:realizeVerifNoCorr} the protocol $\{\tau_i\}_{i \in [n]}$ that realizes $\Verif$ from $\Verif^f$ and $\CoinFlip$. %
\begin{protocol}{$\{\tau_i\}_{i \in [n]}$ Realizing $\Verif$ from $\Verif^f$ and $\CoinFlip$\label{protoc:realizeVerifNoCorr}}
  \begin{enumerate}
  \item For each honest party $i\in H$, $\tau_i$ receives from $\bot_{H} \Verif^f$ a share of a quantum state $\ket{G}$ or an abort message (in which case they output $\bot$ and abort).
  \item Each honest party $\{\tau_i\}_{i\in H}$ asks from $\CoinFlip$ a message $x$ or an abort bit (in which case they output $\bot$ and abort).
  \item Each honest party $\{\tau_i\}_{i\in H}$ applies the stabilizer $\gateX^{x_i}\gateZ^{(Gx)_i}$ on their qubit, and outputs the resulting qubit.
  \end{enumerate}
\end{protocol}
We prove now that \cref{protoc:realizeVerifNoCorr} realizes $\Verif$.

\begin{theoremE}[][end, restate,text link section full proof]\label{thm:verifNoCorr}
  \cref{protoc:realizeVerifNoCorr} realizes $\Verif$.
\end{theoremE}
\needspace{3cm}
\begin{proof}[Proof sketch]
  The full proof is in appendix, but the main idea of the security proof is as follows: first, the simulator will send to the distinguisher the state sent by $\Verif$ after partially applying on it a random stabilizer $x$. Then, the simulator will receive back a set of corrections from the distinguisher to apply on the side of the honest parties: since the simulator cannot apply these corrections as the qubits are on the side of the functionality, we use the fact that for stabilizer states, we can instead apply these corrections on the side of the adversary. Unfortunately, the simulator has also no access to the quantum register of the adversary: instead, the simulator adds this additional correction to $x$, and sends this to the distinguisher, pretending it was the random stabilizer sampled by $\CoinFlip$. We provide an extensive analysis of this simulator in the  \hyperref[proof:prAtEnd\pratendcountercurrent]{full proof} in \pratendSectionlikeCref.
\end{proof}
\begin{proofE}[text link=]~\\ As before, we will follow the Security \cref{def:securityDef} and find simulators for each possible dishonest behaviour. We show the general case in \cref{fig:Thm41}.

\begin{figure}[!ht]
  \centering
  \begin{subfigure}[b]{0.49\textwidth}
    \centering
    \includegraphics[width=\textwidth]{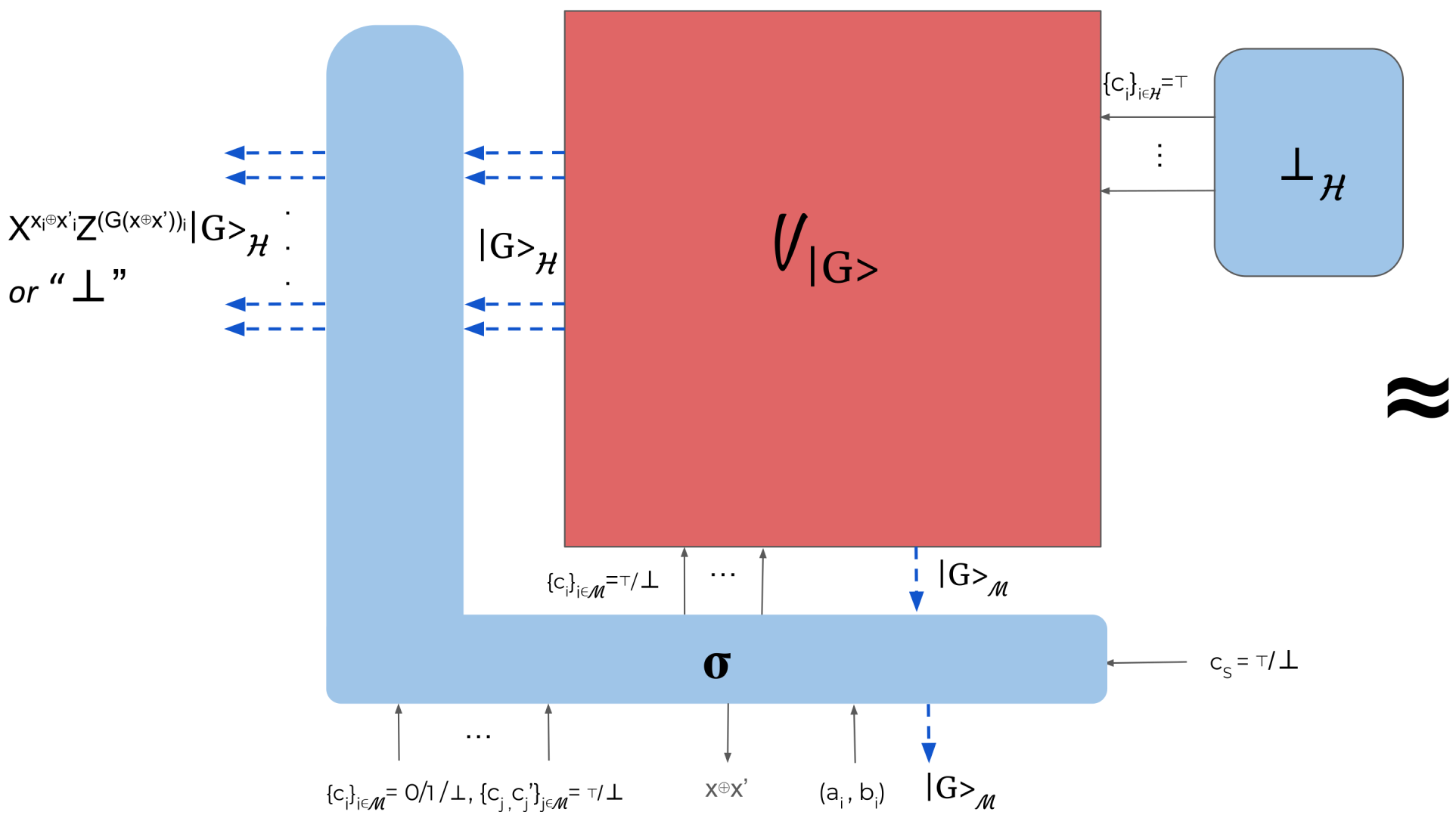}
  \end{subfigure}
  \begin{subfigure}[b]{0.49\textwidth}
    \centering
    \includegraphics[width=\textwidth]{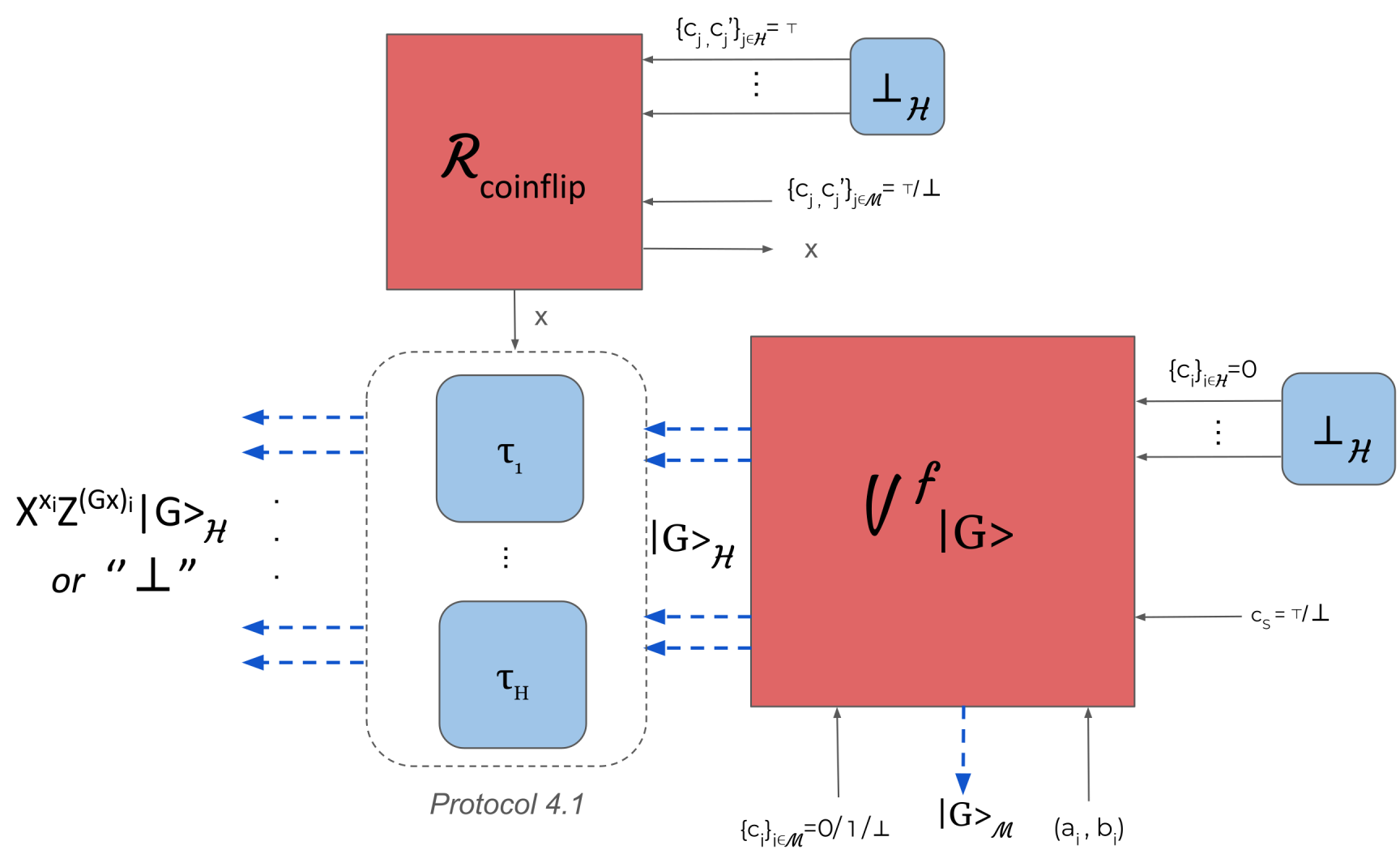}
  \end{subfigure}
   \caption{General case of the security proof of \cref{thm:verifNoCorr}}
  \label{fig:Thm41}
\end{figure}
  \paragraph{Case 1: Correctness.}
  We first need to prove correctness, i.e.\ $ \bot_{[n] \cup \source} \Verif{}\approx \tau_{[n]} \filter_{[n]} (\Verif^f \| \CoinFlip) \filter_{\source} $. If all parties are honest, neither $\Verif^f$ nor $\CoinFlip$ will abort, so all parties will receive a part of $\ket{G}$ with the same $x$. Then, since each party $i$ applies $\gateX^{x_i}\gateZ^{(Gx)_i}$, i.e.\ a part of a stabilizer, the overall resulting state is $\gateX^{x}\gateZ^{Gx}\ket{G} \eqeqLem{lem:stabilizerDiag} \ket{G}$. This is exactly the state obtained in $\bot_{[n] \cup \source} \Verif$. Hence, there exists no distinguisher able to differentiates the two resources, concluding the proof of correctness.

  \paragraph{Case 2: Security.}

  To prove the security of the protocol, we want to show that for any partition $(H,M)$ of  $[n] \cup \{\source\}$, there exists a simulator $\sigma_M$ such that $\bot_{H}\sigma_M \Verif \approx \tau_{H} \bot_H (\Verif^f \| \CoinFlip)$. To match the usual Abstract cryptography terminology, we will denote the left-hand side $\bot_{H}\sigma \Verif$ as the ideal world and the right-hand side $\tau_{H} \bot_H (\Verif^f \| \CoinFlip)$ as the real world. As before, using the classical channel allowing communication between malicious parties in $\Verif^f$, we can consider all local simulators as a single simulator. Let us define the simulator $\sigma$ formally as in \cref{protoc:simulatortocute}. \\
  
  \begin{protocol}{Simulator $\sigma$\label{protoc:simulatortocute}}   
    \begin{enumerate}
    \item First, $\sigma$ sends $c_i = \bot$ to $\Verif$ in order to receive a share $\ket{G}_M$ of $\ket{G}$ for all malicious parties in the set $M$ of corrupted parties.
    \item Then, $\sigma$ receives $c_i$ from each malicious party $i$ (in the real world these messages are sent to $\Verif^f$ by the distinguisher). If any $c_i$ is equal to $\bot$, then $\sigma$ also aborts after sending $c_i' = \bot$ to $\Verif$ to make it abort as well.
    \item $\sigma$ samples a random stabilizer $x \sample \{0,1\}^n$, applies $\gateX^{x_{M'}}\gateZ^{(G_M' x)_{M'}}$ where $M' \eqdef \{i \mid c_i = 1\}$ on the received $\ket{G}_M$ and sends the resulting qubits of parties in $M'$ to the distinguisher.
    \item Then, $\sigma$ receives a set of corrections $(a_i, b_i) \in (\{0,1\}^{|H|})^2$ for $i \in M'$, computes $x'_H \eqdef \xor_i a_i$  and $z'_H \eqdef \xor_i b_i$, and checks that $f_G(M', x'_H, z'_H) = \top$ (otherwise it sends $c_i' = \bot$ to $\Verif$ for all $i$ and aborts).
    \item Since $f_G(M', x'_H, z'_H) = \top$, let $b$ be like in \cref{thm:mainThm}, and let us define $x'_M \eqdef (V^T)^{-1}
      \begin{bmatrix}
        b\\
        \mathbf{0}
      \end{bmatrix}
      $, and $x' \eqdef
      \begin{bmatrix}
        x'_H\\
        x'_M
      \end{bmatrix}
      $.
    \item If the set of corrupted parties contains no party in $[n]$, the simulator can stop. Otherwise, it will receive bits $c_i''$ (sent to $\CoinFlip$ in the real world and called $c_i$ there). It will then send $x \xor x'$ to parties that sent $c_i'' = \bot$, wait for a bit from each party (sent to $\CoinFlip$ in the real world and denoted $c_i'$), and abort if one of these party sent $\bot$, by sending $c_i' = \bot$ to the ideal functionality for all $i \in M$. Finally, it broadcasts $x$ to all parties such that $c_i = 1$.
    \item For each party $i$ such that $c_i = 0$, $\sigma$ will apply on the $i$-th qubit (that remains from the second step) the operation $\gateX^{x_i \xor x_i'}\gateZ^{(G (x \xor x'))_i}$, and output the resulting qubit to the distinguisher.
    \item Finally, it outputs $c_i' = \top$ to $\Verif$ in order to let the functionality broadcast $\ket{G}_{\bar{M}}$ to all honest parties.
    \end{enumerate}

  \end{protocol}
  By construction, this simulator follows exactly the abort pattern of the real world protocol, and the exchanged messages are identical. So without loss of generality, we can assume that the distinguisher never aborts\footnote{We can always turn a distinguisher that aborts into a distinguisher that does not abort while increasing the probability of distinguishing, for instance by intercepting all abort messages from the original distinguisher, and continuing instead the protocol with arbitrary non-aborting inputs while sending to the original distinguisher the messages that would have been sent in both the ideal and the real protocol.}, and in particular sends valid corrections. Similarly, we can assume that the distinguisher is always sending $c_i = 1$ since the case $c_i = 0$ can easily be simulated by the distinguisher given the transcript for $c_i = 1$ since it contains more information. Under this assumption, both the real and ideal worlds can significantly be simplified:
  \begin{itemize}
  \item In the real world, the distinguisher receives first $\ket{G}_M$, then send some (valid) corrections defining $x'_H = \xor_i a_i$ and $z'_H \eqdef \xor_i b_i$, and finally receives the remaining qubits $\gateX^{x_H \xor x_H'} \gateZ^{z'_H \xor (G x)_H}\ket{G}_H$ together with $x$ that was sampled uniformly at random.
  \item In the ideal world, the distinguisher receives first $\gateX^{x_M}\gateZ^{(G x)_M}\ket{G}_M$ where $x$ is sampled uniformly at random, then it sends some (valid) corrections defining $x'_H = \xor_i a_i$ and $z'_H \eqdef \xor_i b_i$, and finally receives the remaining qubits $\ket{G}_H$ together with $x \xor x'$ where $x'_M \eqdef (V^T)^{-1}
    \begin{bmatrix}
      b\\
      \mathbf{0}
    \end{bmatrix}
    $, and $x' \eqdef
    \begin{bmatrix}
      x'_H\\
      x'_M
    \end{bmatrix}
    $.
  \end{itemize}
  First, we can realize that since $\gateX^{x_M}\gateZ^{(G x)_M}\ket{G} = \gateX^{x_M}\gateZ^{(G x)_M}(\gateX^{x }\gateZ^{G x} \ket{G}) = \gateX^{x_H}\gateZ^{(G x)_H}\ket{G}$, the ideal world is indistinguishable from a world where the distinguisher receives $\ket{G}_M$ while the functionality applies $\gateX^{x_H}\gateZ^{(G x)_H}$ on the remaining qubits. This gives this equivalent hybrid system:
  \begin{itemize}
  \item The distinguisher receives first $\ket{G}_M$ then it sends some (valid) corrections, defining $x'_H = \xor_i a_i$ and $z'_H \eqdef \xor_i b_i$, and finally receives the remaining qubits $\gateX^{x_H}\gateZ^{(G x)_H}\ket{G}_H$ (where $x$ is sampled uniformly at random) together with $x \xor x'$ where $x'_M \eqdef (V^T)^{-1}
    \begin{bmatrix}
      b\\
      \mathbf{0}
    \end{bmatrix}
    $, and $x' \eqdef
    \begin{bmatrix}
      x'_H\\
      x'_M
    \end{bmatrix}
    $.
  \end{itemize}
  Then, we can define $\hat{x} \eqdef x \xor x'$. Since $x$ is independent of $x'$, the probability of sampling any $\hat{x}$ is equal to the probability of sampling any $x$, we can therefore sample $x'$ instead of $x$, and use the fact that $x = \hat{x} \xor x'$. So, the above system can be turned into the following indistinguishable system:
  \begin{itemize}
  \item The distinguisher receives first $\ket{G}_M$ then it sends some (valid) corrections defining $x'_H = \xor_i a_i$ and $z'_H \eqdef \xor_i b_i$, and finally receives the remaining qubits $\gateX^{\hat{x}_H \xor x'_H}\gateZ^{(G (\hat{x} \xor x'))_H}\ket{G}_H$ (where $\hat{x}$ is sampled uniformly at random) together with $\hat{x}$ where $x'_M \eqdef (V^T)^{-1}
    \begin{bmatrix}
      b\\
      \mathbf{0}
    \end{bmatrix}
    $, and $x' \eqdef
    \begin{bmatrix}
      x'_H\\
      x'_M
    \end{bmatrix}
    $.
  \end{itemize}
  Finally, we have:
  \begin{align}
    (G x')_H
    &= \begin{bmatrix}
         G_H & \Gamma^T
       \end{bmatrix} x' =
      G_H x'_H \xor U^T
      \begin{bmatrix}
        I_r & 0\\
        R^T & 0
      \end{bmatrix}V^T x'_M
      = G_H x'_H \xor \left(U^T
      \begin{bmatrix}
        I_r & 0\\
        R^T & 0
      \end{bmatrix}V^T\right)\left((V^T)^{-1}
      \begin{bmatrix}
        b\\
        \mathbf{0}
      \end{bmatrix}\right) \\
    &=
      G_H x'_H \xor
      U^T \begin{bmatrix}
            b\\
            R^T b
          \end{bmatrix}
      \eqeqDef{thm:mainThm}
      G_H x'_H \xor
      U^T (U^T)^{-1} (z'_H \xor G_H x'_H) =
      z'_H
  \end{align}
  Therefore, we get $\gateX^{\hat{x}_H \xor x'_H}\gateZ^{(G (\hat{x} \xor x'))_H}\ket{G}_H = \gateX^{\hat{x}_H \xor x'_H}\gateZ^{z'_H \xor (G \hat{x})_H}\ket{G}_H$. Hence, this last system is actually equal to the real world, which concludes the indistinguishability proof.
\end{proofE}

\begin{corollaryE}[][end, restate]\label{cor:mainThmWithoutCorrections}
  Assuming the existence of a protocol for graph state verification fulfilling properties described in \cref{thm:mainThm} and a coin flipping protocol realizing \CoinFlip, there exists a protocol realizing $\Verif$.
\end{corollaryE}
\begin{proofE}
  This is a direct consequence of \cref{thm:mainThm} and \cref{thm:verifNoCorr}, where we run in \cref{protoc:realizeVerifNoCorr} the graph state verification instead of $\Verif^f$ and the coin flipping protocol instead of \CoinFlip. 
  
  More precisely, in \cref{thm:mainThm}, we first proved the equivalence $\pi_{[n] \cup \source} \cR \approx \filter_{[n] \cup \source} \Verif^f$ between a concrete verification protocol and an ideal, but gruesome, resource. Then, in \cref{thm:verifNoCorr}, we explicited the protocol $\{\tau_i\}_{i \in [n]}$, that applies a random stabilizer to the output state of the verification protocol, to construct a simpler functionality $\Verif$ from this gruesome resource alongside with a coin flipping resource. In particular, we proved the equivalence $ \bot_{[n] \cup \source} \Verif{}\approx \tau_{[n]} \filter_{[n]} (\Verif^f \| \CoinFlip) \filter_{\source} $.
  
  Let $\pi'_{[n]}$ be a concrete protocol securely realising the ideal functionality $\CoinFlip$ using resource $\cR'$. We have, in particular, that $\pi'_{[n]}\cR' \approx \bot_{[n]}\CoinFlip$. By composability of resources within the AC framework, we get the correctness equivalence :
  \begin{equation}
  \pi_{[n] \cup \source}\circ\tau_{[n]}  (\cR \| \pi'_{[n]}\cR') \approx \bot_{[n] \cup \source} \Verif{} . \end{equation}
  
  For all subset of malicious parties $M$, we can similarly appropriately compose the simulators from the proofs of the secure construction of $\CoinFlip$, $\Verif^f$ and $\Verif$ to obtain simulators $\sigma_M$. We can use these simulators and the composability of AC to similarly prove all the equivalences necessary in the Security \cref{def:securityDef}. This proves the secure construction of the ideal functionality $\Verif$ by a graph state verification protocol fulfilling properties described in \cref{thm:mainThm} and a coin flipping protocol realizing \CoinFlip.
\end{proofE}

\needspace{5cm}
\section{Use cases}
\label{sec:UseCases}
\textEnd{\subsection{Proofs of \cref{sec:UseCases}}}

Our result readily applies to some already existing protocols in the literature. In the following we show two such use cases and discuss their implication and improvement over previous works.

\subsection{Generic translation tools}

Before studying specific protocols, it is handy to first derive some generic theorems in order to translate between the notations used in existing graph state verification protocols and the notion of distance used in \cref{def:graphStateVerifProtoc}. Indeed, while our notion of distance uses a single parameter $\eps$ quantifying the average distance to the ideal graph state (which is handy since in AC we have a single parameter $\eps'$ used to denote the distance between two systems $A \approx_{\eps'} A$), existing works consider instead two parameters, namely the probability $\delta$ of being $\eta$-close to the ideal graph state.

\begin{remark}\label{rq:simultaneousAbortion}
  We also emphasize that most existing protocols implicitly assume that either all parties abort or that they all accept the final state together. We will refer to that property as \emph{simultaneous abortion}. While \cref{def:graphStateVerifProtoc} does not strictly enforce this behavior, any protocol that does not fulfill simultaneous abortion with high probability will also have low guarantee in term of security. Said differently, the $\eps$ obtained in \cref{def:graphStateVerifProtoc} and \cref{thm:mainThm} will be far from $0$. This makes sense since our resource \cref{protoc:idealResourceVerif} always sends the same abort/accept bit to all parties. So, if the adversary can send different abort bits to different parties, this gives directly a simple way to distinguish the ideal world from the real world, by simply checking if all parties share the same abort bit.

  This could be a real issue in practice, for instance in~\cite{MEVresistant} if a malicious verifier is picked, the verifier could send different abort bits to all parties. As a result, without further checks, this protocol could only be proven composably secure for some \emph{constant} $\eps$ (as a reminder, in the original protocol, we expect to be able to get $\eps$ as small as wanted, the running time scaling polynomially with $O(1/\eps)$). We have two options to avoid this issue:
  \begin{itemize}
  \item Either change the definition of $\Verif$ to get rid of simultaneous abortion, meaning that some honest parties could abort while others would not.
  \item Or we could slightly adapt the protocol using a \emph{broadcast channel}, i.e.\ a channel where all parties receive the same bit, in order to notify to all parties at the same time the final abort bit. 
  \end{itemize}
  Since the first option makes the resource harder to use, we will opt for the second solution. This makes even more sense as in existing protocols, the protocol already needs to obtain a random string known to all parties, which often already implicitly requires a broadcast channel. One might also ask whether, reciprocally, a broadcast channel is needed to realize $\Verif$. It turns out that $\Verif^f$ can already be (ab)used to obtain a kind of broadcast channel: since all parties see the same abort bit, one could decide that ``abort'' means 0 and that ``accept'' means 1. Therefore, it should come at no surprise that a broadcast channel is needed to realize $\Verif$.
\end{remark}

\begin{lemmaE}[][end, restate]\label{lem:equivalentFormulationA}
  Let $\Pi = \{\pi_i\}_{i \in [n] \cup \source}$ be a protocol generating, when all parties are honest, a state $\ket{G}$ shared among all parties but the source. We assume that $\Pi$ has simultaneous abortion (\cref{rq:simultaneousAbortion}), i.e.\ that for any adversary $\cA$, either all honest parties abort at the same time with some probability $1-p$ or accept and output the averaged state $\rho_\top$. Then, if \emph{any} of the following conditions is fulfilled, this protocol is an $\eps$-graph state verification protocol according to \cref{def:graphStateVerifProtoc}:
  \begin{itemize}
  \item If $p(1-F(\rho_\top, \Tr_{[n]\setminus H}(\ketbra{G}{G}))) \leq \eps$. This denotes the fact that the probability of accepting and outputting a state far from $\ket{G}$ to the honest parties is small.
  \item Or if $F(\rho_\top, \Tr_{[n]\setminus H}(\ketbra{G}{G})) \geq 1 - \eps$. This corresponds to the protocol's property to create a state close to the desired state. Note that this condition is strictly stronger than the first one, since an adversary might be able to produce such a state with negligible probability. Yet, unconditionally secure protocols might prefer this formulation.
  \end{itemize}  
\end{lemmaE}
\begin{proofE}
  Since the protocol is correct by assumption, the first point of \cref{def:graphStateVerifProtoc} is trivially fulfilled. We focus now on proving that the protocol is secure according to \cref{def:graphStateVerifProtoc}. Since the protocol has simultaneous abortion, the first part of the security statement is trivially true, hence we just need to find $\sigma$ such that \cref{eq:graphStateVerifProtocSecurity} is true. Let $\rho \eqdef \Tr_{[n] \setminus H}(\pi_{H} \cR \cA)$ be the averaged state obtained at the end of the interaction with $\cA$. Then, since $\pi$ has simultaneous abortion, we know that there exists $p \eqdef \bra{\bot^{|H|}} \rho \ket{\bot^{|H|}}$ (actually this is the same $p$ as in the lemma) such that:
  \begin{align}
    \rho = (1-p)\ketbra*{\bot^{|H|}}{\bot^{|H|}} + p\rho_\top
  \end{align}
  where $\rho_\top$ is a normalized state orthogonal to $\ketbra*{\bot^{|H|}}{\bot^{|H|}}$, corresponding to the averaged state obtained when no party aborted.\\

Let us define
  \begin{align}
    \sigma \eqdef p\Tr_{[n]\setminus H}(\ketbra{G}{G}) + (1-p) \ketbra*{\bot^{|H|}}{\bot^{|H|}}  
  \end{align}
  and find $\eps$ such that $F(\rho,\sigma) \geq 1-\eps$, or, equivalently, such that $1-F(\rho,\sigma) \leq \eps$. To compute this quantity, we can first use the strong concavity of fidelity (\cite[Thm.~9.7]{NC10_QuantumComputationQuantum}) that states that $F(\sum_i p_i \rho_i ,\sum_i q_i \sigma_i) \geq \sum_i \sqrt{p_i q_i} F(\rho_i ,\sigma_i)$. Therefore, we have:
  \begin{align}
    1-F(\rho,\sigma)
    &\leq 1-(1-p)F(\ketbra*{\bot^{|H|}}{\bot^{|H|}}, \ketbra*{\bot^{|H|}}{\bot^{|H|}}) - p F(\rho_\top, \Tr_{[n]\setminus H}(\ketbra{G}{G}))\\
    & = p(1-F(\rho_\top, \Tr_{[n]\setminus H}(\ketbra{G}{G})))\label{eq:poneminusf}
  \end{align}
  where the last equality is a direct consequence of $F(\ketbra{\phi}{\phi},\ketbra{\phi}{\phi}) = 1$ when $\ket{\phi}$ is a pure state. If any of the first two assumptions of this lemma are fulfilled, we can directly inject it in \cref{eq:poneminusf} to obtain \cref{eq:graphStateVerifProtocSecurity}. This shows that the protocol is an $\eps$-graph state verification protocol according to \cref{def:graphStateVerifProtoc} (for the second assumption, we upper bound $p$ by $1$).
\end{proofE}

\begin{lemmaE}[][end, restate]\label{lem:translationFidelityProba}
  If a protocol has simultaneous abortion (\cref{rq:simultaneousAbortion}), and if the probability (on the randomness of $\cA$ and the whole protocol) to have no abort and a final state far from the target $\ket{G}$ is small, more formally:
  \begin{align}
    \pr[\ket{\psi_i} \gets \pi_H \cR \cA]{\Tr_{[n] \setminus H}\ket{\psi_i} \neq \ket{\bot^H} \land \sqrt{1-F^2(\Tr_{[n] \setminus H}\ketbra{\psi_i}{\psi_i}, \Tr_{[n] \setminus H}\ketbra{G}{G})} \geq \eta} \leq \delta\label{eq:assumptionVersionF}
  \end{align}
  or, equivalently,
  \begin{align}
    \pr[\ket{\psi_i} \gets \pi_H \cR \cA]{\Tr_{[n] \setminus H}\ket{\psi_i} \neq \ket{\bot^H} \land \min_U \TD((I^H \otimes U^M)\ket{\psi_i}, \ket{G}) \geq \eta} \leq \delta\label{eq:assumptionVersionMin}
  \end{align}  
  then this protocol is an $(\delta + \eta^2)$-graph state verification protocol according to \cref{def:graphStateVerifProtoc}.
\end{lemmaE}
\begin{proofE}
  First, the fact that these two definitions are equivalent comes from the fact that for pure states, $\TD(\ket{\psi},\ket{\phi}) = \sqrt{1-|\braket{\psi}{\phi}|^2}$ (\cite[Eq.~9.99]{NC10_QuantumComputationQuantum}), hence:
  \begin{align}
    \min_U \TD((I^H \otimes U^M)\ket{\psi_i}, \ket{G})
    & = \min_U \sqrt{1-|\bra{G}(I^H \otimes U^M)\ket{\psi_i}|^2}\\
    & = \sqrt{1-\max_U |\bra{G}(I^H \otimes U^M)\ket{\psi_i}|^2}\\
    & = \sqrt{1-F^2(\Tr_{[n] \setminus H}\ketbra{\psi_i}{\psi_i}, \Tr_{[n] \setminus H}\ketbra{G}{G})}
  \end{align}
  where the last equality comes from Uhlmann's theorem (\cite[Ex.~9.15]{NC10_QuantumComputationQuantum}), and additionally remarking that all purifications are equal up to a local unitary on the purified space. In the following, we will therefore only consider the assumption involving the fidelity.
  
  We will now simplify the assumption \cref{eq:assumptionVersionF}, but first let us define some notations. First, it will be handy to denote, like in \cref{lem:equivalentFormulationA}, $\rho_\top$ as the averaged normalized state obtained by honest parties assuming that the protocol has not aborted:
  \begin{align}
    \rho_\top \eqdef \esp*[\ket{\psi_i} \gets \pi_H \cR \cA\\ \Tr_{[n] \setminus H}\ket{\psi_i} \neq \ket{\bot^H}]{\Tr_{[n] \setminus H} \ketbra{\psi_i}{\psi_i}}
  \end{align}
  If we denote by $\{\ket{\psi_i}\}_i$ the set of all states producible by a $\pi_H \cR \cA$, and if $p_i$ represents the probability of outputting $\ket{\psi_i}$ assuming that it is not aborting ($\Tr_{[n] \setminus H}\ket{\psi_i} \neq \ket{\bot^H}$), we have therefore:
  \begin{align}
    \rho_\top = \Tr_{[n] \setminus H} \left(\sum_i p_i \ketbra{\psi_i}{\psi_i}\right) \label{def:rhotopassum}
  \end{align}
  To simplify further the notation we define for brevity:
  \begin{align}
    F_i & \eqdef F(\Tr_{[n] \setminus H}\ketbra{\psi_i}{\psi_i}, \Tr_{[n] \setminus H}\ketbra{G}{G})\label{eq:notationFi} \\
    F   & \eqdef F(\Tr_{[n] \setminus H}\ketbra{\psi}{\psi}, \Tr_{[n] \setminus H}\ketbra{G}{G})
  \end{align}  
  Using this notation, we can simplify the LHS of the assumption \cref{eq:assumptionVersionF} as:
  \begin{align}
    &\pr[\ket{\psi} \gets \pi_H \cR \cA]{\Tr_{[n] \setminus H}\ket{\psi} \neq \ket{\bot^H} \land \sqrt{1-F^2} \geq \eta}\\
    &= \pr*[\ket{\psi} \gets \pi_H \cR \cA]{\Tr_{[n] \setminus H}\ket{\psi} \neq \ket{\bot^H}}
      \pr*[\ket{\psi} \gets \pi_H \cR \cA\\\Tr_{[n] \setminus H}\ket{\psi} \neq \ket{\bot^H}]{\sqrt{1-F^2} \geq \eta}
  \end{align}
  By defining $p \eqdef \pr*[\ket{\psi} \gets \pi_H \cR \cA]{\Tr_{[n] \setminus H}\ket{\psi} \neq \ket{\bot^H}}$ as the probability of non aborting, this is equal to:
  \begin{align}
    p \pr*[\ket{\psi} \gets \pi_H \cR \cA\\\Tr_{[n] \setminus H}\ket{\psi} \neq \ket{\bot^H}]{\sqrt{1-F^2} \geq \eta}
  \end{align}
  Then, using the fact that the probability for $\ket{\psi}$ to be equal to $\ket{\psi_i}$ is $p_i$, we can simplify this as:
  \begin{align}
    p \pr*[\ket{\psi} \gets \pi_H \cR \cA\\\Tr_{[n] \setminus H}\ket{\psi} \neq \ket{\bot^H}]{\sqrt{1-F^2} \geq \eta}
    = p \sum _p p_i \pr*[]{\sqrt{1-F_i^2} \geq \eta}
    = p \sum_{p \mid \sqrt{1-F_i^2} \geq \eta} p_i
  \end{align}
  Therefore, assumption \cref{eq:assumptionVersionF} is equivalent to:
  \begin{align}
    p \sum_{i \mid \sqrt{1-F_i^2} \geq \eta} p_i \leq \delta \label{eq:boundCepssimpler}
  \end{align}
  
  Now, we can remark that using \cref{lem:equivalentFormulationA}, it is enough to show that $p(1-F(\rho_\top, \Tr_{[n]\setminus H}(\ketbra{G}{G}))) \leq \eps$ to prove that the protocol is an $\eps$-graph state verification protocol. Since we aim to upper bound $p(1-F(\rho_\top, \Tr_{[n]\setminus H}(\ketbra{G}{G})))$, we can simplify this expression as follows:
  \begin{align}
    & p(1-F(\rho_\top, \Tr_{[n]\setminus H}(\ketbra{G}{G})))\\
    & \eqeq{def:rhotopassum} p\left(1-F\left(\sum_i p_i \Tr_{[n] \setminus H} \ket{\psi_i}, \Tr_{[n]\setminus H}(\ketbra{G}{G})\right)\right)                          \\
    & \ \leq p\left(1-\sum_i p_i F_i\right) \tag{Concavity of fidelity (\cite[Ex.~9.20]{NC10_QuantumComputationQuantum})} \\
    & \ = p\sum_i p_i \left(1- F_i\right) \tag{$\sum_i p_i = 1$}\\
    & \ = p\sum_{i \mid \sqrt{1-F_i^2} \geq \eta} p_i \left(1- F_i\right) \nonumber + p\sum_{i \mid \sqrt{1-F_i^2} < \eta} p_i \left(1- F_i\right) \tag{Split the sum} \\
    & \ \leq \delta + p\sum_{i \mid \sqrt{1-F_i^2} < \eta} p_i \left(1- F_i\right) \tag{Fidelity is $\geq 0$ + \cref{eq:boundCepssimpler}}
  \end{align}
  Now, in the remaining terms, we assume $\sqrt{1-F_i^2} < \eta$, i.e.\ $1-F_i^2 < \eta^2$. But since $0 \leq F_i \leq 1$, we have $0 \leq F_i^2 < F_i$ and therefore $1-F_i \leq 1-F_i^2 < \eta^2$.
  Therefore, after injecting this in the previous equation, we get:
  \begin{align}
    p(1-F(\rho_\top, \Tr_{[n]\setminus H}(\ketbra{G}{G})))
    &\leq \delta + p\sum_{i \mid \sqrt{1-F_i^2} < \eta} p_i \eta^2\\
    &\leq \delta + \eta^2\tag{Upper bound probabilities by $1$}
  \end{align}
\end{proofE}

\subsection{GHZ-state verification}
 In~\cite{MEVresistant}, the authors develop and analyze an $n$-party verification protocol consisting only of classical communication and local quantum operations once the state is shared. One of the parties, called the \textit{Verifier}, has a central role in the protocol: it sends instructions to all parties and broadcasts the output of the verification. We recall the protocol of~\cite{MEVresistant} in~\cref{protoc:GHZprot}. In the AC language, it uses quantum and classical authenticated secret channels, a broadcast channel and a Common Random String resource (CRS) as concrete resources to build an ideal GHZ-state sharing functionality. Note that it is important that the random string is sent by the CRS once all parties received the state from the source. 

\begin{protocol}{Multipartite entanglement verification protocol\label{protoc:GHZprot}}
  \begin{enumerate}
  \item The source creates an $n$-qubit state locally equivalent to the GHZ state (up to local Hadamard and phase shift $\sqrt{Z}$ gate) and sends each qubit $i$ to party $i$.
  \item After receiving the state, the parties receive $r \sample \{0,1\}^S$ and $i \sample [n]$ from a CRS resource, where $S$ is a security parameter. Note that the source should not reveal $r$ and $i$ before the state is received by all parties.
  \item If $r=0\dots0$, the state received at step $1$ is outputted by each party and the protocol stops.
  \item Otherwise, if $r \neq 0\dots 0$, party $i$ is designed as the Verifier.
  \item The Verifier selects for each $i\in[n]$ a random input $x_{i}\in\{0,1\}$ such that $\sum_{i=1}^{n}x_{i}\equiv 0 \mod 2$ and sends it to the corresponding party via an authenticated private classical channel resource. The Verifier keeps one to themselves.
  \item If $x_{i}=0$, party $i$ performs a Z operation on their qubit. If $x_{i}=1$, party $i$ performs a Hadamard operation. 
  \item Each party $i$ measures their qubit in the $\{\ket{0},\ket{1}\}$ basis and sends their outcome $y_{i}$ to the Verifier via the classical channel.
  \item The Verifier accepts and broadcasts $b_{out}\eqdef 0$ if and only if 
    \begin{equation*}
      \sum_{i=1}^{n}y_{i}\equiv\frac{1}{2}\sum_{i=1}^{n}x_{i} \mod 2
    \end{equation*}
    Otherwise, the Verifier broadcasts $b_{out}=1$ and the protocol aborts.
  \item If the Verifier has not rejected, we restart from Step 1.
\end{enumerate}
\end{protocol}

This protocol has been extensively studied and presents desirable properties: notably the probability of aborting increases when the state sent by the source is further away from the target state (locally equivalent to) $\ket{GHZ}$, and this probability can be made arbitrary small by increasing the number of rounds. More precisely, for one round, if $\rho$ denotes the state shared among the parties by the source, and if the verifier is honest, the probability of aborting is $\pr{b_{out} = 1} = \frac{\tau^2}{4}$ with:
\begin{equation}
    \tau = \min_{U}\mbox{TD}(\ketbra{GHZ}{GHZ}, U\rho U^\dagger)
\end{equation}
where TD is the trace distance and $U$ is an operator acting only on the space of the dishonest parties. Since the verifier might not always be honest, we can repeat this protocol to increase the abort probability when the state is maliciously prepared. The verification protocol thus consists in rounds where the source shares a state to the parties, who then samples $r \sample \{0,1\}^S$ to decide if they verify the state or if they keep it. This process goes on until until either $r=0\dots0$ in step (3), in which case the parties output a quantum state, or if $b_{out}=1$ in step (8), in which case the parties abort. The fact that the parties decide to verify or keep the state \emph{after} it is shared ensures that possible malicious parties do not adapt their behaviour to a particular round. In \cite{MEVresistant}, the authors prove that the probability that a state $\ket{\psi}$ is accepted after repeating the protocol until $r=0\cdots0$, and $\ket{\psi}$ is further than $\eps$ from the GHZ state is given by:
\begin{equation} 
\pr{\ket{\psi} \textrm{accepted} \land \min _{U} \mbox{TD}(U\ket{\Psi},\ket{GHZ}) \geq \epsilon} =  2^{-S} \frac{4 n}{h \eps^2}
\end{equation}
where $U$ acts on the space of the dishonest parties, $n$ is the total number of parties and $h$ is the number of honest parties.

In later work~\cite{MEVexperimental}, an experimental realization of a loss-tolerant variant of this protocol has been implemented, hinting towards practicability of this protocol in real-life networks. Note that the result that we show in this section also applies to this variant. \\

A previous composable security study of this protocol~\cite{YDK21_ComposableSecurityMultipartite} could not show the composable security of this protocol in the case of a coalition on malicious party and source and conjectured that it was impossible to prove in the AC framework. Our present work answers this conjecture by the negative. The key difference with the previous work lies in \cref{rq:GlobalSimu}, which allows local simulators to communicate between each others via an additional communication channel resource. Equivalently, this allows to consider a single global simulator having access to all the corrupted interfaces of the ideal result, which is actually the choice made by the Universal Composability framework. Building on the work from the previous sections, we can show the following lemma proving composable security of the protocol from \cite{MEVresistant}:

\begin{lemmaE}[][end, restate]
  The protocol defined in~\cite{MEVresistant}, assuming that we use a broadcast channel to transmit the abort bit, is an $\eps$-graph state verification protocol for $\eps \eqdef \frac{1}{2^{S/2}}(4n + 1)$, where $2^S$ is the average number of tests before outputting a state as defined in~\cite{MEVresistant} and $n$ is the total number of parties. As a result, it $(2\sqrt{2\eps - \eps^2})$-realizes $\Verif[\ket{G}]^f$ as defined in \cref{thm:mainThm}, and can be turned into a protocol that $(2\sqrt{2\eps - \eps^2})$-realizes $\Verif$.
\end{lemmaE}
\begin{proofE}
  This is a corollary of \cref{lem:translationFidelityProba}, using the second version of the assumption (\cref{eq:assumptionVersionMin}). The probability in \cref{eq:assumptionVersionMin} corresponds exactly to the $\pr{C_\eps}$ defined in~\cite{MEVresistant} right before Theorem~3. This same theorem actually states that for any $\eps > 0$, $\pr{C_\eps} \leq \frac{4n}{2^{S} h \eps^2}$, where $2^S$ is the average number of tests before outputting a state as defined in~\cite{MEVresistant}, $n$ is the total number of parties, and $h$ is the number of honest parties. 
  
In AC the distance between two resources is the same irrespective of the number of corrupted parties. We thus aim to upper bound this by a number independent from the number of malicious parties. We can assume that at least one party is honest, i.e. $h \geq 1$, to get $\pr{C_\eps} \leq \frac{4n}{2^{S} \eps^2}$.  Since this is true for any $\eps$, we can in particular define:
  \begin{align}
    \eta &\eqdef \sqrt[4]{\frac{1}{2^S}} \label{def:epsstar}\\
    \delta &\eqdef \frac{4n}{2^{S} \eta^2} = \frac{4n}{2^{S/2}}
  \end{align}
  and we have $\pr{C_\eta} \leq \delta$. So we can apply \cref{lem:translationFidelityProba} to show that the protocol is a $\delta + \eta^2$-graph state verification protocol, i.e.\ an $\eps$-graph state verification protocol for $\eps \eqdef \delta + \eta^2 = \frac{4n}{2^{S/2}} + \sqrt{\frac{1}{2^S}} = \frac{1}{2^{S/2}}(4n + 1) $. Finally, we can conclude the proof by using \cref{thm:mainThm} and \cref{cor:mainThmWithoutCorrections} that directly show that this protocol realizes $\Verif^f$ and can be turned, with little changes, into a protocol realizing $\Verif$.
\end{proofE}

\subsection{Graph state verification}

The work of \cite{UM22_VerificationGraphStates} presents a generalization of the previous protocol to arbitrary graph states, where the verifier will ask parties to measure random stabilizers of the graph state. We can similarly show that it is composably secure. Note however that this protocol, like the previous one, shows security scaling polynomially with the security parameter while we usually expect it to scale super-polynomially. However, this allows significantly simpler protocols and most existing graph state verification protocols have this property. But our framework, of course, also applies to protocols that are super-polynomially secure.

\begin{lemmaE}[][end, restate]
  We define, as in the theorem 3 of ~\cite{UM22_VerificationGraphStates} (where we use the fact that conditioned on non-aborting, we have $N_{\mathsf{pass}} \geq \lambda J N_{\mathsf{test}} - \frac{N_{\mathsf{test}}}{2J}$ as described in protocol 2 to simplify the expression of $p_0$ and avoid any dependency on a number that might be different every time we run the protocol):
  \begin{itemize}
  \item $J = 2^n$ or $J = n$ depending on $G$ as described in \cite[Thm.\ 2]{UM22_VerificationGraphStates},
  \item $\lambda$ be the security parameter growing polynomially with the number of tests
  \item $m$ and $c$ some positive constants chosen so that $p_0$ and $\eta_0$ defined later are greater than $0$,
  \item $p_0 \eqdef [1-\sum_{x=0}^\lambda (1-\frac{1}{n})^x(\frac{1}{n} J^{\frac{-2cm}{3}})^{\lambda-x}]^J$ (we got rid of the number of honest parties $|H| \geq 1$ since we want this to be independent of the number of malicious parties)
  \item $\eta_0 \eqdef \left(\frac{1}{\lambda} - \frac{1}{\lambda^2} \right) + \left(1+\frac{1}{\lambda} \right)\frac{\sqrt{c}+1/2}{J}$
  \end{itemize}
  The symmetric protocol 2 defined in~\cite{UM22_VerificationGraphStates}, assuming that we use a broadcast channel to transmit the abort bit, is an $\eps$-graph state verification protocol for $\eps \eqdef 1 - p_0 + 2\eta_0-\eta_0^2$. As a result, it $(2\sqrt{2\eps - \eps^2})$-realizes $\Verif[\ket{G}]^f$ as defined in \cref{thm:mainThm}, and can be turned into a protocol that $(2\sqrt{2\eps - \eps^2})$-realizes $\Verif$.
\end{lemmaE}
\begin{proofE}
  This is mostly a corollary of \cref{lem:translationFidelityProba}, using the first version of the assumption (\cref{eq:assumptionVersionF}). Indeed, using \cite[Thm.~3]{UM22_VerificationGraphStates}, and by denoting by $\rho_i$ the reduced state outputted by honest parties during a given run (this corresponds to $\rho_H^{\mathsf{avg}}$ in \cite[Thm.~3]{UM22_VerificationGraphStates}), we know that for any adversary $\cA$:
  \begin{align}
    \pr[\rho_i \gets \Tr_{[n] \setminus H} \pi_H \cR \cA, \rho_i \neq \ket{\bot^{H}}]{F(\rho_i, \Tr_{[n] \setminus H}\ketbra{G}{G}) \geq 1-\eta_0} \geq p_0
  \end{align}
  But for any $\rho_i$,
  \begin{align}
     F(\rho_i, \Tr_{[n] \setminus H}\ketbra{G}{G}) \geq 1-\eta_0
    \Leftrightarrow
    \sqrt{1-F^2(\rho_i, \Tr_{[n] \setminus H}\ketbra{G}{G})} \leq \sqrt{1-(1-\eta_0)^2} = \sqrt{2\eta_0-\eta_0^2}
  \end{align}
  So let $\eta \eqdef \sqrt{2\eta_0-\eta_0^2}$, and $\delta \eqdef 1 - p_0$. We have therefore:
  \begin{align}
    \pr[\rho_i \gets \Tr_{[n] \setminus H} \pi_H \cR \cA, \rho_i \neq \ket{\bot^{H}}]{\sqrt{1-F^2(\rho_i, \Tr_{[n] \setminus H}\ketbra{G}{G})} \geq \eta_0} \leq \delta
  \end{align}
  and, since $\pr{X \land E} = \pr[E]{x} \pr{E} \leq \pr[E]{x}$, we get as well:
  \begin{align}
    \pr[\rho_i \gets \Tr_{[n] \setminus H} \pi_H \cR \cA]{\rho_i \neq \ket{\bot^{H}} \land \sqrt{1-F^2(\rho_i, \Tr_{[n] \setminus H}\ketbra{G}{G})} \leq \eta_0} \leq \delta
  \end{align}
  So according to \cref{lem:translationFidelityProba}, the protocol is an $(\delta + \eta^2)$-graph state verification protocol, i.e.\ an $\eps$-graph state verification protocol for $\eps \eqdef \delta + \eta^2 = 1 - p_0 + 2\eta_0-\eta_0^2$. Finally, we can conclude the proof by using \cref{thm:mainThm} and \cref{cor:mainThmWithoutCorrections} that directly show that this protocol realizes $\Verif^f$ and can be turned, with little changes, into a protocol realizing $\Verif$.
\end{proofE}

\section{Conclusion}

In this work, we studied the composable security of generic graph-state verification protocols, i.e.\ protocols consisting of a source sharing graph-states to a network of parties, who then only perform local operations and classical communications, to decide whether they can use some of the states for another protocol. We showed that they can be considered equivalent to an ideal abstract resource, up to an $\eps$ that we explicited. This resource shares graph-states while allowing a restricted class of corrections to be applied to them by malicious parties. We proved that without modifying the original protocols, these malicious actions cannot be prevented when considering generic black-box simulators. We thus showed how to modify the original verification protocols to prevent the malicious parties to effectively act on the shared graph-state. The modification consists in asking the parties to jointly apply a random stabilizer to the state that they receive. By doing so, the modified protocol can be considered equivalent to a ideal functionality simply sharing graph-states. In this modified version, the dishonest parties can only force the protocol to abort.

To prove our results, we studied the class of \emph{mergeable states}, to which graph-states belongs. We showed that we can merge two copies of a graph state into one copy, by acting only on a partition of the qubits. This generalizes entanglement swapping to arbitrary graph state. We explicited the measurements and local operations to do this swapping using scalable ZX-calculus~\cite{CHP19_SZXcalculusScalableGraphical}. For graph-state manipulation, the scalable ZX-calculus formalism proved to be more handy than the usual density matrix formalism, and we provide an introduction in \cref{appendix:introZX}. We believe this work is one of the first using such methods in the context of quantum cryptography.

We emphasize that our results mostly preserves the security features of the original protocol. Notably, if the initial protocol is $\eps$-secure, our claim is that it $\eps'$-realizes the above functionalities where $\eps'$ is polynomially related to $\eps$. In particular, if $\eps$ is scaling inverse polynomially with the security parameter, $\eps'$ will also scale inverse polynomially. On the other hand if $\eps$ is negligible, so will $\eps'$. 

Our results are crucial in the context of network protocol development. The class of verification protocols that we studied are use as building-blocks in many quantum network protocols. Our work explicits to which extent they can be repeatedly used by computing of communication protocols to get verified graph-state without threatening the overall security. We showed in the last section of the paper how to apply our result to two existing verification protocols~\cite{MEVresistant,UM22_VerificationGraphStates}. In the process, we answered negatively a conjecture posed in~\cite{YDK21_ComposableSecurityMultipartite} stating that graph-state verification protocols cannot be proven composably secure in the Abstract cryptography framework. Moreover, our composability proof readily applies in a network where trusted parties wish to build trust on a source of graph-state. This is the case in many real-life scenario where a group of people wish to get a graph-state from an untrusted network to perform, for example, a multiparty computing protocol. In future works, we will study different contexts in which our result can apply.

\section{Acknowledgment}

The authors deeply thank Robert Booth for suggesting a more elegant representation of graph states. This work is co-funded by the European Union (ERC, ASC-Q, 101040624) and the Horizon Europe program (QUCATS). It is also supported by the Dutch National Growth Fund (NGF), as part of the Quantum Delta NL programme, and the Government of Spain (Severo Ochoa CEX2019-000910-S and FUNQIP), Fundació Cellex, Fundació Mir-Puig, Generalitat de Catalunya (CERCA program). Views and opinions expressed are however those of the author(s) only and do not necessarily reflect those of the European Union or the European Research Council. This work is supported by France 2030 under the French National Research Agency projects HQI ANR-22-PNCQ-0002 and EPiQ ANR-22-PETQ-0007 as well as ANR Project SecNISQ. Neither the European Union nor the granting authority can be held responsible for them.

\printbibliography[heading=bibintoc]

\newpage

\appendix
\section{Introduction to the (scalable) ZX-calculus}\label{appendix:introZX}

\subsection{Introduction to the (scalable) ZX-calculus.}

\paragraph{ZX-diagrams.}

A ZX-diagram is a graph with some inputs (on the left) and outputs (on the right). We can associate to any ZX-diagram a matrix, using the following interpretation of the generators, where the colored nodes are being called ``spiders'', and $H$ being strictly speaking only a syntactic sugar:
\begin{align*}
  \begin{split}
    % Green spider
    \left\llbracket\, \zx{\leftManyDots{n} \zxZ{\alpha} \rightManyDots{m}} \,\right\rrbracket &= \ket{0}^{\otimes m}\bra{0}^{\otimes n}+e^{i \alpha}\ket{1}^{\otimes m}\bra{1}^{\otimes n}\\
    % Red spider
    \left\llbracket\, \zx{\leftManyDots{n} \zxX{\alpha} \rightManyDots{m}} \,\right\rrbracket &= \ket{+}^{\otimes m}\bra{+}^{\otimes n}+e^{i \alpha}\ket{-}^{\otimes m}\bra{-}^{\otimes n}\\
    % Swap
    \left\llbracket\,
        \begin{ZX}
          \zxN{} \ar[rd,s] &[\zxWCol] \zxN{}\\[\zxWCol]
          \zxN{} \ar[ru,s] &[\zxWCol] \zxN{}
        \end{ZX}
    \,\right\rrbracket &= \ketbra{00}{00}+\ketbra{10}{01}+\ketbra{01}{10}+\ketbra{11}{11}\\
    % H
    \left\llbracket\,\zx{\zxNone{} \rar &[\zxwCol] \zxH{} \rar &[\zxwCol] \zxNone{}}\,\right\rrbracket &= \ketbra{+}{0}+\ketbra{-}{1}
  \end{split}%\qquad
  \begin{split}
    % Wire
    \left\llbracket\, \zx{\zxNone{} \rar &[2mm] \zxNone{}} \,\right\rrbracket  &= \ketbra{0}{0}+\ketbra{1}{1}\\
    % Bell preparation
    \left\llbracket\,\zx{\zxN{} \dar[C]\\[\zxWCol] \zxN{}}\,\right\rrbracket &= \ket{00}+\ket{11}\\
    % Bell projection
    \left\llbracket\,\zx{\zxN{} \dar[C-]\\[\zxWCol] \zxN{}}\,\right\rrbracket &= \bra{00}+\bra{11}\\
    % Empty
    \llbracket\, \zx[inner sep=1mm]{\zxEmptyDiagram{}} \,\rrbracket &= \begin{pmatrix} 1 \end{pmatrix}
  \end{split}
\end{align*}
Note that when $\alpha = 0$, we can omit the angle in the spiders. Moreover, we can compose these generators sequentially (resp.\ in parallel). We then obtain the resulting natural interpretation, by inductively computing the matrix product (resp.\ tensor product) of the interpretation of the sub-diagrams.

\paragraph{Circuit to ZX.}

One can turn any quantum circuit into a ZX-diagram. Let $(a,b) \in \{0,1\}^2$ and $\alpha \in \R$. Then, up to a non-relevant re-normalisation scalar and global phase, we represent basic states using $\zx{\zxX*{a\pi} \rar &[\zxwCol]} = \ket{a}$, $\zx{\zxZ*{a\pi} \rar &[\zxwCol]} = H\ket{a}$ and one-qubits gates using $\zx{\zxNL \zxH{} \zxNR} = \gateH$, $\zx{\zxN{} \rar &[\zxwCol] \zxX{\alpha} \rar &[\zxwCol] \zxN{}} = \gateRX(\alpha)$ and $\zx{\zxN{} \rar &[\zxwCol] \zxZ{\alpha} \rar &[\zxwCol] \zxN{}} = \gateRZ(\alpha)$. We have in particular $\zx{\zxN{} \rar &[\zxwCol] \zxX*{a\pi} \rar &[\zxwCol] \zxN{}} = \gateX^a$ and $\zx{\zxN{} \rar &[\zxwCol] \zxZ*{a\pi} \rar &[\zxwCol] \zxN{}} = \gateZ^a$. Two qubit gates are represented as\footnote{Although the generators do not allow vertical wires, we will see that we can freely bend wires and move nodes since ``only topology matters''. We can therefore consider ZX-diagrams as undirected graphs.} %
$\zx{
  \zxN{} \rar &[\zxwCol] \zxZ{} \rar \dar &[\zxwCol] \zxN{}\\
  \zxN{} \rar &[\zxwCol] \zxX{} \rar &[\zxwCol] \zxN{}
} = \gateCNOT$ and %
$\zx{
  \zxN{} \rar &[\zxwCol] \zxZ{} \rar \dar[H] &[\zxwCol] \zxN{}\\[\zxHRowFlat]
  \zxN{} \rar &[\zxwCol] \zxZ{} \rar &[\zxwCol] \zxN{}
} = \gateCZ$. Moreover, we represent a measurement in the computational basis whose outcome is $a$ using $\zx{ \rar &[\zxwCol] \zxX*{a\pi}} = \bra{a}$ and a measurement in the Hadamard basis, i.e. the projection on $\ket{0} + (-1)^a\ket{1}$, whose outcome is $a$ $\zx{ \rar &[\zxwCol] \zxZ*{a\pi}} = \bra{a}H$ . Finally, $
\begin{ZX}
   \rar & [\zxwCol] \zxZ*{a\pi} \dar[C-]\\
   \rar & [\zxwCol] \zxX*{b\pi}\\
\end{ZX}
$ represents a Bell measurement (projection on $\ket{0b}+(-1)^a\ket{1\bar{b}}$) with outcomes $(a,b)$.

\paragraph{Scalable ZX.}
The scalable ZX-calculus~\cite{CHP19_SZXcalculusScalableGraphical} generalises these generators: wires can be grouped (or ungrouped) together. We represent these grouped wires as bold wires, where we specify  above the number of wires in the group when there can be a confusion. The grouping/ungrouping of wires is done using so-called gatherers and dividers, shown in the first two diagrams of \cref{eq:ScalableZxDef} below. In scalable ZX, spiders can contain lists of real numbers, corresponding to stacked spiders, as we show for example in the two last diagram of \cref{eq:ScalableZxDef} with $\alpha \in \R$ and $\beta \in \R^{n-1}$:
\begin{align}
\label{eq:ScalableZxDef}
  \begin{aligned}
    % Defines divider
    \begin{ZX}
      \rar[Bn'Args={n}{pos=.9}] & [\zxWCol] \zxN{} \ar[dr, '>, B] & [\zxWCol]                        & [\zxWCol] \\[\zxZeroRow+.5mm]
      &                                 & \zxDivider-{} \rar[wc,Bn'={n+m}] &           \\[\zxZeroRow+.5mm]
      \rar[Bn.Args={m}{pos=.9}] & \zxN{} \ar[ur, .>, B]
    \end{ZX}\qquad
    % Defines gatherers
    \begin{ZX}
      & [\zxWCol]                                  & [\zxwCol] \zxN{} \rar[Bn'Args={n}{pos=.1}] & [\zxWCol] \\[\zxZeroRow+.5mm]
      \rar[Bn'={n+m},wce] & \zxDivider{} \ar[B, dr, <.] \ar[ur, B, <'] &                                                        \\[\zxZeroRow+.5mm]
      &                                            & \zxN{} \rar[Bn.Args={m}{pos=.1}]           & 
    \end{ZX}&\qquad\qquad
    % Spiders as lists
    \begin{ZX}
      \zxN{} \rar[Bn'Args={n}{pos=.3},dr, '>, B] \ar[dd,3 vdots={$k$}] & [\zxWCol]                                                                                  & [\zxWCol] \zxN{} \ar[dd,3 vdotsr={\footnotesize$l$}] \\[\zxZeroRow+.5mm]
      & \zxZ[B]{\alpha{::}\beta} \ar[ru, <', Bn'Args={n}{pos=.7}] \ar[rd, <', Bn.Args={n}{pos=.7}] &                                                      \\[\zxZeroRow+.5mm]
      \zxN{} \rar[Bn.Args={n}{pos=.3},ru, .>, B]                       &                                                                                            & \zxN{}
    \end{ZX}
    \ \hypertarget{eq:zxRuleZ}{\eqAbove[\eqdef]{\myZXRule{Z}}}\ %
    \begin{ZX}
      \zxN{} \rar[B,wc,Bn'Args={n}{pos=.4}] \ar[d,3 vdots={$k$}] & [\zxwCol] \zxDivider{} \rar[] \ar[dr,B] & [\zxwCol] \zxZ{\alpha} \rar \ar[rd]  & [\zxwCol] \zxDivider-{} \rar[wc, Bn'Args={n}{pos=.6}] & [\zxwCol] \zxN{} \ar[d,3 vdotsr={$l$}] \\
      \zxN{} \rar[B,wc,Bn.Args={n}{pos=.4}]                      & \zxDivider{} \rar[B] \ar[ur]            & \zxZ[B]{\beta} \ar[ru, B, wc] \ar[r, B] & \zxDivider-{} \rar[wc, Bn.Args={n}{pos=.6}]           & \zxN{}
    \end{ZX}&\qquad\qquad
    % Hadamard
    \begin{ZX}[mbr]
      \zxN{} \rar[Bn'Args={n+m}{}] &[\zxWCol+2mm] \zxH[B]{} \rar[Bn.Args={n+m}{}] &[\zxWCol+2mm] \zxN{}
    \end{ZX}
    \hypertarget{eq:zxRuleW}{\eqAbove[\eqdef]{\myZXRule{W}}}
    \begin{ZX}[mbr=1]
      \rar[B] & \zxDivider{} \ar[r,o',H={B},Bn'Args={n}{pos=.25}] \ar[r,o.,H={B},Bn.Args={m}{pos=.25}] &[\zxHCol+5mm] \zxDivider-{} \rar[B] &           
    \end{ZX}
    % \begin{ZX}[mbr=2]
    %           & [\zxwCol]                                    & [\zxwCol] \zxH{} \ar[rd,'>]    & [\zxwCol]\zxN{}       & [\zxwCol] \\[\zxZeroRow]
    %   \rar[B] & \zxDivider{} \ar[ru,<'] \ar[rd,<.,Bn.={n-1}] &                                & \zxDivider-{} \rar[B] &           \\[\zxZeroRow]
    %           &                                              & \zxH[B]{} \ar[ru,Bn.={n-1},.>] & [\zxWCol] \zxN{}
    % \end{ZX}
  \end{aligned}
\end{align}
The interest of the ZX-calculus is that we can rewrite a ZX diagram using some rewriting rules while preserving the interpretation of the corresponding matrix. The first rule is that \emph{only connectivity matters}, i.e.\ we can bend wires arbitrarily as soon as the corresponding undirected graph is left unchanged. Other rules and theorems are described in \cref{fig:rulesSZX}. Note that, for simplicity, we will remove all scalars, i.e.\ sub-graphs having no input nor output, as they correspond to global phases (not observable physically) and/or a re-normalisation of the state.

\begin{lemma}[{Rewiring rule \protect\myZXRule{R}~\cite[Thm.~3.2]{CHP19_SZXcalculusScalableGraphical}}]\label{eq:zxRuleR}
  Two diagrams composed only of identity, gatherers and dividers are equal iff their respective number of inputs and outputs are equal.
\end{lemma}

\begin{figure}[!ht]
  \centering
  \begin{alignat*}{6}
    \begin{ZX}
      \zxN[a=a]{}  &                       & \zxN[a=c]{}  \\[\zxZeroRow]
                   & \zxZ[a=alpha,B]{\alpha} &              \\[\zxZeroRow]
      \zxN[a=b]{}  &                       & \zxN[a=myd]{}  \\[\zxwRow]
      \zxN[a=aa]{} &                       & \zxN[a=cc]{} \\[\zxZeroRow]
                   & \zxZ[a=beta,B]{\beta}   &              \\[\zxZeroRow]
      \zxN[a=bb]{} &                       & \zxN[a=mydd]{}
      % First lines
      \ar[from=a, to=alpha,'>,B]       
      \ar[from=b, to=alpha,.>,B]       
      \ar[to=c, from=alpha,<',B]       
      \ar[from=alpha,to=myd,<.,B]       
      \ar[from=a,to=b,3 vdots]       
      \ar[from=c,to=myd,3 vdots]       
      % second lines
      \ar[from=aa, to=beta,'>,B]       
      \ar[from=bb, to=beta,.>,B]       
      \ar[to=cc, from=beta,<',B]       
      \ar[to=mydd, from=beta,<.,B]
      \ar[from=aa,to=bb,3 vdots]       
      \ar[from=cc,to=mydd,3 vdots]       
      % Between
      \ar[from=alpha, to=beta, o-={-=.2,L=.3},B]
      \ar[from=alpha, to=beta, -o={-=.2,L=.3},B]
      \ar[from=alpha, to=beta, 3 dots]
    \end{ZX}
    & \hypertarget{eq:zxRuleS}{\eqZX{S}}
    \begin{ZX}
      \zxN[a=a]{} &                                 & \zxN[a=c]{} \\[\zxZeroRow]
                  & \zxZ[a=alpha,B]{\alpha + \beta} &             \\[\zxZeroRow]
      \zxN[a=b]{} &                                 & \zxN[a=myd]{} 
      % First lines
      \ar[from=a, to=alpha,'>,B]       
      \ar[from=b, to=alpha,.>,B]       
      \ar[to=c, from=alpha,<',B]       
      \ar[from=alpha,to=myd,<.,B]       
      \ar[from=a,to=b,3 vdots]       
      \ar[from=c,to=myd,3 vdots]       
    \end{ZX} & \qquad
    %%%%%%%%% Euler
    \begin{ZX}[ampersand replacement=\&]
      \zxN{} \rar \&[\zxwCol] \zxH{} \rar \&[\zxwCol] \zxN{}
    \end{ZX} &\hypertarget{eq:zxRuleEU}{\eqZX{EU}}
    \begin{ZX}
      \zxN[a=myb]{} \rar & \zxFracZ{\pi}{2} \rar      & \zxFracX{\pi}{2} \rar & \zxFracZ{\pi}{2} \rar & \zxN{}
    \end{ZX}&\qquad
    %%%%%%% Hadamard rule
    \begin{ZX}
      \zxN{} \ar[rd,B,-N.,H={B,pos=.35}]
      \ar[dd,3 vdots]                    & [\zxWCol] & [\zxWCol] \zxN{} \ar[dd,3 vdots] \\[\zxNRow]%%
                                         & \zxX[B]{\alpha} \ar[B,ru,N'-,H={B,pos=.65}]
      \ar[rd,N.-,B,H={B,pos=.65}]        &                                              \\[\zxNRow]
      \zxN{} \ar[B,ru,-N',H={B,pos=.35}] & [\zxWCol] & [\zxWCol] \zxN{}
    \end{ZX} & \hypertarget{eq:zxRuleH}{\eqZX{H}}
    \begin{ZX}
      \zxN{} \ar[B,rd,-N.] \ar[dd,3 vdots] & [\zxWCol]                              & [\zxWCol] \zxN{} \ar[dd,3 vdots] \\[\zxNRow]
                                           & \zxZ[B]{\alpha} \ar[B,ru,N'-] \ar[B,rd,N.-] &                                  \\[\zxNRow]
      \zxN{} \ar[B,ru,-N']                 & [\zxWCol]                              & [\zxWCol] \zxN{}
    \end{ZX}&\qquad
    %%%%%%%%%%% Identity rules
    \begin{ZX}
      \zxN{} \rar[B] &[\zxwCol] \zxZ[B]{} \rar[B] &[\zxwCol] \zxN{}
    \end{ZX}
    & \hypertarget{eq:zxRuleI}{\eqZX{I}}
    \begin{ZX}
      \zxN{} \rar[B] &[\zxWCol] \zxN{}
    \end{ZX}\\
    %%%%%%%%%%% Copy
    \begin{ZX}
                             &                         & [\zxwCol] \zxN[a=a]{} \ar[dd,3 vdots] \\[\zxZeroRow]
      \zxX[B]*{a\pi} \rar[B] & \zxZ[B,a=theta]{\theta} &                                       \\[\zxZeroRow]
                             &                         & \zxN[a=b]{}
      \ar[from=theta, to=a,<',B]
      \ar[from=theta, to=b,<',B]      
    \end{ZX} & \hypertarget{eq:zxRuleC}{\eqZX{C}}
    \begin{ZX}
      \zxX[a=a,B]*{a\pi} \rar[B] & \zxN[]{} \ar[d,3 vdots] \rar[B] & \zxN{} \\
      \zxX[a=b,B]*{a\pi} \rar[B] & \zxN[]{}  \rar[B]               & \zxN{}
    \end{ZX} & \qquad
    %%%%%%%%%%%% Pi-commutation
    \begin{ZX}
                     & [\zxwCol]              &                         & \zxN[a=a]{} \ar[dd,3 vdots] \\ 
      \zxN{} \rar[B] & \zxX[B]*{a\pi} \rar[B] & \zxZ[B,a=theta]{\theta} &                             \\
                     &                        &                         & \zxN[a=b]{}
      \ar[from=theta, to=a,<',B]
      \ar[from=theta, to=b,<',B]
    \end{ZX} & \hypertarget{eq:zxRulePi}{\eqZX{Pi}}
    \begin{ZX}
                     & [\zxwCol]                       & \zxX[a=a,B]*{a\pi} \rar[B] &[\zxwCol] \zxN[]{} \ar[dd,3 vdots] \\[\zxZeroRow]
      \zxN{} \rar[B] & \zxZ[B,a=theta]{(-1)^{a}\theta} &                                                       \\[\zxZeroRow]
                     &                                 & \zxX[a=b,B]*{a\pi} \rar[B] & \zxN[]{}
      \ar[from=theta, to=a,<',B]
      \ar[from=theta, to=b,<',B]
    \end{ZX}&\qquad
    %%%%%%%%%%% Bialgebra
    \begin{ZX}
      \zxN{} \rar & [\zxwCol] \zxZ{} \rar[o'] \ar[rd,] & \zxX{} \rar & [\zxwCol] \zxN{} \\
      \zxN{} \rar & \zxZ[a=myBase]{} \rar[o.] \ar[ru,] & \zxX{} \rar & \zxN{}
    \end{ZX} & \hypertarget{eq:zxRuleB}{\eqZX{B}}
    \begin{ZX}
      \zxN{} \ar[dr,-N.] & [\zxwCol]   &                                & [\zxwCol] \zxN{} \\
                         & \zxX{} \rar & \zxZ{} \ar[ru,N'-] \ar[rd,N.-] &                  \\
      \zxN{} \ar[ru,-N'] &             &                                & \zxN{}
    \end{ZX}             & \qquad
    %%%%%%%%%%% Hopf rule (not needed, but really practical usually and it does not cost much to add it)
    \begin{ZX}
      \zxN{} \rar & \zxZ{} \ar[r,o'] \ar[r,o.] & \zxX{} \rar & \zxN{} \\
    \end{ZX} & \hypertarget{eq:zxRuleR}{\eqZX{Ho}}
    \begin{ZX}
      \zxN{} \rar &[\zxwCol] \zxZ{} & \zxX{} \rar &[\zxwCol] \zxN{}
    \end{ZX}    
  \end{alignat*}
  \caption{Rules and theorems of the SZX calculus (the rules also hold when read from right to left and after exchanging colors, note that we omit the more generic Euler rule as it is not needed for the Clifford fragment that we use here). $a \in \{0,1\}^n$ and $b \in \{0,1\}^n$ are arbitrary vectors of binary variables, $\theta \in \R^n$ is an arbitrary vector, and operations between vectors are component-by-component. In the S rule, at least one connection must exist between both spiders, and bold wires indicate wires of size $\geq 1$.}
  \label{fig:rulesSZX}
\end{figure}
%% Note: in the compact axiomatisation, the rules Z1 and Z2 are equivalent to Z, itself equivalent to our definition (just apply Z and P rule).

The scalable ZX-calculus also provides a way to use arbitrary binary matrices to specify the connectivity between nodes: for any matrix $A \in \Z_2^{m \times n}$, we have:
\begin{align}
  %% Definition of matrices
  \left\llbracket\, \begin{ZX}
    \zxN{} \rar[B] & [\zxwCol] \zxMatrix{A} \rar[B] &[\zxwCol] \zxN{}
  \end{ZX} \,\right\rrbracket \eqSubeq[\coloneqq]{eq:defMatrix}{a} (\ket{x} \mapsto \ket{Ax})\qquad
  \begin{ZX}[mbr]
    \zxN{} \rar[B] & [\zxwCol] \zxMatrix{A} \rar[B] &[\zxwCol] \zxN{}
  \end{ZX} \eqSubeq{eq:defMatrix}{b}
  \begin{ZX}[mbr=2]
                      & [\zxwCol]                          & \zxZ{} \ar[dr] \ar[dd,3 vdots] &                                      & \zxX{} \ar[rd,'>] \ar[dd,3 vdots] &                          & [\zxwCol] \\
    \zxN{} \rar[B,wc] & \zxDivider{} \ar[ru,<'] \ar[rd,<'] &                                & |[inner sep=.7mm]| A \ar[ru] \ar[rd] &                                   & \zxDivider-{} \rar[B,wc] & \zxN{}    \\
                      &                                    & [\zxwCol] \zxZ{} \ar[ru]       &                                      & \zxX{} \ar[ru,.>]                 &                                      \\
  \end{ZX}\qquad
  %% Definition of transposition
  \begin{ZX}[mbr]
    \zxN{} \rar[B] & \zxMatrix*{A} \rar[B] & \zxN{}
  \end{ZX} \eqSubeq[\coloneqq]{eq:defMatrix}{c}
  \begin{ZX}[mbr=2]
                & [\zxWCol] \zxN{} \ar[B,dr,s]                 & [\zxWCol] \zxN{} \\[\zxZeroCol+.3mm]
    \ar[B,dr,s] & \zxMatrix{A} \ar[B,u,C] \ar[B,d,C-] &                  \\
                &                                              & \zxN{}           \\
  \end{ZX}\label{eq:defMatrix}
\end{align}
Where, in the second equation, $A$ represents the biadjency matrix of the bipartite green/red graph.

For any binary vector $u$ and binary matrix $A$ and $B$ of appropriate size, we also have the following rewriting rules, that informally come from the fact that a green spider copies states in the computational basis while a red spider performs the modulo sum of its inputs:
\begin{alignat}{6}
  \begin{ZX}[mbr=3]
                   & [\zxwCol]                     &                                        & [\zxWCol] \\[5mm] %% The matrix has overlay, we avoid it to overlay the content of the text
                   &                               &                                        & \zxN{}    \\
    \zxN{} \rar[B] & \zxMatrix/_{A\\B}{} \rar[B,wc] & \zxDivider{} \ar[ru,<',B] \ar[rd,<.,B] & \zxN{}\\
                   &                               &                                        & \zxN{} 
  \end{ZX} &\eqSubeq{eq:matrixManipulationA}{a}
  \begin{ZX}[mbr=2]
                   & [\zxWCol]                           & \zxMatrix{A} \rar[B]  & [\zxwCol] \zxN{} \\[\zxZeroRow]
    \zxN{} \rar[B] & \zxZ[B]{} \ar[ru,<',B] \ar[rd,<.,B] & \zxN{}                                   \\[\zxZeroRow]
                   &                                     & \zxMatrix-{B} \rar[B] & \zxN{} 
  \end{ZX}&\quad
  % Block matrix
  \begin{ZX}[mbr=1]
    \zxN{} \rar[B] & \zxMatrix_{A & B}{} \rar[B] & \zxN{}
  \end{ZX} &\eqSubeq{eq:matrixManipulationA}{b}
  \begin{ZX}[mbr=2]
                      & [\zxwCol]                              & \zxMatrix{A} \ar[rd,B,'>]   &                &[\zxwCol] \\[\zxZeroRow]
    \zxN{} \rar[B,wc] & \zxDivider{} \ar[ru,<',B] \ar[rd,<.,B] & \zxN{}                      & \zxX[B]{} \rar[B] & \zxN{} \\[\zxZeroRow]
                      &                                        & \zxMatrix-{B} \rar[ru,B,.>] & 
  \end{ZX}&\quad
  % A + B
  \begin{ZX}[mbr=1]
    \zxN{} \rar[B] & \zxMatrix{A + B} \rar[B] & \zxN{}
  \end{ZX} &\eqSubeq{eq:matrixManipulationA}{c}
  \begin{ZX}[mbr=2]
                      & [\zxwCol]                           & \zxMatrix{A} \ar[rd,B,'>]   &                   & [\zxwCol] \\[\zxZeroRow]
    \zxN{} \rar[B,wc] & \zxZ[B]{} \ar[ru,<',B] \ar[rd,<.,B] & \zxN{}                      & \zxX[B]{} \rar[B] & \zxN{}    \\[\zxZeroRow]
                      &                                     & \zxMatrix-{B} \rar[ru,B,.>] & 
  \end{ZX}\label{eq:matrixManipulationA}\\
  % A * B
  \begin{ZX}[mbr=1]
    \zxN{} \rar[B] & \zxMatrix{BA} \rar[B] & \zxN{}
  \end{ZX} &\eqSubeq{eq:matrixManipulationB}{a}
  \begin{ZX}[mbr=1]
    \zxN{} \rar[B] & [\zxwCol] \zxMatrix{A} \rar[B] & \zxMatrix{B} \rar[B] & [\zxwCol] \zxN{}
  \end{ZX}            & \quad
  \begin{ZX}[mbr=1]
    \zxN{} \rar[B,wc] &[\zxwCol] \zxMatrix{A} \rar[B] & \zxZ[B]*{u\pi} \rar[B] &[\zxwCol] \zxN{}
  \end{ZX} &\eqSubeq{eq:matrixManipulationB}{b}
  \begin{ZX}[mbr=1]
    \zxN{} \rar[B,wc] &[\zxwCol] \zxZ[B]{A^Tu\pi} \rar[B] & \zxMatrix{A} \rar[B] &[\zxwCol]  \zxN{}
  \end{ZX}&\quad
  \begin{ZX}[mbr=1]
    \zxN{} \rar[B,wc] &[\zxwCol] \zxMatrix{A} \rar[B] & \zxZ[BBw]*{u\pi}
  \end{ZX} &\eqSubeq{eq:matrixManipulationB}{c} \begin{ZX}[mbr=1]
    \zxN[B]{} \rar[B,wc] &[\zxwCol] \zxZ[BBw]{A^Tu\pi}
  \end{ZX}\label{eq:matrixManipulationB}\\
  % A * 0 = 0
    \begin{ZX}[mbr=1]
    \zxN{} \rar[B,wc] &[\zxwCol] \zxX[B]*{u\pi} \rar[B] & \zxMatrix{A} \rar[B] &[\zxwCol]  \zxN{}
  \end{ZX} &\eqSubeq{eq:matrixManipulationC}{a}
  \begin{ZX}[mbr=1]
    \zxN{} \rar[B,wc] &[\zxwCol] \zxMatrix{A} \rar[B] & \zxX[B]{Au\pi} \rar[B] &[\zxwCol] \zxN{}
  \end{ZX} &
  \begin{ZX}[mbr=1]
    \zxX[BBw]*{u\pi} \rar[B,wc] & \zxMatrix{A} \rar[B] &[\zxwCol] \zxN{}
  \end{ZX} &\eqSubeq{eq:matrixManipulationC}{b}
  \begin{ZX}[mbr=1]
    \zxX[BBw]{Au\pi} \rar[wc] &[\zxwCol] \zxN{}
  \end{ZX}\label{eq:matrixManipulationC}
\end{alignat}

\begin{lemma}[{\cite[Thm.~4.8,4.9]{CHP19_SZXcalculusScalableGraphical}}]\label{lem:injectiveSurjective}
  If $A$ is injective, we have
  $\begin{ZX}[mbr=1,arrows={B}]
      \zxN{} \rar[B] & \zxMatrix{A} \rar & \zxMatrix*{A} \rar[B] & \zxN{}
    \end{ZX} = 
    \begin{ZX}[mbr=1,arrows={B}]
      \zxNL \zxNR
    \end{ZX}$
  and similarly, if $A$ is surjective, we have
  $\begin{ZX}[mbr=1,arrows={B}]
      \zxN{} \rar[B] & \zxMatrix*{A} \rar & \zxMatrix{A} \rar[B] & \zxN{}
    \end{ZX} =
    \begin{ZX}[mbr=1,arrows={B}]
      \zxNL \zxNR
    \end{ZX}$. In particular, if $A$ is bijective, we have:
    \begin{align}
      \begin{ZX}[mbr=1,myBold]
        \zxNL \zxMatrix{A} \zxNR
      \end{ZX} \eqeqLem{lem:injectiveSurjective} \begin{ZX}[mbr=1,myBold]
        \zxNL \zxMatrix{A} \rar & \zxMatrix{A^{-1}} \rar & \zxMatrix*{A^{-1}} \zxNR
      \end{ZX} \eqSubeq{eq:matrixManipulationB}{a} \begin{ZX}[mbr=1,myBold]
        \zxNL \zxMatrix{A^{-1}A} \rar & \zxMatrix*{A^{-1}} \zxNR
      \end{ZX} = \begin{ZX}[mbr=1,myBold]
        \zxNL \zxMatrix{I} \rar & \zxMatrix*{A^{-1}} \zxNR
      \end{ZX} \eqSubeq{eq:defMatrix}{b} \begin{ZX}[mbr=1,myBold]
        \zxNL \zxMatrix*{A^{-1}} \zxNR
      \end{ZX}\label{eq:bijectionA}
    \end{align}
\end{lemma}

\subsection{Graph states} We list below some basic properties to graphically represent and manipulate graph states. Most of these facts and proofs can be found in \cite{Car20_NoteDiagonalGates,CHP19_SZXcalculusScalableGraphical}.

\begin{lemma}\label{lem:shapeCZ}
  Let $\Gamma \in \Z_2^{m \times n}$, then the following ZX diagram applies a $\gateCZ$ between the $i$-th qubit of the first group of qubits and the $j$-th qubit of the second group of qubits iff $\Gamma_{j,i} = 1$, or, more formally:
  \begin{align}
    \begin{ZX}
      \zxN{} \rar[Bn']   & [\zxWCol] \zxZ[B]{} \rar[Bn'] \dar[B] & [\zxWCol] \zxN{} \\
                         & \zxMatrix./{\Gamma} \dar[B,H={B}]     &                  \\[\zxHRow+1mm]
      \zxN{} \rar[Bn'=m] & \zxZ[B]{} \rar[Bn'=m]                 & \zxN{}                    
    \end{ZX}
    = \prod_{(i,j), \Gamma_{j,i} = 1} \gateCZ_{i, m+j}
  \end{align}
  where $\gateCZ_{i, j}$ applies a $\gateCZ$ between the $i$-th and $j$-th qubit, assuming $i$, $j$ and the qubits are all indexed starting from $0$.
\end{lemma}

\begin{definition}\label{def:notationGraphs}
  Let $G = (V, E)$ be an undirected graph on $n$ ordered vertices in $V$. By slightly abusing notations, we will also denote its representation in term of an adjacency matrix as $G \in \Z_2^{n \times n}$, the rows and columns of $G$ being indexed by elements in $V$, where $G_{i,j} = 1$ iff there exists an edge between $i$ and $j$ and $G_{i,j} = 0$ otherwise. We also denote $G^{\matrixUr}$ and $G^\matrixLl$ as, respectively, the upper and lower triangular matrix of $G$ (note that $G$ has zeros on its diagonal). In particular, $G = G^\matrixUr + G^\matrixLl$, and since $G^T = G$, we also have $(G^\matrixUr)^T = G^\matrixLl$. For any partition $(A,B)$ of $V$ of $G$ (for simplicity, we assume that elements in $V$ are ordered so that elements in $A$ appear before elements in $B$), we define $G_A \in \Z_2^{|A| \times |A|}$, $G_B \in \Z_2^{|B| \times |B|}$, and $\Gamma_{A \rightarrow  B} \in \Z_2^{|B| \times |A|}$ (or simply $\Gamma$) such that:
  \begin{align}
    G =
    \begin{bmatrix}
      G_A & \Gamma^T\\
      \Gamma & G_B
    \end{bmatrix}\label{eq:decompositionG}
  \end{align}
  In particular, $G_A$ and $G_B$ are the subgraphs of $G$ restricted to vertices in $A$ and $B$ respectively, and $\Gamma$ is the biadjency matrix between elements in $G_A$ and elements in $G_B$.
\end{definition}

\begin{definition}\label{def:graphDiag}
  For any undirected graph $G = (V,E)$ on $n$ ordered vertices, we define $\zx[mbr]{\zxN{} \rar[Bn'] &[\zxWCol] \zxBox{G} \rar[Bn'] &[\zxWCol] \zxN{}}$ as the operation that applies a $\gateCZ$ gate between all input qubits connected in $G$ (the wires being ordered following the order on $V$). We can represent this operation diagrammatically (see \cite{Car20_NoteDiagonalGates}) as:
  \begin{align}
    \begin{aligned}
      \begin{ZX}[myBold]
        \zxN{} \rar[B] &[\zxwCol] \zxBox{G} \rar[B] &[\zxwCol]
      \end{ZX} \eqdef
      \begin{ZX}[myBold]
                       &                                       &                                                   &                                  & \\
                       & \zxH{} \ar[to=Gll,C'] \ar[to=root,<.] &                                                   & \zxMatrix[a=Gll]'/*{G^\matrixLl} & \\
        \zxN{} \ar[rr] &                                       & [\zxwCol] \zxZ[a=root]{} \ar[B,rr] \ar[to=Gll,.>] &                                  & 
      \end{ZX}\label{eq:defGraphicalG}
    \end{aligned}
  \end{align}
  Moreover, we define the graph states $G$ as $\zx{\zxBox{G} \rar[B] &[\zxwCol] \zxN{}} \eqdef \zx{\zxZ[B]{} \rar[B] &[\zxwCol] \zxBox{G} \rar[B] &[\zxwCol] \zxN{}}$.
\end{definition}

\begin{lemma}
  For any graph $G$, $\left\llbracket\,\zx{\zxBox{G} \rar[B] &[\zxwCol]} \,\right\rrbracket = \ket{G}$.
\end{lemma}
The proof of this statement can be found in \cite{Car20_NoteDiagonalGates}, but the main idea is to show that this state is stabilized by the stabilizers of $\ket{G}$, which is formalized by the following lemma:

\begin{lemma}[{\cite{Car20_NoteDiagonalGates}}]\label{lem:stabilizerDiag}
  Let $G = (V,E)$ be a graph, and $x \in \Z_2^{|V|}$ be a vector (indexed by vertices in $V$). Then, $\zx{\zxNL \zxX[BBw]*{x\pi} \rar & \zxZ[BBw]{Gx\pi} \zxNR}$ is a stabilizer of $G$:
  \begin{align}
    \begin{ZX}
      \zxBox{G} \rar & \zxX[BBw]*{x\pi} \rar & \zxZ[BBw]{Gx\pi} \zxNR
    \end{ZX} =
    \begin{ZX}[arrows={B}]
      \zxBox{G} \zxNR
    \end{ZX}\label{eq:stabilizerDiag}
  \end{align}
\end{lemma}
\begin{proof}
  The proof of this statement can be found in \cite{Car20_NoteDiagonalGates}, but we rewrite it here for completeness:
  \begin{align}
    &\begin{ZX}
      \zxBox{G} \rar & \zxX[BBw]*{x\pi} \rar & \zxZ[BBw]{Gx\pi} \zxNR
    \end{ZX} \eqeq{eq:defGraphicalG}
    \begin{ZX}[myBold]
                                               &                                                   &                                  & \\
         \zxH{} \ar[to=Gll,C'] \ar[to=root,<.] &                                                   & \zxMatrix[a=Gll]'/*{G^\matrixLl} & \\
                                               & [\zxwCol] \zxZ[a=root]{} \ar[B,rr] \ar[to=Gll,.>] &                                  & \zxX[BBw]*{x\pi} \rar & \zxZ[BBw]{Gx\pi} \zxNR
    \end{ZX} \eqZX{Pi}
    \begin{ZX}[myBold]
                                          &                                               &                                       & \\
         \zxH{} \dar \ar[to=Gll,C']       &                                               & \zxMatrix[a=Gll]'/*{G^\matrixLl} \dar & \\
         \zxX[BBw]*{x\pi} \ar[to=root,<.] &                                               & \zxX[BBw]*{x\pi}                      & \\
                                          & [\zxwCol] \zxZ[a=root]{} \ar[B,rr] \ar[ru,.>] &                                       & \zxZ[BBw]{Gx\pi} \zxNR
    \end{ZX} \eqZXandSubeq{H}{eq:matrixManipulationC}{a}
    \begin{ZX}[myBold]
                                              &                                               &                             & \\
         \zxZ[BBw]*{x\pi} \dar \ar[to=Gll,C'] &                                               & \zxX[BBw,a=Gll]*{G^\matrixLl x\pi} \dar & \\
         \zxH{} \ar[to=root,<.]               &                                               & \zxMatrix[]'/*{G^\matrixLl} & \\
                                              & [\zxwCol] \zxZ[a=root]{} \ar[B,rr] \ar[ru,.>] &                             & \zxZ[BBw]{Gx\pi} \zxNR
    \end{ZX} \\&\eqZX{Pi}
    \begin{ZX}[myBold]
                                                          &                                               &                             & \\
         \zxX[BBw]*{G^\matrixLl x\pi} \dar \ar[to=Gll,C'] &                                               & \zxZ[BBw,a=Gll]*{x\pi} \dar & \\
         \zxH{} \ar[to=root,<.]                           &                                               & \zxMatrix[]'/*{G^\matrixLl} & \\
                                                          & [\zxwCol] \zxZ[a=root]{} \ar[B,rr] \ar[ru,.>] &                             & \zxZ[BBw]{Gx\pi} \zxNR
    \end{ZX} \eqZXandSubeq{H}{eq:matrixManipulationB}{b}
    \begin{ZX}[myBold]
                                                           &                                                   &                                       & \\
         \zxH{} \dar \ar[to=Gll,C']                        &                                                   & \zxMatrix[a=Gll]'/*{G^\matrixLl} \dar & \\
         \zxZ[BBw]*{G^\matrixLl x\pi} \ar[wcs, to=root,<.] &                                                   & \zxZ[BBw]*{G^\matrixLl x\pi}          & \\
                                                           & [\zxwCol] \zxZ[a=root]{} \ar[B,rr] \ar[ru,wce,.>] &                                       & \zxZ[BBw]{Gx\pi} \zxNR
    \end{ZX} \eqZX{S}
    \begin{ZX}[myBold]
                                                    &                                                   &                                  & \\
         \zxH{} \ar[to=Gll,C'] \ar[wcs, to=root,<.] &                                                   & \zxMatrix[a=Gll]'/*{G^\matrixLl} & \\
                                                    & [\zxwCol] \zxZ[a=root]{} \ar[B,rr] \ar[ru,.>,wce] &                                  & \zxZ[BBw]{(Gx \xor Gx)\pi} \zxNR
    \end{ZX} \eqZX{I} 
    \begin{ZX}[myBold]
      \zxBox{G} \zxNR
    \end{ZX}
  \end{align}
\end{proof}

\begin{lemma}\label{lem:GDestroys}
  For any graph $G$, $
    \begin{ZX}[mbr=1]
      \rar[B] &[\zxwCol] \zxBox{G} \rar[B] & \zxBox{G} \rar[B] &[\zxwCol] 
    \end{ZX}  =
    \begin{ZX}[mbr=1]
      \rar[B] &[\zxWCol] 
    \end{ZX}
  $.
\end{lemma}
\begin{proof}
  This fact is a direct consequence of the property $\gateCZ\gateCZ = \gateI$. More formally:
  \begin{align}
    \begin{ZX}[mbr=1]
      \rar[B] &[\zxwCol] \zxBox{G} \rar[B] & \zxBox{G} \rar[B] &[\zxwCol] 
    \end{ZX} &\eqeqDef{def:graphDiag}
    \begin{ZX}[myBold,mbr=1]
                     & \zxH{} \ar[to=Gll,C'] \ar[to=root,<.] &                                                    & \zxMatrix[a=Gll]'*{G^\matrixLl} & \zxH{} \ar[to=Gll2,C'] \ar[to=root2,<.] &                                                     & \zxMatrix[a=Gll2]'*{G^\matrixLl} & \\
      \zxN{} \ar[rr] &                                       & [\zxwCol] \zxZ[a=root]{} \ar[B,rrr] \ar[to=Gll,.>] &                                 &                                         & [\zxwCol] \zxZ[a=root2]{} \ar[B,rr] \ar[to=Gll2,.>] &                                   & 
    \end{ZX} \eqZX{S}
    \begin{ZX}[myBold,mbr=1]
                     & \zxH{} \ar[to=Gll,C'] \ar[to=root,<.] & \zxH{} \ar[to=Gll2,C'] \ar[to=root]      & \zxMatrix[a=Gll2].*{G^\matrixLl} \ar[to=root2,<.] &                                                                  & \zxMatrix[a=Gll]'*{G^\matrixLl} \\
      \zxN{} \ar[rr] &                                       & [\zxwCol] \zxZ[a=root]{} \ar[B,to=root2] &                                  & [\zxwCol] \zxZ[a=root2]{} \ar[B,rr] \ar[to=Gll,.>]   &  & 
    \end{ZX} \eqZX{H}
    \begin{ZX}[myBold,mbr=2]
                    & [\zxwCol]   & [\zxWCol+1mm]                                                 & \zxMatrix[a=Gll]*{G^\matrixLl}                   &                                                    & [\zxWCol] \\
                    &             &                                                               & \zxMatrix[a=Gll2]*{G^\matrixLl} \ar[to=root2,C-] &                                                    &           \\
      \zxN{} \rar[] & \zxH{} \rar & \zxX[a=root]{} \ar[to=Gll2,C] \ar[to=Gll,C] \ar[B,to=root2,H] &                                                  & [\zxwCol] \zxZ[a=root2]{} \ar[B,rr] \ar[to=Gll,C-] &  & 
    \end{ZX}\\ &\eqSubeq{eq:matrixManipulationA}{c}
    \begin{ZX}[myBold,mbr=3]
                    &                                                                                                                                                                             \\
                    & [\zxwCol]   & [\zxwCol]                                      & \zxMatrix[a=Gll]/*{G^\matrixLl+G^\matrixLl} &                                                    & [\zxWCol] \\
      \zxN{} \rar[] & \zxH{} \rar & \zxX[a=root]{} \ar[to=Gll,C] \ar[B,to=root2,H] &                                             & [\zxwCol] \zxZ[a=root2]{} \ar[B,rr] \ar[to=Gll,C-] &  & 
    \end{ZX} =              
    \begin{ZX}[myBold,mbr=3]
                    &                                                                                                                                                       \\[2mm]
                    & [\zxwCol]   & [\zxwCol]                                      & \zxMatrix[a=Gll]/*{0} &                                                    & [\zxWCol] \\
      \zxN{} \rar[] & \zxH{} \rar & \zxX[a=root]{} \ar[to=Gll,C] \ar[B,to=root2,H] &                       & [\zxwCol] \zxZ[a=root2]{} \ar[B,rr] \ar[to=Gll,C-] &  & 
    \end{ZX} \eqSubeq{eq:defMatrix}{b}
    \begin{ZX}[myBold,mbr=3]
                    &                                                                                                                                  \\
                    & [\zxwCol]   & [\zxwCol]      \zxX[a=Glle]{}      &             & \zxZ[a=Glls]{}                                      & [\zxWCol] \\
      \zxN{} \rar[] & \zxH{} \rar & \zxX[a=root]{} \ar[to=Glle,C] \rar & \zxH{} \rar & [\zxwCol] \zxZ[a=root2]{} \ar[B,rr] \ar[to=Glls,C-] &  & 
    \end{ZX} \eqZXXX{H}{S}{I}
    \begin{ZX}[myBold]
      \zxN{} \rar &[\zxWCol] \zxN{}
    \end{ZX} 
  \end{align}
\end{proof}

\begin{lemma}\label{lem:GCommutes}
  For any graph $G$, $
    \begin{ZX}[mbr=2]
           & [\zxwCol] \dar   & [\zxwCol]         & [\zxwCol] \\
      \rar & \zxZ[BBw]{} \rar & \zxBox{G} \rar[B] & 
    \end{ZX}  =
    \begin{ZX}[mbr=2]
              & [\zxwCol]      & [\zxwCol] \dar & [\zxwCol] \\
      \rar[B] & \zxBox{G} \rar & \zxZ[BBw]{} \rar    & 
    \end{ZX}
  $.
\end{lemma}
\begin{proof}
  This is a direct application of the spider rule on the definition of \zx[myBold]{\zxNL \zxBox{G} \zxNR}.
\end{proof}

We describe now a way to partition a graph $G$ into two subgraphs.
\begin{lemma}\label{lem:graphicalCutG}
  Let $G = (V,E)$. Let $H \subseteq V$ and $M \eqdef V \setminus H$ be a partition of $V$ (for simplicity we assume that we re-order elements in $V$ so that elements in $H$ are smaller than elements in $M$). Then, using notations from \cref{def:notationGraphs}, we have:
  \begin{align}
    \zx[mbr=2,myBold]{
                      & [\zxwCol]                          & [\zxwCol] \\
    \zxBox{G} \rar[B] & \zxDivider{} \ar[ru,<'] \ar[rd,<.] &           \\
                      &                                    & 
    }
    =
    \begin{ZX}[myBold]
      \zxBox{G_H} \rar[B] & [\zxWCol] \zxZ{} \rar[B] \dar[B]                      & [\zxWCol] \zxN{} \\
                          & \zxMatrix./{\Gamma^G_{H \rightarrow M}} \dar[B,H={B}] &                  \\[\zxHRow+1mm]      
      \zxBox{G_M} \rar[B] & \zxZ{} \rar[B]                                        & \zxN{}                    
    \end{ZX}
  \end{align}
\end{lemma}
\begin{proof}
  Intuitively, this lemma only means that in order to create $G$, we can first create the graph states $G_H$ and $G_M$, and apply after $\gateCZ$ gates between elements in $G_H$ and $G_M$. This can be formalized diagrammatically using the decomposition rules for block matrices:
  \begin{align}
    \ket{G} &\eqeqDef{def:graphDiag}
    \begin{ZX}[myBold,mbr=2]
      \zxH{} \ar[to=Gll,C'] \ar[to=root,<.] &                                                   & \zxMatrix[a=Gll]'/*{G^\matrixLl} &                                                & \\
                                            & [\zxwCol] \zxZ[a=root]{} \ar[B,rr] \ar[to=Gll,.>] &                                 & \zxDivider[a=finalDiv]{} \ar[ru,<'] \ar[rd,<.] & \\[\zxwRow]
                                            &                                                   &                                 &                                                & 
    \end{ZX} \eqeq{eq:decompositionG}
    \begin{ZX}[myBold,mbr=2]
      \zxH{} \ar[to=Gll,C'] \ar[to=root,<.] &                                                   & \zxMatrix[a=Gll]'*_{G_H^\matrixLl & 0 \\\Gamma & G_M^\matrixLl}{} &                                                &[\zxwCol] \\
                                            & [\zxwCol] \zxZ[a=root]{} \ar[B,rr] \ar[to=Gll,.>] &                                                             & \zxDivider[a=finalDiv]{} \ar[ru,<'] \ar[rd,<.] &           \\[\zxwRow]
                                            &                                                   &                                                             &                                                &           
    \end{ZX} \eqSubeq{eq:matrixManipulationA}{a}
    \begin{ZX}[myBold,mbr=3]
                                                & \zxDivider[a=div]{} \ar[to=Gll,'>]  \ar[to=Gll2,'>={1-=1,1L=0,2-=0}]                         &                                   &                                                                               \\
      \zxH{} \ar[to=div,<',wce] \ar[to=root,<.] &                                                                                              & \zxMatrix[a=Gll]'*_{G_H^\matrixLl & 0}{} & \zxMatrix[a=Gll2]'*_{\Gamma                    & G_M^\matrixLl}{} &  &[\zxwCol] \\
                                                & [\zxwCol] \zxZ[a=root]{} \ar[B,to=finalDiv] \ar[to=Gll,.>] \ar[to=Gll2,.>={1-=.8,1L=0,2-=0}] &                                   &      & \zxDivider[a=finalDiv]{} \ar[ru,<'] \ar[rd,<.] &                  &  & \\[\zxwRow]
                                                &                                                                                              &                                   &      &                                                & 
    \end{ZX} \\&\eqSubeq{eq:matrixManipulationA}{b}
    \begin{ZX}[myBold,mbr=3]
                                                & \zxDivider[a=div]{} \ar[to=gatherA,'>]  \ar[to=gatherB,'>={1-=1,1L=0,2-=0}]           &                                                                      &                        &                                                                                                                        \\
                                                &                                                                                       & \zxX[a=gatherA]{} \ar[to=Gll] \ar['>={1-=.8,1L=0,2-=0},to=zero]      &                        & \zxX[a=gatherB]{}  \ar[to=Gll2] \ar[to=GM,'>={1-=.8,1L=0,2-=0}]                                                        \\
      \zxH{} \ar[to=div,<',wce] \ar[to=root,<.] &                                                                                       & \zxMatrix[a=Gll]'*{G_H^\matrixLl}                                    & \zxMatrix[a=zero]'*{0} & \zxMatrix[a=Gll2]'*{\Gamma}                                         & \zxMatrix[a=GM]'*{G_M^\matrixLl}                 \\
                                                &                                                                                       & \zxDivider[a=dividA]{} \ar[to=Gll] \ar[to=zero,.>={1-=.8,1L=0,2-=0}] &                        & \zxDivider[a=dividB]{} \ar[to=Gll2] \ar[to=GM,.>={1-=.8,1L=0,2-=0}] &                                                & \\
                                                & [\zxwCol] \zxZ[a=root]{} \ar[B,to=finalDiv] \ar[to=dividA,S,wce] \ar[to=dividB,s,wce] &                                                                      &                        &                                                                     & \zxDivider[a=finalDiv]{} \ar[ru,<'] \ar[rd,<.] & \\[\zxwRow]
                                                &                                                                                       &                                                                      &                        &                                                                     &                                                &  &  & 
    \end{ZX} \eqSubeq{eq:defMatrix}{b}
    \begin{ZX}[myBold,mbr=3]
                                                & \zxDivider[a=div]{} \ar[to=Gll,'>]  \ar[to=gatherB,'>={1-=1,1L=0,2-=0}]               &                                                                      &                &                                                                                                                        \\
                                                &                                                                                       &                                                                      &                & \zxX[a=gatherB]{}  \ar[to=Gll2] \ar[to=GM,'>={1-=.8,1L=0,2-=0}]                                                        \\
      \zxH{} \ar[to=div,<',wce] \ar[to=root,<.] &                                                                                       & \zxMatrix[a=Gll]'*{G_H^\matrixLl}                                    & \zxZ[a=zero]{} & \zxMatrix[a=Gll2]'*{\Gamma}                                         & \zxMatrix[a=GM]'*{G_M^\matrixLl}                 \\
                                                &                                                                                       & \zxDivider[a=dividA]{} \ar[to=Gll] \ar[to=zero,.>={1-=.8,1L=0,2-=0}] &                & \zxDivider[a=dividB]{} \ar[to=Gll2] \ar[to=GM,.>={1-=.8,1L=0,2-=0}] &                                                & \\
                                                & [\zxwCol] \zxZ[a=root]{} \ar[B,to=finalDiv] \ar[to=dividA,S,wce] \ar[to=dividB,s,wce] &                                                                      &                &                                                                     & \zxDivider[a=finalDiv]{} \ar[ru,<'] \ar[rd,<.] & \\[\zxwRow]
                                                &                                                                                       &                                                                      &                &                                                                     &                                                &  &  & 
    \end{ZX} \\&\eqZXXX{Z}{W}{R}
    \begin{ZX}[myBold]
                                             &                                                                    &                                                                                                &  &                               \\
                                             &                                                                    &                                                                                                &  &                               \\
      \zxH{} \ar[to=Gll,C']  \ar[to=root,<.] &                                                                    & \zxMatrix[a=Gll]'*{G_H^\matrixLl}                                                              &  &  &                            \\
                                             &                                                                    &                                                                                                &  &  &                &           \\
                                             & [\zxwCol] \zxZ[a=root]{} \ar[to=end1]  \ar[to=Gll,.>] \ar[to=Gll2,'>] &                                                                                                &  &  & \zxN[a=end1]{} &           \\[\zxwRow]
                                             &                                                                    & \zxMatrix[a=Gll2]'{\Gamma}                                                                                                        \\
                                             &                                                                    & \zxX[a=gatherB]{}  \ar[to=Gll2] \ar[to=GM,'>={1-=.8,1L=0,2-=0}]                                                                   \\
                                             &                                                                    & \zxH{} \ar[to=root2] \ar[to=gatherB]                                                           & \zxMatrix[a=GM]'*{G_M^\matrixLl} \\
                                             & \zxZ[a=zero]{}                                                     & \zxZ[a=root2]{} \ar[to=zero,.>={1-=.8,1L=0,2-=0}] \ar[to=end2] \ar[to=GM,.>={1-=.8,1L=0,2-=0}] &  &  & \zxN[a=end2]{} &  &  & 
    \end{ZX} \eqZXX{S}{H}
    \begin{ZX}[myBold]
      \zxBox{G_H} \rar[B] & [\zxWCol] \zxZ{} \rar[B] \dar[B]                      & [\zxWCol] \zxN{} \\
                          & \zxMatrix./{\Gamma} \dar[B,H={B}] &                  \\[\zxHRow+1mm]      
      \zxBox{G_M} \rar[B] & \zxZ{} \rar[B]                                        & \zxN{}                    
    \end{ZX}
  \end{align}
\end{proof}

\begin{lemma}\label{lem:GaussPivot}
  Let $\Gamma \in \Z_2^{m \times n}$ be a binary matrix. Then, there exists an invertible matrix $U \in \Z_2^{n \times n}$, an integer $r$ (the rank), a matrix $R \in \Z_2^{r \times n-r}$ and an invertible matrix $V \in \Z_2^{m \times m}$ such that:
  \begin{align}
    \Gamma = V
    \begin{bmatrix}
      I_r & R\\
      \mathbf{0} & \mathbf{0}
    \end{bmatrix}
    U
  \end{align}
  where $I_r \in \Z_2^{r \times r}$ is the identity matrix and $\mathbf{0}$ is the zero matrix.
\end{lemma}
\begin{proof}
  This is a direct consequence of the Gaussian elimination algorithm: using Gaussian elimination we can obtain a row echelon form (using only elementary row operations, i.e.\ swapping rows and adding a multiple of one row to another row). Then using column swap, we can reorder the row echelon form matrix to bring the pivots on the diagonal of the matrix. Finally, by subtracting for each row $j$ (starting from the last row) the rows $i$ (for any $i > j$) if $j$'s row contains a one on $i$-th column, we can remove all elements above the diagonal, leading to the identity in the upper right corner. By combining all (invertible) elementary row and column operations into $U$ and $V$, we obtain our final form.
\end{proof}

\begin{lemma}\label{lem:physicallyImplementable}
  Let $U \in \Z_2^{n \times n}$ be an invertible matrix, then $\zx[mbr,arrows={B}]{\rar &[\zxwCol] \zxMatrix{U} \rar &[\zxwCol]}$ is physically implementable without auxiliary qubits using only $\gateCNOT$ and swap operations, and $\zx[mbr,arrows={B}]{\rar &[\zxwCol] \zxMatrix*{U} \rar &[\zxwCol]} = \zx[mbr,arrows={B}]{\rar &[\zxwCol] \zxMatrix{U^{-1}} \rar &[\zxwCol]}$ .
\end{lemma}
\begin{proof}
  The fact that $\zx[mbr,arrows={B}]{\rar &[\zxwCol] \zxMatrix*{U} \rar &[\zxwCol]} = \zx[mbr,arrows={B}]{\rar &[\zxwCol] \zxMatrix{U^{-1}} \rar &[\zxwCol]}$ is a direct consequence of \cite[Lem.~4.8,~4.9]{CHP19_SZXcalculusScalableGraphical}. To see that it can be written as a sequence of \gateCNOT{} and swap operations, we can first realize that since $U$ is invertible, using the Gaussian elimination, we can find elementary row and column operations $E_1, \dots, E_n$ and $E'_1, \dots, \times E'_n$ such that:
  \begin{align}
    E_1  \dots  E_n U E'_1  \dots  E'_n = I
  \end{align}
  The elementary row (resp.\ column) operations can either be:
  \begin{itemize}
  \item Operation 1: Multiply line (resp.\ column) $j$ by a non-null scalar: since the only non-null scalar modulo $2$ is $1$, all operations of this form would be identity so we can remove such operations.
  \item Operation 2: Exchange the $j$-th line (resp.\ column) and the $j'$-th line (resp.\ column).
  \item Operation 3: Add the column $j$ to column $j'$ (in theory, we can multiply first the $j$-th line by an arbitrary factor, but since we are working modulo $2$, the only interesting case is when this scalar is $1$ since when it is equal to $0$ nothing happens)
  \end{itemize}
  We can see that all these operations are self-inverse, so we have:
  \begin{align}
     U = E_n  \dots  E_1 E'_n  \dots  E'_1
  \end{align}
  i.e.\
  \begin{align}
     \zx[myBold,mbr=1]{\zxNL \zxMatrix{U} \zxNR} = \zx[myBold,mbr=1]{\zxNL \zxMatrix{E'_1} \rar & \dots \rar & \zxMatrix{E'_n} \rar & \zxMatrix{E_1} \rar & \dots \rar & \zxMatrix{E_n} \zxNR}
  \end{align}  
  Moreover, any such operation can be implemented using a swap or a \gateCNOT{} gate. If the elementary operation $E$ is Operation 2, i.e.\ a swap, then $E$ can literally be implemented by the same swap operation. This can be seen for instance by realizing that a swap can be realized via a matrix of this form:
  \begin{align}
    E = \begin{bmatrix}
      \mathbf{1^k}        \\
      & 0 &  & 1         \\
      &   & \mathbf{1^l} \\
      & 1 &  & 0         \\
      &   &  &  & \mathbf{1^m}
    \end{bmatrix}\label{eg:diagSwap}
  \end{align}
  where $1^k$ is the diagonal matrix of size $k \times k$ with ones on its diagonal, and $k=j-1$, $l=j'-j-1$. Therefore, using the characterization in \cref{eq:defMatrix}, we have:
  \begin{align}
    \begin{ZX}[myBold,mbr=1]
      \zxNL \zxMatrix{E} \zxNR
    \end{ZX} \eqSubeq{eq:defMatrix}{b}
    \begin{ZX}[mbr=2]
                        & [\zxwCol]                          & \zxZ{} \ar[dr] \ar[dd,3 vdots] &                                      & \zxX{} \ar[rd,'>] \ar[dd,3 vdots] &                          & [\zxwCol] \\
      \zxN{} \rar[B,wc] & \zxDivider{} \ar[ru,<'] \ar[rd,<'] &                                & |[inner sep=.7mm]| E \ar[ru] \ar[rd] &                                   & \zxDivider-{} \rar[B,wc] & \zxN{}    \\
                        &                                    & [\zxwCol] \zxZ{} \ar[ru]       &                                      & \zxX{} \ar[ru,.>]                 &                                      
    \end{ZX} \eqeq{eg:diagSwap}
    \begin{ZX}[mbr=3]
      \rar[Bn'Args={k}{}]       & [\zxWCol] \zxZ[B]{} \rar[B] & [\zxWCol+3mm] \zxX[B]{} \rar[B] & [\zxWCol] \\[.5mm]
      \zxN{} \rar               & \zxZ{} \ar[to=jprime,S]     & \zxX[a=j]{}      \rar           &           \\[.5mm]
      \rar[Bn'Args={l}{pos=.2}] & \zxZ[B]{} \rar[B]           & \zxX[B]{} \rar[B]               &           \\[.5mm]
      \zxN{} \rar               & \zxZ{} \ar[to=j,S]          & \zxX[a=jprime]{} \rar           &           \\[.5mm]
      \rar[Bn.Args={m}{}]       & \zxZ[B]{} \rar[B]           & \zxX[B]{} \rar[B]               & 
    \end{ZX} \eqZX{I}
    \begin{ZX}[mbr=3]
      \rar[Bn'Args={k}{}]       & [\zxWCol+5mm]    \\[.5mm]
      \zxN{} \ar[to=jprime,S]   & \zxN[a=j]{}      \\[.5mm]
      \rar[Bn'Args={l}{pos=.2}] &                  \\[.5mm]
      \zxN{} \ar[to=j,S]        & \zxN[a=jprime]{} \\[.5mm]
      \rar[Bn.Args={m}{}]       & 
    \end{ZX}     
  \end{align}
  Similarly, if $E$ is Operation 3, then we can write $E$ as a matrix of this form (the $1$ might be on the other side of the diagonal if $j' < j$):
  \begin{align}
    E = \begin{bmatrix}
          \mathbf{1^k} \\
           & 1 &  &             \\
           &   & \mathbf{1^l}    \\
           & 1  &  & 1            \\
           &   &  &  & \mathbf{1^m}
       \end{bmatrix}\label{eg:formCNOT}
  \end{align}
  where $k=j-1$, $l=j'-j-1$ which gives, using the characterization in \cref{eq:defMatrix}:
  \begin{align}
    \begin{ZX}[myBold,mbr=1]
      \zxNL \zxMatrix{E} \zxNR
    \end{ZX} \eqSubeq{eq:defMatrix}{b}
    \begin{ZX}[mbr=2]
                        & [\zxwCol]                          & \zxZ{} \ar[dr] \ar[dd,3 vdots] &                                      & \zxX{} \ar[rd,'>] \ar[dd,3 vdots] &                          & [\zxwCol] \\
      \zxN{} \rar[B,wc] & \zxDivider{} \ar[ru,<'] \ar[rd,<'] &                                & |[inner sep=.7mm]| E \ar[ru] \ar[rd] &                                   & \zxDivider-{} \rar[B,wc] & \zxN{}    \\
                        &                                    & [\zxwCol] \zxZ{} \ar[ru]       &                                      & \zxX{} \ar[ru,.>]                 &                                      
    \end{ZX} \eqeq{eg:formCNOT}
    \begin{ZX}[mbr=3]
      \rar[Bn'Args={k}{}]       & [\zxWCol] \zxZ[B]{} \rar[B]  & [\zxWCol+3mm] \zxX[B]{} \rar[B] & [\zxWCol] \\[.5mm]
      \zxN{} \rar               & \zxZ{} \rar \ar[to=jprime,S] & \zxX[a=j]{}      \rar           &           \\[.5mm]
      \rar[Bn'Args={l}{pos=.2}] & \zxZ[B]{} \rar[B]            & \zxX[B]{} \rar[B]               &           \\[.5mm]
      \zxN{} \rar               & \zxZ{} \rar                  & \zxX[a=jprime]{} \rar           &           \\[.5mm]
      \rar[Bn.Args={m}{}]       & \zxZ[B]{} \rar[B]            & \zxX[B]{} \rar[B]               & 
    \end{ZX} \eqZX{I}
    \begin{ZX}[mbr=3]
      \rar[B]     & [\zxWCol] \zxN{} \rar[B]   &[\zxWCol] \\[.5mm]
      \zxN{} \rar & \zxZ{} \rar \ar[to=jprime] &          \\[.5mm]
      \rar[B]     & \zxN{} \rar[B]             &          \\[.5mm]
      \zxN{} \rar & \zxX[a=jprime]{} \rar      &          \\[.5mm]
      \rar[B]     & \zxN{} \rar[B]             &
    \end{ZX}
  \end{align}
  which corresponds exactly to a \gateCNOT{} gate where the $j$-th qubit is the source and the $j'$-th qubit is the target.
\end{proof}

\section{Proofs}

We details in this appendix some proofs of the main paper.

\printProofs

\end{document}